\documentclass[11pt]{article}
\usepackage{amsmath, amsthm, amssymb}
\usepackage{amsbsy}
\usepackage{graphicx}
\usepackage{setspace}
\usepackage{times}

\usepackage{amsfonts}
\usepackage{graphicx}
\usepackage{overpic}
\usepackage{stdmath}

\usepackage[hmargin=1in,vmargin={1in,1in}]{geometry}

\usepackage{natbib}
 \bibpunct[, ]{(}{)}{,}{a}{}{,}%
 \def\bibfont{\small}%
 \def\bibsep{\smallskipamount}%
 \def\bibhang{24pt}%
\usepackage[usenames,dvipsnames]{color}

\newcommand{\newgw}[1]{#1}
\newcommand{\delgw}[1]{}

\newcommand{\delsa}[1]{}

\newcommand{\newrj}[1]{{#1}}
\newcommand{\delrj}[1]{}

\newcommand{\intext}[1]{}

\newenvironment{sproof}[1]{{\noindent\bf #1}}{ \hfill ~\qed}

\newcommand{\mc}[1]{\ensuremath{\mathcal{#1}}}
\newcommand{\mbf}[1]{\ensuremath{\mathbf{#1}}}

\newcommand{\mfr}[1]{\ensuremath{\mathfrak{#1}}}
\newcommand{\norm}[1]{\ensuremath{\left\|#1\right\|}}
\newcommand{\ul}[1]{\ensuremath{\underline{#1}}}
\newcommand{\ol}[1]{\ensuremath{\overline{#1}}}

\newcommand{\addnote}[1]{\textbf{\color{red}#1}}
\newcommand{\delnote}[1]{}

\newcommand{\muVeci}{\ensuremath{\boldsymbol{\mu}_{-i}}}
\newcommand{\muVecm}{\ensuremath{\boldsymbol{\mu}^{(m)}}}
\newcommand{\muVecmone}{\ensuremath{\boldsymbol{\mu}^{(m-1)}}}

\newcommand{\muVec}{\ensuremath{\boldsymbol{\mu}}}

\newcommand{\onep}{\ensuremath{1\mhyphen p}}

\newcommand{\fmi}{\ensuremath{\boldsymbol{f}^{(m)}_{-i, t}}}

\newcommand{\indic}[1]{\ensuremath{\textbf{1}_{#1}}}

\newcommand{\vTilde}{\ensuremath{\tilde{V}}}

\newcommand{\xhat}{\ensuremath{\hat{x}}}

\renewcommand{\v}[1]{\ensuremath{\boldsymbol{#1}}}
\renewcommand{\b}[1]{\ensuremath{\overline{#1}}}

\mathchardef\mhyphen="2D

\newcommand{\ignore}[1]{}

\begin{document}
\date{First Version: October 2010; This Version: November, 2011}

\title{Equilibria of Dynamic Games with Many Players:\\
Existence, Approximation, and Market Structure
\thanks{The authors are
grateful for helpful conversations with Vineet Abhishek, Lanier Benkard, Peter Glynn,
Andrea Goldsmith, Ben Van Roy, and Assaf Zeevi, and seminar
participants at the INFORMS Annual Meeting and the Behavioral and
Quantitative Game Theory Workshop.  This work was 
supported by DARPA under the ITMANET program, and by the National
Science Foundation.}
}

\author{Sachin Adlakha\thanks{S. Adlakha is with the Center for Mathematics of Information,
California Institute of Technology, Pasadena, CA, 91125. {\tt\small adlakha@caltech.edu}}
\and
Ramesh Johari
\thanks{R. Johari is with the Department of Management Science and Engineering,
Stanford University, Stanford, CA, 94305. {\tt\small ramesh.johari@stanford.edu}}
\and
Gabriel Y. Weintraub
\thanks{G. Y. Weintraub is with the Columbia Business School, Columbia University, 
New York, NY, 10027. {\tt\small gweintraub@columbia.edu}}
}


\maketitle



\begin{abstract} 
In this paper we study stochastic dynamic games with many players;
these are a fundamental model for a wide range of economic
applications.  The standard solution concept for such 
games is {\em Markov perfect equilibrium} (MPE), but it is well known
that MPE computation becomes intractable as the number of players
increases. We
instead consider the notion of {\em stationary equilibrium} (SE), where players optimize
assuming the empirical distribution of others' states remains constant at
its long run average.  We make two main contributions. First, we
provide a rigorous justification for using SE.  
 In particular, we provide a parsimonious collection of
exogenous conditions over model primitives that guarantee
existence of SE, and ensure that an appropriate approximation property to MPE
holds, in a general model with possibly unbounded state
spaces. Second, we  draw a significant connection between the validity
of SE, and market structure: under the same conditions that imply SE
exist and approximates MPE well, the market becomes fragmented in the
limit of many firms.  To illustrate this connection, we study in
detail a series of dynamic 
oligopoly examples.  These examples show that our conditions
enforce a form of
``decreasing returns to larger states''; this yields fragmented
industries in the limit. By contrast, violation of these conditions
suggests ``increasing returns to larger states'' and potential market
concentration. In that sense, our work uses a fully dynamic framework
to also contribute  to a longstanding issue in industrial
organization: understanding  the determinants of market structure in
different industries. 

\end{abstract}

\newpage

\section{Introduction}
\label{sec:introduction}

A common framework to study dynamic economic systems of interacting
agents is a \textit{stochastic game}, as pioneered by
\cite{shapley_1953}. In a stochastic game agents' actions directly
affect underlying state variables that influence their payoff.  The
state variables evolve according to a Markov process in discrete time,
and players maximize their infinite horizon expected discounted
payoff. Stochastic games provide a valuable general framework for a
range of economic settings, including {\em dynamic
  oligopolies}\newrj{---i.e., models of competition among firms over time.} In
particular, since the introduction of 
the dynamic oligopoly model of \cite{ericson_1995}, they have been
extensively used to study industry dynamics with heterogeneous firms
in different applied settings (see \cite{pakes_2006} for a survey of
this literature).

The standard solution concept for stochastic games is {\em Markov
  perfect equilibrium} (MPE) \citep{fudenberg_1991}, where a player's
equilibrium strategy depends on the current state of all players. MPE
presents two significant obstacles as an analytical tool, particularly
as the number of players grows large.  First is {\em computability}:
the state space expands in dimension with the number of players, and
thus the ``curse of dimensionality'' kicks in, making computation of
MPE infeasible in many problems of practical
interest. 
Second is {\em plausibility}: as the number of players grows large, it becomes
increasingly difficult to believe that individual players track the
exact behavior of the other agents.

To overcome these difficulties, previous research has considered an
asymptotic regime in which the number of agents is infinite
\citep{jovanovic_1988, hopenhayn_1992}. In this case, individuals take
a simpler view of the world: they postulate that fluctuations in the
empirical distribution of other players' states have ``averaged out''
due to a law of large numbers, and thus they optimize holding the
state distribution of other players fixed. Based on this insight, this
approach considers an equilibrium concept where agents optimize only
with respect to the long run average of the distribution of other
players' states; \cite{hopenhayn_1992} refers to this concept as {\em
  stationary equilibrium} (SE), and we adopt his terminology.  SE are
much simpler to compute and analyze than MPE, making this a useful
approach across a wide range of applications. In particular, SE of
{\em infinite models} have also been extensively used to study
industry dynamics (see, for example, \citealp{LUT07}, \citealp{MEL03},
\citealp{KLEKOR03}, and \citealp{HOPROG93}).

In this paper, we address two significant questions.  
First, {\em under what conditions is it justifiable to use SE as
  a modeling tool}?  We provide theoretical foundations for the use of
SE.  In particular, our main results provide a parsimonious collection of
\textit{exogenous conditions over model primitives} that guarantee
existence of SE, and ensure that an appropriate approximation property
holds. These results provide a rigorous justification for using SE of infinite models  to
study stochastic games with a large but finite number of players.

The second question we address relates to a longstanding topic of research in industrial
organization: {\em when do industries fragment, and when do they
  concentrate?}  \newgw{In a fragmented industry all firms have small market shares,  
 with no single firm or group of firms becoming dominant. By contrast, in a concentrated industry, few participants that hold a notable market share can exert significant market power.}  In dynamic oligopoly models in particular, this is a
challenging question to answer due to the inherent complexity of MPE.
Our second main contribution is to draw a significant connection
between the validity of SE, and market structure: under the same conditions that imply SE exist and
an appropriate approximation property holds, the market becomes fragmented in the limit of many
firms.  In particular, we interpret our conditions over model primitives
as enforcement of a form of ``decreasing returns to larger states'' for an individual firm, that
yields fragmented industries in the limit.  By contrast, as we
discuss, violation of these conditions suggests ``increasing returns
to larger states'' and potential market concentration.

Our
main results are described in detail below.

\begin{enumerate}
\item {\em Theoretical foundations for SE: Existence of SE}.  We
  provide natural conditions over model primitives 
  that guarantee existence of SE over {\em unbounded} state
  spaces. This is distinct from prior work on SE, which typically
  studies models with {\em compact} state spaces.  Crucially, considering
  unbounded state spaces allows us to obtain sharp distinctions
  between increasing and decreasing returns to higher states, and 
  the resulting concentration or fragmentation of
  an industry.

In addition, even though SE of a
  given model may exist over any compact state space, it may fail to
  exist over an unbounded state space. The reason is that agents may
  have incentives to grow unboundedly large and in this case the
  steady-state distribution is not well defined. Hence, a key aspect
  of our conditions is that they ensure the stability of the
  stochastic process that describes each agent's state evolution, and
  that the resulting steady-state distribution is well defined. In
  this way, we guarantee the compactness of an appropriately defined
  ``best-response'' correspondence. Our conditions also ensure the
  continuity and convexity of this correspondence, allowing us to use
  a topological fixed-point approach to prove existence. 

\item {\em Theoretical foundations for SE: Approximating MPE}.  We
  show that the same conditions  
over model primitives that ensure the existence of SE, imply that SE of infinite models
approximate well MPE of models with a finite number of players, as the number of agents increases. An important
condition that is required for this approximation result to hold is
that the distribution of players' states in the SE under consideration
must possess a {\em light-tail}, as originally observed in
\cite{weintraub_2008} for a sequence of finite games, and in
\cite{weintraub_2010} for a limiting infinite model like the one
studied in this paper. In a light-tailed equilibrium, no single agent
is ``dominant;'' without such a condition it is not possible for
agents' to rationally ignore the state fluctuations of their dominant
competitors.

Crucially, the light-tail assumption as used
in~\cite{weintraub_2008} and~\cite{weintraub_2010} is an {\em endogenous}
condition on the equilibrium outcome. A central contribution of this
work is to develop {\em exogenous conditions over model primitives}
that ensure the existence of light-tailed SE. In fact, the conditions
that guarantee compactness in the existence result ensure that {\em all} SE
are light-tailed. Thus approximation need not be verified separately;
verification of our conditions simultaneously guarantees existence of
SE as well as a good approximation to MPE as the number of agents increases. 

\item {\em Market structure in dynamic industries}.  Our results
  provide important insights into market structure in dynamic
  industries. The
  literature on dynamic oligopoly models has largely study individual
  industries in which market outcomes are very sensitive to certain
  model features and parameters \citep{pakes_2006}. In contrast,
  our results provide conditions for which we can predict important features of the
  equilibrium market structure for a broad range of parameters and
  specifications.

  In particular, our conditions over model primitives
  imply that all SE are light-tailed, and therefore, in all SE the
  industry yields a fragmented market structure and no dominant firms
  emerge. Moreover, all these SE are valid approximations to MPE.  
  While these conditions cannot pin-down the equilibrium exactly, they
  guarantee that in all of them the market structure is fragmented. In
  that sense, our work contributes to the ``bounds approach'' in the
  industrial organization literature pioneered by \cite{SUT91}, which
  aims to identify broad structural properties in industries that
  would yield a fragmented or a concentrated market structure. A
  novelty of our analysis compared to previous work is that it is done
  in a {\em fully dynamic} framework.
\end{enumerate}

To illustrate the connection between our theoretical results and
market structure in dynamic industries, we study in detail a collection of
three examples in industrial organization.  For each of these
examples, we demonstrate that our conditions on model primitives that
guarantee existence of light-tailed SE can be interpreted as enforcing
``decreasing returns to higher states.'' Conversely, our analysis of
the examples suggests that when these conditions are violated, the
resulting models exhibit ``increasing returns to higher states,'' and
SE are not expected to provide accurate approximations or may not even
exist. We note that, as emphasized above, unbounded
state spaces are necessary to highlight the difference between increasing and
decreasing returns to higher states.

The first example we discuss is a quality-ladder dynamic oligopoly
model where firms can invest to improve a  firm-specific state; e.g.,
a firm might invest in advertising to 
improve brand awareness, or invest in R\&D to improve product
quality \citep{pakes_1994}.  Firms' single period profits are determined through a
monopolistic competition model. 
Through a limiting construction where the number of
firms and market size both scale to infinity, we use our conditions to
show that light-tailed SE exist and approximate MPE asymptotically if
the single period profit function exhibits diminishing marginal
returns to higher quality.

Next, we discuss a model with
positive spillovers between firms \citep{GIR98}. Here our conditions
impose a form of decreasing returns in the spillover effect that,
together with the decreasing returns to investment condition
introduced in the previous model, ensure SE exist and provide good
approximations to MPE. When the spillover effect is controlled in
this way, the market is more likely to fragment.

Finally, we discuss a dynamic oligopoly that incorporates
``learning-by-doing'', so that firms become more efficient as they
gain experience in the marketplace \citep{FUD83}. In this case, we
find that firms' learning processes must exhibit decreasing returns to
scale to ensure existence of light-tailed SE. These conditions are
consistent with prior observations in the literature that suggest
industries with prominent learning-by-doing effects will tend to
concentrate; our results compactly quantify such intuition.

Indeed, in all these examples, our results validate intuition by
providing quantifiable insight into market structure. Industries with
increasing returns are typically concentrated and dominated by few
firms, so SE would not be good approximations.  By contrast, our
conditions on model primitives delineate a broad range of industries
with decreasing returns that become fragmented in the limit and for
which SE provide accurate approximations.

The remainder of the paper is organized as follows. Section
\ref{sec:related} describes related literature.  Section
\ref{sec:model} introduces our stochastic game model, and there we define
both MPE and SE. We then preview our results and discuss the
motivating examples above in detail in Section \ref{sec:examples}. In
Section \ref{sec:exist}, we develop exogenous conditions over model
primitives that ensure existence of light-tailed SE. In Section
\ref{sec:approx}, we show that under our conditions any light-tailed
SE approximates MPE asymptotically.  Section \ref{sec:discussion}
revisits the examples in light of the theoretical results provided in
the two previous sections. We conclude and discuss future research
directions in Section \ref{sec:conclusions_new}. The appendices
contain all mathematical proofs as well as important
complementary material.

\section{Related Work}
\label{sec:related}

Our work is related to previous literature that studies stationary equilibria or closely-related equilibrium concepts.
SE is sometimes called {\em mean field
equilibrium} because of its relationship to mean field models in
physics, where large systems exhibit macroscopic behavior that is
considerably more tractable than their microscopic description.  (See,
e.g., \cite{blume_1993} and
\cite{morris_2000} for related ideas applied to static games.)  In the
context of stochastic games, SE and related approaches have been
proposed under a variety of monikers across economics and engineering;
see, e.g., studies of anonymous sequential games
\citep{jovanovic_1988, bergin_1995}; dynamic stochastic general equilibrium in
macroeconomic modeling \citep{stokey_1989}; Nash certainty equivalent
control \citep{huang_2006, huang_2007}; mean field games
\citep{lasry_2007}; and dynamic user equilibrium \citep{friesz_1993}.
SE has also been studied in
recent works on information percolation models \citep{duffie_2009},
sensitivity analysis in aggregate games \citep{acemoglu_2009},
coupling of oscillators \citep{yin_2010}, scaling behavior of markets
\citep{bodoh_creed_2010}, and in analysis of stochastic games with
complementarities \citep{adlakha_2010_b}.  

{Prior work has considered existence of equilibrium in
  stochastic games in general, but these are typically established
  only in restricted classes such as zero-sum games and games of
  identical interest; see \cite{mertens_94} for background.}
\cite{DORSAT03} and \cite{ESC08} show existence of MPE for different
classes of stochastic games under appropriate concavity assumptions.
Our work 
is particularly related to \cite{jovanovic_1988} and
\cite{hopenhayn_1992} that consider existence of SE.  The former paper
considers a model similar to ours but restricts attention to compact
sets, while the latter paper is focused on a specific model of
oligopoly competition.  {\cite{adlakha_2010_b} also consider existence of SE;
  they focus on games with strategic complementarities, and establish
  existence using a constructive approach based on lattice theoretic
  methods.} The preceding three papers study a
different setting to ours and do not establish an approximation theorem. 
Several prior papers have considered various
notions of approximation properties for SE {in specific
  settings, either with bounded state spaces \citep{glynn_2004, 
  tembine_2009, bodoh_creed_2010} or with an exogenous compactness assumption
\citep{adlakha_2010_cdc}, or in linear-quadratic payoff models
\citep{huang_2007, adlakha_2008_cdc}.} 

We briefly discuss here relation to our own prior work.  In our
 previous conference papers
 \citep{adlakha_2008_cdc, adlakha_2010_cdc}, we study SE in a less general
 model of stochastic games than this paper.  Though we study existence of SE
 and an appropriate approximation property, we make an {\em endogenous}
 assumption of compactness; in other words, we assume the model is
 such that in searching for SE we can restrict attention to a compact
 set.  As a result, those results do not relate model primitives to
 either validity of SE as an approximation, nor to market structure.  By contrast, in this
 paper, we derive {\em
   exogenous} conditions on model primitives that guarantee compactness, existence
 of SE, and an appropriate approximation property. In addition, as a consequence, we are able to
 apply our results to derive sharp insight into market structure.

Our paper is also closely related to \cite{weintraub_2010}, who
study a class of industry dynamic models. They also show a result
that depends endogenously on SE: if a given SE
satisfies an appropriate light-tail condition, then it approximates MPE well as the
number of firms grows.  
\delrj{Our paper further contributes to providing a rigorous justification for using infinite model SE in the study of industry dynamics.
More specifically,} Our paper provides several important contributions with respect to
\cite{weintraub_2010}.  {First, we consider a more general stochastic game
model that allows us, for example, to study the models with spillovers
and learning-by-doing. On the other hand, we do not consider entry and
exit as they do; we discuss this extension in the conclusions section.
We also consider a stronger approximation property. 
Second, and more importantly, the light-tail condition
used to prove the approximation result in \cite{weintraub_2010} is a
condition over {\em equilibrium outcomes}; by contrast, we provide conditions over
{\em model primitives} that guarantee all SE are light-tailed and
hence approximate MPE asymptotically. As a consequence, these
  conditions also give sharp insight into market structure in our paper.
Finally, we provide a 
novel result pertaining to existence of SE, particularly over
unbounded state spaces.  We close by noting that \cite{weintraub_2010} also
consider an analog of SE called ``oblivious equilibrium'' (OE) in
models with finitely many agents. They study the relation between OE
and SE by analyzing  the hemicontinuity of the OE correspondence at
the point where number of firms becomes infinite.

\section{Preliminaries and Definitions}
\label{sec:model}

In this section we define our general model of a stochastic game, and
proceed to define two equilibrium concepts: Markov perfect equilibrium (MPE) and stationary equilibrium (SE). We conclude by defining the asymptotic Markov equilibrium property, which requires that SE approximates MPE well as the number of players grows large.

\subsection{Stochastic Game Model}

\label{subsec:stochastic-game}
In this section, we describe our stochastic game model. Compared to standard stochastic games in the literature~\citep{shapley_1953}, in
our model, every player has an individual state. Players are coupled through their payoffs and
state transitions. A stochastic game has the following
  elements:

\textit{Time.} The game is played in discrete time. We index time
periods by~$t = 0, 1, 2, \ldots$.

\textit{Players.} There are $m$ players in the game; we use $i$ to
denote a particular player.

\textit{State.} The state of player $i$ at time $t$ is denoted by
$x_{i,t} \in \mc{X}$, where $\mc{X} \subseteq \Z^{d}$ is a
subset of the $d$-dimensional integer
lattice. We use~$\mbox{\boldmath{$x$}}_{ t}$ to denote the state of all players at time $t$ and $\mbox{\boldmath{$x$}}_{-i, t}$ to denote the state of
all players except player~$i$ at time~$t$. For indication of
  how to proceed with compact but not necessarily discrete state spaces, we refer the
reader to the recent independent work of \cite{bodoh_creed_2010}.

\textit{Action.} The action taken by player $i$ at time $t$ is denoted
by $a_{i,t} \in \mc{A}$, where $\mc{A} \subseteq \R^{q}$ is a
subset of the $q$-dimensional Euclidean space. We use
$\mbox{\boldmath{$a$}}_{ t}$ to denote the action of all players at time $t$.

\textit{Transition Probabilities.} The state of a player evolves
in a Markov fashion. Formally, let $h_t = \{ \v{x}_0, \v{a}_0, \ldots,
\v{x}_{t-1}, \v{a}_{t-1}\}$ denote the {\em history} up to time $t$.
Conditional on $h_t$, players' states at time~$t$ are {\em
  independent} of each other. This assumption implies that random shocks are idiosyncratic, ruling out aggregate random shocks that are common to all players. The assumption is important to derive our asymptotic results. Player $i's$ state $x_{i,t}$ at time $t$
depends on the past history $h_t$
only through the state of player $i$
at time $t-1$, $x_{i,t-1}$; the states of other players at time $t-1$,
$\v{x}_{-i,t-1}$; and the action taken by player $i$ at time $t-1$,
$a_{i,t-1}$.  We represent the distribution of the next state as a
transition kernel $\mbf{P}$, where:
\begin{align}
\label{eqn:tm}
\mbf{P}(x_i'\ | \ x_i, a_i, \v{x}_{-i}) = \prob\big(x_{i,t+1} = x_i'\ | \ x_{i,t}
= x_i, a_{i,t} = a_i, \v{x}_{-i,t} = \v{x}_{-i} \big).
\end{align}

\textit{Payoff.} In a given time period, if the state of player $i$ is $x_i$, the state of
other players is $\v{x}_{-i}$, and the action taken by player $i$ is
$a_i$, then the single period payoff to player~$i$ is
$\pi\big(x_{i},a_{i},\v{x}_{-i}\big) \in \R$.

\textit{Discount Factor.} The players discount their future payoff by
a discount factor $0 < \beta < 1$.  Thus a player $i$'s infinite
horizon payoff is given by: $ \sum_{t = 0}^\infty \beta^{t}\pi\big(x_{i,t}, a_{i,t},
\v{x}_{-i,t}\big).$

In a variety of games, coupling between players is independent of the
identity of the players. The notion of \textit{anonymity} captures
scenarios where the interaction between players is via aggregate
information about the state (e.g., see \citealp{jovanovic_1988}). Let
$\fmi(y)$ denote the fraction of players (excluding player $i$) that
have their state as $y$ at time $t$, i.e.:
\begin{align}
\label{eqn:actual-dist}
\fmi(y) = \frac{1}{m-1}\sum_{j \neq i}\indic{\{x_{j,t} = y\}},
\end{align}
where $\indic{\{x_{j,t} = y\}}$ is the indicator function that the
state of player $j$ at time $t$ is $y$.  We refer to $\fmi$ as the
{\em population state} at time $t$ (from player $i$'s point of view).

\begin{definition}[Anonymous Stochastic Game]
\label{def:mean-field-games}
A stochastic game is called an {\em anonymous stochastic game} if
the payoff function $\pi(x_{i,t}, a_{i,t}, \v{x}_{-i,t})$ and transition kernel
$\mbf{P}(x_{i,t}'\  |\ x_{i,t}, a_{i,t}, \v{x}_{-i,t})$ depend on
$\v{x}_{-i,t}$ only through $\fmi$.  In an abuse of notation, we
write $\pi\big(x_{i,t}, a_{i,t}, \fmi\big)$ for the payoff to player~$i$, and $\mbf{P}(x_{i,t}'\ |\ x_{i,t}, a_{i,t}, \fmi)$ for the
transition kernel for player $i$.
\end{definition}

For the remainder of the paper, we focus our attention on anonymous
stochastic games. For ease of notation, we often drop the subscript
$i$ and $t$ and denote a generic transition kernel by $\mbf{P}(\cdot \ | \ x, a, f)$, and a generic payoff function by $\pi(x, a, f)$, where $f$ represents the
population state of players other than the player under
consideration. Anonymity
requires that a firm's
single period payoff and transition kernel depend on the states of
other firms via their empirical distribution over the
state space, and not on their specific identify. The examples we discuss in the next section satisfy this assumption. Second, in an anonymous stochastic game the
functional form of the payoff function is the same, regardless of the
number of players~$m$. In that sense, we often interpret the
  profit function $\pi(x,a,f)$ as representing a limiting regime in
  which the number of agents is infinite. In  Section
  \ref{sec:examples} we discuss how to derive this limiting profit
  function in different applications. Moreover, in Appendix~\ref{se:sequence} we briefly discuss how our results can be extended to include the case where there is a sequence of payoff functions that depends on the number of
agents.

We introduce some additional useful notation. Let $\mfr{F}$ be the
set of all possible population states on $\mc{X}$: 
\begin{align}
\label{eqn:dist-set}
\mfr{F} = \big\{f : \mc{X} \rightarrow [0, 1]\ |\ f(x) \geq 0, \sum_{x
  \in \mc{X}} f(x) = 1 \big\}.
\end{align}

In addition, we let $\mfr{F}^{(m)}$ denote the set of all
population states in $\mfr{F}$ over $m-1$ players, i.e.:
\begin{equation}
\label{eqn:dist-set-mpe}
\mfr{F}^{(m)} = \Big\{ f \in \mfr{F} : \text{ there exists } \v{x}
\in \mc{X}^{m-1} \text{ with } f(y) = \frac{1}{m-1}
\sum_{j}\indic{\{x_{j} = y\}} ,\ \forall y\in\mc{X} \Big\}. \notag
\end{equation}

\subsection{Markov Perfect Equilibrium}
\label{subsec:MPE}

In studying stochastic games, attention is typically focused on 
\textit{Markov} strategies, where the action of a player at each time is a function of only current state of every player~\citep{fudenberg_1991, MASTIR88}. In the
context of anonymous stochastic games, a Markov strategy depends
on the current state of the player as well as the current population
state.  Because a player using such a strategy tracks the evolution of
the other players, we refer to such strategies in our context as {\em cognizant} strategies.

\begin{definition}
Let $\mfr{M}$ be the set of cognizant strategies available to a player. That is, $\mfr{M} = \big\{ \mu\ |\ \mu : \mc{X} \times \mfr{F} \rightarrow \mc{A} \big\}$.
\end{definition}

\ignore{Consider  an $m$-player anonymous stochastic game and let $\mu_i \in \mfr{M}$ denote the cognizant
strategy used by player~$i$, i.e., we have~$a_{i, t} = \mu_i(x_{i,t},\fmi)$.}
Consider an $m$-player anonymous stochastic game. At every time~$t$, player~$i$ chooses an action~$a_{i,t}$ 
that depends on its current state and on the current population state $\fmi \in \mfr{F}^{(m)}$. Letting $\mu_{i} \in \mfr{M}$
denote the cognizant strategy used by player~$i$, we have $a_{i,t} = \mu_{i}(x_{i,t}, \fmi)$. The next state
of player~$i$ is randomly drawn according to the kernel $\mbf{P}$:
\begin{equation}
\label{eq:mpe_dynamics}
x_{i,t+1} \sim \mbf{P}\left( \cdot\ \Big |\ x_{i,t}, \mu_i(x_{i,t},
  \fmi), \fmi \right).
\end{equation}
We let $\muVecm$ denote the strategy
vector where every player has chosen strategy $\mu$. Define~$V^{(m)}\big(x, f\ |\ \mu', \muVecmone \big)$ to be the expected net
present value for a player with initial state~$x$, and with initial
population state~$f \in \mfr{F}^{(m)}$, given that the player follows a strategy~$\mu'$ and
every other player follows the strategy $\mu$. In particular, we
have
\begin{multline}
\label{eqn:mpe-value-func}
V^{(m)}\big(x, f\ |\ \mu', \muVecmone \big) \triangleq\\
\E\left[\sum_{t =0}^{\infty}\beta^{t}\pi\big(x_{i,t}, a_{i,t},\fmi\big) \ \big| \
x_{i,0} = x, f_{-i,0}^{(m)} = f;
\mu_{i} = \mu', \muVeci = \muVecmone \right],
\end{multline}
where $\muVeci$ denotes the strategies employed by every player except $i$.
Note that state sequence~$x_{i,t}$ and population state
sequence~$\fmi$ evolve according to the transition dynamics \eqref{eq:mpe_dynamics}.

We focus our attention on a {\em symmetric Markov perfect equilibrium} (MPE),
where all players use the same cognizant strategy~$\mu$. In an abuse of
notation, we
write~$V^{(m)}\big(x, f\ |\ \muVecm \big)$ to refer to the expected
discounted value as given in equation~\eqref{eqn:mpe-value-func} when
every player follows the same cognizant strategy~$\mu$.

\begin{definition}[Markov Perfect Equilibrium]
\label{def:mpe}
The vector of cognizant strategies~$\muVecm \in \mfr{M}$ is a
\textit{symmetric Markov perfect equilibrium} (MPE) if for all initial
states~$x \in \mc{X}$ and population states $f \in \mfr{F}^{(m)}$ we have
$\sup_{\mu'\in \mfr{M}} V^{(m)}\big(x, f\ |\ \mu', \muVecmone \big) =  V^{(m)}\big(x, f\ |\ \muVecm \big).$
\end{definition}

Thus, a Markov perfect equilibrium is a profile of cognizant strategies that
simultaneously maximize the expected discounted payoff for every
player, given the strategies of other players.\footnote{Under the assumptions we make later in this paper, it can be shown
that for any vector of cognizant strategies of players other than $i$,
an optimal cognizant strategy always exists for player $i$.}
 It is a well known fact that computing a Markov perfect equilibrium
for a stochastic game is computationally challenging in
general~\citep{pakes_2006}. This is because to find an optimal cognizant strategy, each
player needs to
track and forecast the exact evolution of the entire population state. In certain
scenarios, it might be infeasible to exchange or learn this information at
every step because of limited communication capacity between
players or limited cognitive ability. Moreover, even if this is possible, the computation of an optimal cognizant strategy is subject to a curse of dimensionality; the state space $\mfr{F}^{(m)}$  grows too quickly as the number of agents $m$ and/or the number of individual states $\mc{X}$ becomes large. As a consequence, computing Markov perfect equilibrium in practice is only possible in models with few agents and few individual states, severely restricting the set of problems for which this equilibrium concept can be used. In the next subsection, we describe a
scheme for approximating Markov perfect equilibrium that alleviates
these difficulties.

\ignore{The state space $\mfr{F}^{(m)}$  grows too quickly as the number of agents $m$ and/or the number of individual states $\mc{X}$ becomes large. Hence,  computing MPE is only feasible for models with few agents and few individual states, severely restricting the set of problems for which MPE can be used. The concept of stationary equilibrium alleviates
these difficulties.}

\subsection{Stationary Equilibrium}
\label{subsec:OE}

In a game with a large number of players, we might expect that
fluctuations of players' states ``average out'' and hence the actual population state remains roughly constant over time. Because the effect of other players on a single player's payoff
and transition probabilities
is only via the population state, it is intuitive that, as
the number of players increases, a single player has negligible effect
on the outcome of the game. Based on this intuition, related schemes for
approximating MPE have been proposed in different application domains via a solution concept we call
\textit{stationary equilibrium} or SE (see
Sections~\ref{sec:introduction} and \ref{sec:related} for references
  on SE and 
related work).

We consider a limiting model with an infinite number of agents in which a law of large numbers holds exactly. In an SE of this model, each player optimizes its payoff assuming the population state is
fixed at its long-run average. Thus, rather than keep track of the
exact population state, a single player's immediate action depends
only on his own current state. \ignore{Hence,  a single player's 
immediate action depends
only on his own current state.} We call such players {\em
  oblivious}, and refer to their strategies as {\em oblivious strategies}.  (This terminology is due to \citealp{weintraub_2008}.)
Formally, we let $\mfr{M}_O$ denote
the set of (stationary, nonrandomized) oblivious strategies, defined
as follows.
\begin{definition}
Let $\mfr{M}_O$ be the set of oblivious strategies
available to a player. That is, $\mfr{M}_O = \big\{ \mu\ |\ \mu : \mc{X} \to \mc{A}  \big\} $.
\end{definition}
Given a strategy $\mu \in \mfr{M}_O$, an oblivious player~$i$ takes an action $a_{i,t} = \mu(x_{i,t})$ at time~$t$; as before, the next state of the player 
is randomly distributed according to the transition
kernel~$\mbf{P}$:
\begin{equation}
\label{eq:oe_dynamics}
x_{i,t+1} \sim \mbf{P}( \cdot \ | \ x_{i,t}, \mu(x_{i,t}), f) 
\end{equation}
Note that because we are considering a limiting model, the player's
state evolves according to a transition kernel with fixed population state
$f$.  The interpretation is that a single player conjectures the
population state to be $f$; therefore, in determining a player's future
expected payoff stream, it considers a transition kernel where its
own state evolution is affected by the fixed population state~$f$.

\ignore{Note that an oblivious player conjectures the population state to be fixed at~$f$
and hence its state evolves according to a transition kernel with fixed population state
$f$. }

We define the
\textit{oblivious value function}~$\vTilde\big(x\ |\ \mu, f\big)$ to
be the expected net present value for any oblivious player with
initial state~$x$, when the long run average population state
is~$f$, and the player uses an oblivious strategy $\mu$. We have
\begin{align}
\label{eqn:oe-value-func}
\vTilde\big(x\ |\ \mu, f\big) \triangleq
\E\Big[\sum_{t=0}^{\infty}\beta^{t}\pi\big(x_{i,t}, a_{i, t},f \big)\ \Big| \ x_{i,0} = x;\ \mu \Big].
\end{align}
Note that the state sequence~$x_{i,t}$ is determined by the strategy~$\mu$ according to the dynamics~\eqref{eq:oe_dynamics}, where the population state is fixed at $f$. We define the {\em optimal oblivious value function} $\vTilde^*(x \ | \ f)$ as $\vTilde^*(x \ | \ f) = \sup_{\mu \in \mfr{M}_{O}}\vTilde(x \ | \ \mu, f)$. Given a population state $f$, an oblivious player computes an
optimal strategy by maximizing its oblivious value function. Note
that because an oblivious player does not track the evolution of the
population state and its state evolution depends only on the population state~$f$, if an optimal stationary nonrandomized strategy exists, it will only be a function of the player's current
state---i.e., it must be oblivious even if optimizing over cognizant
strategies. We capture this optimization step via
the correspondence $\mc{P}$ defined next. 

\begin{definition}
\label{def:P}
The correspondence $\mc{P}: \mfr{F}\rightarrow \mfr{M}_O$ maps a distribution
$f \in \mfr{F}$ to the set of optimal oblivious strategies for a
player. That is, $\mu \in \mc{P}(f)$ if and only if $\vTilde\big(x\ |\
\mu, f\big) = \vTilde^*(x \ | \ f) $ for all $x$.
\end{definition}
Note that $\mc{P}$ maps a distribution to a
  {\em stationary, nonrandomized} oblivious strategy. This is typically without loss of
  generality, since in most models of interest there always exists
  such an optimal strategy.  We later establish that under our
  assumptions $\mc{P}(f)$ is nonempty.

Now suppose that the population state is $f$, and {\em all} players
are oblivious and play using a stationary strategy $\mu$. Because of
averaging effects, we expect that if the number of agents is large,
then 
the long run population state should
in fact be an invariant 
distribution of the Markov process on $\mc{X}$ that describes the evolution of an individual agent, with transition kernel
\eqref{eq:oe_dynamics}.  We capture this relationship via the
correspondence $\mc{D}$, defined next. 

\begin{definition}
The correspondence $\mc{D} : \mfr{M}_O \times \mfr{F} \rightarrow \mfr{F}$ maps the
oblivious strategy $\mu$ and population state $f$ to the set of invariant distributions
$\mc{D}(\mu, f)$ associated with the dynamics~\eqref{eq:oe_dynamics}.
\end{definition}
Note that the image of the correspondence~$\mc{D}$ is empty if the strategy
does not result in an invariant distribution. We later establish
conditions under which~$\mc{D}(\mu,f)$ is nonempty. In addition, while
we do not impose this restriction {\em a priori}, there are many models of
interest where $\mc{D}$ is actually a {\em function}; that is, for all
$\mu$ and $f$ the Markov process associated with the dynamics
\eqref{eq:oe_dynamics} will be ergodic and admit a unique invariant
distribution.

\ignore{Assume that every agent conjectures the long run population state to be~$f$ and plays
an optimal oblivious strategy~$\mu$. Stationary equilibrium requires that the equilibrium population state $f$ must in fact be an
invariant distribution of the dynamics \eqref{eq:oe_dynamics} under
the strategy $\mu$ and the initially conjectured population state $f$. The consistency of players' conjectures is captured in the following definition.}

We can now define stationary equilibrium.  If every agent
conjectures that $f$ is the long run population state, then
every agent would prefer to play an optimal oblivious strategy $\mu$.
On the other hand, if every agent plays $\mu$ and the population state
is in fact $f$, then we should expect
the long run population state of all players to be an invariant
distribution of \eqref{eq:oe_dynamics}.  Stationary equilibrium requires a consistency
condition: the equilibrium population state $f$ must in fact be an
invariant distribution of the dynamics \eqref{eq:oe_dynamics} under
the strategy $\mu$ and the same population state $f$.

\begin{definition}[Stationary Equilibrium]
\label{def:oe}
An oblivious strategy~$\mu \in \mfr{M}_O$ and a
distribution~$f \in \mfr{F}$ constitute a
stationary equilibrium (SE) if $\mu \in \mc{P}(f)$ and $f \in
\mc{D}(\mu, f)$.
\end{definition}

In the event that the Markov
chain induced by~$\mu$ and~$f$ has multiple invariant distributions,
the agents must all conjecture the population state in equilibrium to
be~$f$.  Further, in the event that there exist multiple optimal
strategies given $f$, the agents must all choose to play~$\mu$.
\ignore{In the event that there exist multiple optimal
strategies given $f$ or that the chain induced by~$\mu$ and~$f$ has multiple invariant distributions,
the agents must all choose to play the same optimal strategy and they must all have the same conjecture about the equilibrium population state.}
In many models of interest (such as the examples presented in Section
\ref{sec:examples}), both $\mc{P}$ and $\mc{D}$ are singletons, so
such problems do not arise. For later reference, we define the
correspondence $\Phi : \mfr{F} \to 
\mfr{F}$ as follows:
\begin{equation}
\label{eq:phi}
\Phi(f) = \mc{D}(\mc{P}(f), f).
\end{equation}
Observe that with this definition, {\em a pair $(\mu, f)$ is an SE
if and only if $f$ is a fixed point of $\Phi$, $f \in \Phi(f)$, such that
$\mu \in \mc{P}(f)$ and $f\in\mc{D}(\mu, f)$}

\subsection{Approximation}
\label{ssec:approx}

A central goal of this paper is to determine conditions under which SE
provides a good approximation to MPE as the number of players
grows large.  Here we formalize the
approximation property of interest, referred to as the asymptotic
Markov equilibrium (AME) property. Intuitively,
this property requires that a stationary equilibrium strategy is
approximately optimal 
even when compared against Markov
strategies, as the number of players grows large.

\begin{definition}[Asymptotic Markov Equilibrium]
A stationary equilibrium $(\mu, f)$
possesses the asymptotic Markov equilibrium (AME) property if for all
states~$x$ and sequences of cognizant strategies $\mu_m \in \mfr{M}$, we have:
\begin{equation}
\label{eqn:ame-def}
\limsup_{m \to \infty}~ V^{(m)}\big(x,f^{(m)}\ |\
\mu_m, \muVec^{(m-1)} \big) -  V^{(m)}\big(x, f^{(m)}\ |\
\muVec^{(m)} \big) \leq 0,
\end{equation}
almost surely, where the initial population state $f^{(m)}$ is derived by sampling
each other player's initial state independently from the probability mass function $f$.
\end{definition}

Note that $V^{(m)}\big(x, f^{(m)} \ | \ \mu', \muVec^{(m-1)}\big)$ is the
\textit{actual} value function of a player as defined in
equation~\eqref{eqn:mpe-value-func}, when the player uses a cognizant strategy
$\mu'$ and every other player plays an oblivious strategy~$\mu$.
Similarly, $ V^{(m)}\big(x, f^{(m)}\ |\ \muVec^{(m)} \big)$ is the actual value function of a player as defined in equation~\eqref{eqn:mpe-value-func} when every player is playing the oblivious strategy~$\mu$.
AME requires that the error when using the SE strategy
approaches zero almost surely with respect to the randomness in the
initial population state. Hence, AME requires that the SE strategy becomes approximately optimal as
the number of agents grows, with respect to population states that
have nonzero probability of occurrence when sampling individual states
according to the invariant distribution.\footnote{As noted earlier, under the assumptions we make an optimal cognizant strategy can be shown to exist, for any vector of cognizant strategies of the opponents.  Therefore the AME property can be equivalently stated as the requirement that for all $x$: $\lim_{m \to \infty} \left( \sup_{\mu_m \in \mfr{M}} V^{(m)}\big(x,f^{(m)}\ |\ \mu_m, \muVec^{(m-1)} \big) -  V^{(m)}\big(x, f^{(m)}\ |\ \muVec^{(m)} \big) \right) = 0,\ \mbox{almost surely.}$}  This definition can be shown to be stronger than the
definition considered by~\cite{weintraub_2008}, where AME is defined
only in expectation with respect to randomness in the initial
population state.

\subsection{Extensions to the Basic Model}
\label{sec:extensions}

We briefly mention two extensions for which all our results follow.
These extensions are often important in applications, but do not
require any significant technical arguments.  See
Appendix~\ref{appendix-extension} for further details.

First, note that players are ex-ante homogeneous in the model considered, in the sense that
they share the same model primitives. This is not a particularly consequential choice, and is made
primarily for notational convenience; indeed, by an appropriate redefinition of state we can model
agent heterogeneity via types.

Second, note that in the game defined here, players are coupled through their states: both the
transition kernel and the payoff depend on the current state of all
players. 
However, in many
models of interest the transition kernel and payoff of a player may depend on both the current state and {\em current actions} of other
players. In particular, the example in Section~\ref{subsec:learning} is a model where players are
coupled through their actions. All the results of this
paper naturally extend to a setting where players may also be coupled through their actions, i.e.,
where the transition kernel and payoff may depend on the current
actions of all players as well. 
In the context of this paper, when players are coupled through actions, for technical simplicity
  we focus on finite action spaces.  In this setting, to ensure
  existence of equilibrium, we assume that players maximize payoffs
  with respect to randomized strategies.  In addition, we briefly
discuss how our 
  results could be extended to include continuous action spaces as well (see Appendix ~\ref{appendix-extension}).

\section{Preview of Results and Motivating Examples}
\label{sec:examples}

As discussed in the Introduction, this paper makes two complementary
contributions.  On one hand, we establish sufficient conditions over
model primitives that
provide justification for use of SE (in particular, that guarantee
existence of SE and the AME property).  On the other hand, we
demonstrate that our conditions encode an economic dichotomy, broadly,
between ``increasing'' and ``decreasing'' returns to higher states;
the latter corresponds to those models where the industry becomes fragmented in the limit and SE is an appropriate
modeling tool.  In this way, our conditions directly provide insight
into market structure.

This section is devoted to introducing examples drawn from industrial
organization that motivate and illustrate our results.
Each example
presents the same basic difficulty: in terms of the parameters of the
model, where does the boundary lie between those markets where
fragmentation might arise, and those markets where concentration might
be expected?  As suggested by the preceding discussion, we use SE as a
tool to inform this market structure question.  In each example,
we discuss how our technical results yield sharp conditions under
which SE exist, the AME property holds, and all SE yield market fragmentation.  We also discuss how
failure of the conditions would suggest market concentration.

To set the stage, we first briefly preview the approach behind our
main technical results (see Section~\ref{sec:exist} and~\ref{sec:approx}).  The mathematical complexity in our analysis arises
due to unbounded state spaces; these are essential if we hope to
identify a boundary between fragmentation and concentration in the
limit of many firms.  Unfortunately, with unbounded state spaces, both
existence of SE and the AME property may become difficult to
establish.  Informally, this is because mass in the population state
may ``escape'' to larger states as the number of firms grows;
alternatively, firms may choose strategies that lead to unbounded
steady state distributions over the state space.

The key condition we require to overcome these hurdles is to ensure
that SE have {\em
  light tails}, i.e., limited mass at larger states (in a sense we
make precise later).  We develop exogenous conditions over model primitives that ensure all SE
population states have {\em light tails}, and we further show that all
light-tailed SE satisfy the AME property (extending a prior result of
\cite{weintraub_2010}).  Light tails ensure that no single dominant
agent emerges in the limit of many firms.  Note that {\em in market
  structure terms, this is exactly market fragmentation}.


Interpretation of our exogenous conditions reveals exactly the
dichotomy introduced above: the conditions enforce a form of
``decreasing returns to higher states'' in the optimization problem
faced by an individual agent, while their failure corresponds roughly
to ``increasing returns.''  Notably, all our results in the examples
are simply applications of the {\em same} theoretical architecture.
As we point out, when the examples below violate the assumptions we
require---in particular, in models that exhibit increasing returns to
higher states---we also expect that SE will {\em not} satisfy the AME
property, and indeed, may not exist.  Thus despite the fact that we
only discuss {\em sufficient} conditions for existence and
approximation in this paper, the examples suggest that perhaps these
sufficient conditions identify a reasonable boundary between those
models that admit analysis via SE, and those that do not.

For the rest of this section, we consider stochastic games
with~$m$ players in which the state of a player takes values on
$\Z_{+}$.

\subsection{Dynamic Oligopoly Models}
\label{subsec:RD}

Dynamic oligopoly models have received significant attention in the
recent industrial organization literature (see \citealp{pakes_2006}
for a survey).  In these models, firms' states correspond to some variable that affects profitability; for example, the state could represent the
firm's product quality, its current productivity level, or its
capacity.  Per period profits are based on a static competition
game, with heterogeneity among firms determined by their respective
quality levels.  Firms take actions to improve their quality; in the
absence of this investment quality degrades over time.  

Such models are extremely broad and capture a wide range of dynamic
phenomena in industrial organization.  In this context, we address the following important question:
under what conditions on the model
primitives do we obtain concentration of the market, and under what
conditions do we obtain fragmentation?  Intuitively, we might expect
that firms need to exhibit decreasing returns to their investments to
obtain fragmentation. 
Our technical results yield a simple condition on model primitives that
formalizes this intuition: we require that the single stage profit
function exhibits decreasing 
returns to firm quality.
In this case SE exist, the AME property
holds, and the market structure is fragmented in the limit.

We now describe our specific model and our result in more detail.

{\em States}.  
For
concreteness, here we consider the quality ladder model of
\cite{pakes_1994}, where the state $x_{i,t} \in \Z_+$ represents the quality of
the product produced by firm $i$ at time $t$.

{\em Actions}.  Investments improve the state variable over time. At each time~$t$,
firm~$i$ invests~$a_{i,t} \in [0, \ol{a}]$ to improve the quality of
its product. The action changes the state of the firm in a stochastic
fashion as described below.

{\em Payoffs}.  We consider a payoff function derived from price
competition under a classic logit demand system.  In such a model, there are
$n$ consumers in the market.  In period $t$,
consumer $j$ receives utility $u_{ijt}$ from consuming the good
produced by firm $i$ given by:
$ u_{ijt} = \theta_1\ln(x_{it}+1) + \theta_2 \ln(Y-p_{it}) +
\nu_{ijt},$ where $\theta_1,\theta_2>0$, $Y$ is the consumer's
income, and $p_{it}$ is the price of the good produced by firm
$i$. Here $\nu_{ijt}$ are i.i.d.~Gumbel random variables that
represent unobserved characteristics for each consumer-good
pair. 

We assume that there are~$m$ firms that set prices in the spot market.
For a constant marginal production cost $c$, there is a unique Nash equilibrium in pure
strategies of the pricing game, denoted $p^*_t$ \citep{CAPNAL91}. For
our limit profit function, we consider an asymptotic regime in which
the market size $n$ and the number of firms $m$ grow to infinity at
the same rate. The limiting profit function corresponds to a logit model
of monopolistic competition \citep{BESETAL90} and is given by {$\pi(x, a, f) =
\frac{\tilde{c}(x+1)^{\theta_1}}{\sum_{y}f(y)(y+1)^{\theta_1}} - da$},
where $\tilde{c}$ is a constant that depends on the limit equilibrium
price, $c$, $\theta_2$, and $Y$. Here the second term is the cost of
investment, where $d > 0$ is the marginal cost per unit investment.

\ignore{ There is also an outside good that provides consumers zero
utility. We assume consumers buy at most one product each period and
that they choose the product that maximizes utility. Under these
assumptions our demand system is a classical logit
model.

We assume that firms set prices in the spot market. If there is a
constant marginal production cost $c$, there is a unique Nash equilibrium in pure
strategies, denoted $p^*_t$ \citep{CAPNAL91}. Expected profits at time
$t$ in an
industry with $m$ firms and $n$ consumers are given by:
\[
\pi_{n}\big(x_{i,t}, \fmi,  m\big) = n \sigma\big(x_{i,t}, m \fmi\big)\big(p^{*}_{i,t} - c\big),
\]
where $\sigma$ represents the market share from the logit model. For
our limit profit function, we consider an asymptotic regime in which
the market size $n$ and the number of firms $m$ grow to infinity at
the same rate. The limit payoff function corresponds to a logit model
of monopolistic competition \citep{BESETAL90} and is given by:
\begin{align*}
\pi(x, a, f) =
\frac{\tilde{c}(x+1)^{\theta_1}}{\sum_{y}f(y)(y+1)^{\theta_1}} - da,
\end{align*}
where $\tilde{c}$ is a constant that depends on the limit equilibrium
price, $c$, $\theta_2$, and $Y$.  Here the first term is the limit of
$\pi_n$ as $m,n \to \infty$ at the same rate, and the second term is the cost of
investment, where $d > 0$ is the marginal cost per unit investment.
}

\textit{Transition dynamics.} We use
dynamics similar to those in~\cite{pakes_1994} that have
been widely used in dynamic oligopoly models. Compared to that paper,
we assume  random shocks are idiosyncratic. At each time period, a
firm's investment of~$a$ is successful with probability $\frac{\alpha
  a}{1 + \alpha a}$ for some $\alpha > 0$, in which case the quality
level of its product increases by one level. The parameter~$\alpha$
represents the effectiveness of the investment. The firm's product
depreciates one quality level with probability $\delta \in (0, 1)$
independently at each time period. Thus a firm's state decreases by one with
probability $\frac{\delta}{1 + \alpha a}$; it increases by one with probability 
$\frac{(1 - \delta)\alpha a}{1 + \alpha a}$ and stays at the same level with probability
$\frac{1 - \delta + \delta \alpha a}{1 + \alpha a}$.

\ignore{

Thus, the transition probabilities are given by:
\begin{align}
\label{eqn:TD-RD}
\mbf{P}(x'\ |\ x, a) = \left\{
\begin{array}{ll}
\frac{(1 - \delta)\alpha a}{1 + \alpha a},&\ \text{if}\ x' = x + 1; \\
&\\
\frac{1 - \delta + \delta \alpha a}{1 + \alpha a},&\ \text{if}\ x' =
x; \\
&\\
\frac{\delta}{1 + \alpha a},&\ \text{if}\ x' = x - 1.
\end{array}
\right.
\end{align}
\vspace{0.1in}

}

{\em Discussion.}  Our main result for this model is the
following proposition.  The proof can be found in Section \ref{sec:discussion_DO}.

\begin{proposition}
\label{prop:DO}
Suppose that $\theta_1 < 1$.  Then there exists an SE for the dynamic
oligopoly model, and all SE possess the AME property.
\end{proposition}

The preceding result has a natural interpretation in terms of
increasing and decreasing returns to higher states. Recall that
$\theta_1$ represents how much consumers value the quality of the
products, and hence if $\theta_1 < 1$, firms have strictly {\em decreasing}
marginal returns in their payoff from increasing their own state. This
implies that as their state grows, firms have less incentives to
invest in improving their own
state and ensures that, in equilibrium, the distribution of firms over the
state space has a light tail and, therefore, the market structure becomes fragmented in the limit. On the
other hand, if
$\theta_1 \geq 1$, then firms have an {\em increasing} marginal gain in
their payoff from increasing their own state. Because the marginal
cost of investment is constant, firms may continue to invest large
amounts to improve their state even at very large states. Thus, a
single firm optimization problem may not even induce a stable Markov
process, 
and
hence an SE may not exist (and the AME
property may fail).

This result matches our intuition for exactly those regimes where SE
work well as approximations to equilibria in finite 
models. In industries with decreasing returns, we expect to see
a fragmented structure in the limit. By contrast, in industries
with increasing returns, market concentration would likely result in
the limit, {i.e., a few firms capture most of the demand in the
market.}  This is precisely where the AME property ceases to hold.

\subsection{Dynamic Oligopoly Models with Positive Spillovers}
\label{subsec:spillovers}

In this section, we extend the previous model to account for
positive {\em spillovers}, or externalities, across firms.  Spillovers are
commonly observed in industry
data and could arise, for example, due to laggard firms imitating
leaders' R\&D activities \citep{GIR98}.  The main difference from the
preceding model is that now transition dynamics are coupled among the
firms: one firm's state is more likely to increase if other firms are
at higher quality levels.

Again, we are led to consider the effect of spillovers on market
structure.  From a technical standpoint, the main complexity is that
firms' best responses may lead to unbounded distributions over the
state space, due to the spillover effect.  Thus, in order to ensure
existence of SE and the AME property, we need a condition that
controls the spillover effect: intuitively, if the spillover effect is
not ``too strong'', then the dynamics will effectively exhibit decreasing
returns.  Our results quantify this sufficient condition.  As before,
in this case, market fragmentation is obtained in the limit of many
firms.
%

To introduce spillovers, we consider a formal model in which the state
space, action space, and payoff are identical to the previous section,
and we continue to use the same notation.  However, we modify the
transition kernel to include spillovers, as described below.


{\em Transition dynamics}. We follow the model of ~\cite{XU06}, in which transition dynamics
depend not only on the action of the firm, but also on the state of
its competitors.
Formally, let~$s_{-i,t}^{(m)}$ be the {\em spillover effect} of
the population state on player $i$ at time $t$, where: $s_{-i,t}^{(m)} = \sum_{y \in \mc{X}}\fmi(y) h_{i,t}(y)$. 
Here $h_{i,t}(y)$ is a weight function that distinguishes the effect
of different states.  For this example, we use $h_{i,t}(y) = \zeta(y){\bf
  1}_{\{y>x_{i,t}\}}$ for some uniformly bounded function $\zeta(y)$. In this
case, a firm is affected with spillovers only from firms that have a
better state than its own, which seems natural. We define the {\em
  effective investment} of 
player~$i$ at time $t$ by: $a_{i,t} + \gamma s_{-i,t}^{(m)} \triangleq e_{i,t}$.
The constant $\gamma$ is a spillover coefficient and it captures the
effect of industry state on the state transition. A higher value of
$\gamma$ means a higher spillover effect. 
With an effective investment of~$e$,
similar to Section \ref{subsec:RD}, a firm's 
state increases by one level with probability $\frac{\alpha
  e}{1+\alpha e}$. Finally, as before, the firm's product depreciates in quality by one
level with probability $\delta \in (0, 1)$ independently at each time
period.\vspace{0.1in}

\ignore{In the last section we argued that in a model without spillovers, our
results apply if and only if $\theta_1 < 1$; with positive spillovers,
the same requirement remains necessary.  However, now the kernel also
depends on the population state $f$, and this introduces an additional
potential complexity: even if $\theta_1 < 1$, the population state of
an agent may grow due to large competitor states.  This may lead to a
scenario where the image of $\Phi$ is unbounded, because firms may
exhibit unbounded growth.  The following proposition provides a simple
condition for existence of SE.}

{\em Discussion.} Since the kernel now depends on the population state $f$ through the spillover effect, even if $\theta_1 < 1$, the population state of
an agent may grow due to large competitor states.  This may lead to a
scenario where the image of $\Phi$ is unbounded, because firms may
exhibit unbounded growth.  The following proposition provides a simple
condition for existence of SE.  The proof can be found in Section
\ref{sec:discussion_DOspill}.  

\begin{proposition}
\label{prop:DOspill}
Suppose that $\theta_1 < 1$, and:
\begin{equation}
\label{eq:spillovers}
\gamma < \frac{\delta}{(1-\delta) \alpha \sup_y \zeta(y)}
\end{equation}
Then there exists an SE for the dynamic
oligopoly model with spillovers, and all SE possess the AME property.
\end{proposition}

Condition \eqref{eq:spillovers} admits a simple interpretation.  This condition
 enforces a form of decreasing returns in the spillover
  effect. If the spillover effect is too large relative to
  depreciation---i.e., if \eqref{eq:spillovers} fails---then the
state of a given firm has positive drift whenever other firms have
large states; and in this case we
expect that, for some $f$, the spillover effect can lead to optimal oblivious
strategies that yield unbounded growth. 
On the other hand, when \eqref{eq:spillovers} holds, then this effect
is controlled, and despite the presence of positive spillovers the
state distribution has a light tail in equilibrium and the industry
becomes fragmented in the limit. 

What happens when the sufficient condition fails?  We present one
informal scenario that suggests market concentration may result.
Observe that it is plausible  that if the condition fails, few firms will have enough incentives to grow large to obtain a competitive advantage. 
Moreover, it is also plausible that a  significant fraction of ``fringe'' firms will remain small to free-ride on the ``dominant'' firms.
In this sense, when our condition is
violated, a dramatically different market structure might be expected. 

\subsection{Learning-By-Doing}
\label{subsec:learning}
Another example that commonly arises in oligopoly settings is {\em
  learning-by-doing}, where firms become more efficient by producing
goods.  In a learning-by-doing model, the state of the firm represents
its experience level; this grows in response to production, and
otherwise depreciates over time.  

In this type of model, it is clear that we require a dichotomy between
``increasing'' and ``decreasing'' returns to experience.  Firms have
to produce even in the absence of learning, simply to earn profits in
each period.  Note that if experience levels continue to grow without
bound, then it will be impossible to ensure SE are light tailed.  We
show this is in fact sufficient: as long as experience begins to
depreciate at sufficiently large states (in a sense we make precise),
then SE exist, the market becomes fragmented in the limit, and the AME property holds.  

We now describe our model; the variant we study is inspired by \cite{FUD83}.

{\em States}.  We let the state~$x_{i,t}$ represent the
cumulative experience level of a firm at time~$t$; this represents the
knowledge accumulated through past production.

{\em Actions}.  The action~$a_{i,t}$ represents the firm's output
(i.e., goods produced) at time~$t$. We consider a model in which firms 
compete on quantity; thus firms are coupled to each other through their {\em actions}.
As discussed in Section \ref{sec:extensions}, such an extension can be accommodated within
our framework by restricting pure actions to lie on a finite subset $S =
\{0,1,\ldots,s_{\max}\}$ of the integers.\footnote{This amounts to
  discretizing the action space of 
  production quantities. In this case, we allow for mixed strategies to ensure existence of SE (see Proposition \ref{prop:finiteconvex}). However, note that in many models of interest, under
  the appropriate concavity assumptions, this is not very restrictive
  as firms will mix between two adjacent pure actions in equilibrium.}

{\em Payoffs}.  At each time period, firms produce goods and compete
in a market with~$n$ consumers. Let $P_{n}(\cdot) \geq 0$ be the
inverse demand function for a market size of~$n$. For state $x$, pure
action $s$, and population state-action profile $f$, we can
write the payoff function as $\pi_n\big(x,s,f,m\big) = s P_{n}\Big(s + (m-1)\sum_{x',s'}s' f(x',s')\Big) - C(x,s),
$ where the argument of $P_{n}$ is the aggregate output (from $m$ firms)
in the market. Note that $f$ is a distribution over state-action
pairs.  Here, $C(x, s)$ denotes the cost of producing quantity~$s$
when the firm's experience level is~$x$. We assume that $C$ is nonnegative,
decreasing, and convex in $x$; is increasing and convex in~$s$; and has decreasing 
differences between~$x$ and~$s$. Consider a limiting case
where both the number of firms $m$ and the market size $n$ become large at the
same rate. 
We assume that there exists a limiting decreasing continuous demand function $P$
such that the limit profit function is given by $\pi\big(x, s, f\big) = s P\left(\sum_{x',s'}s' f(x',s')\right) - C(x, s)$.
Note that the limiting case represents perfect competition as firms
become price takers. 


{\em Transition dynamics}.  A firm's cumulative experience is improved
as it produces more goods since it learns from the production process. On the other
hand, experience capital depreciates over time due to ``organizational
forgetting.'' We assume that a firm's experience evolves independent
of the experience level or the output of other firms in the
market. For concreteness, we assume the
transition dynamics are the same as those described in Section~\ref{subsec:RD}.

{\em Discussion}.  Let $\lim_{x\rightarrow \infty}C(x, s) =
\ul{C}(s)$, that is, $\ul{C}(s)$ is the cost of producing quantity~$s$
for a firm with infinite experience.  Our main result for this
model is the following proposition.  The proof can be found in Section
\ref{sec:discussion_LbD}.  

\begin{proposition}
\label{prop:LbD}
Let $s^*$ be the production level that maximizes $sP(0) -
\ul{C}(s)$.  Suppose that for all sufficiently large $x$ and all
actions $s \in  [0,s^*]$, we have $\sum_{x'} x' \mbf{P}(x' | x, s) <
x$; i.e., the state has negative drift at all such pairs $(x,s)$.
Then there exists an SE for the learning-by-doing model, and all SE
possess the AME property.
\end{proposition}

Observe that $s p- C(x,s)$ is the single period profit to a firm when the market price is $p$, the firm produces quantity $s$, and its experience level is $x$.  Generally speaking, because of learning, firms at low experience levels face strong
incentives to increase their experience, leading them to produce
beyond the single period optimal quantity.  However, for firms at high
experience levels, the choice of optimal quantity is driven primarily
by maximization of single period profit (because $C(x,s)$ is
  decreasing and convex in $x$).  The quantity $s^*$ is an upper bound on
the maximizer of single period profits, so the drift condition in the 
proposition ensures that at high experience levels, firms'
maximization of single period profit does not continue to yield
unbounded growth in the experience level.\footnote{For example, consider
$C(x,s)=s/x$. Then $s^*$ is the largest allowable pure action, hence,
the condition requires that all actions have negative drift for
sufficiently large experience levels. For a
less restrictive case, consider $C(x,s)=s^2/x+s^2/c$. Then,
$s^*=cP(0)/2$, so the condition requires that all actions less than or
equal to $cP(0)/2$ eventually have negative drift.}

The condition requires that the
transition kernel must exhibit sufficiently strong decreasing returns
to scale; 
as long as the possible
productivity gains induced by learning-by-doing are reduced at larger
states, light-tailed SE will exist and the market becomes fragmented in the limit.  However, if
there are not diminishing returns to learning-by-doing, then a firm's experience
level will grow without bound and hence a light-tailed SE may not exist.   This is consistent with prior
observations: an industry for which learning-by-doing is prevalent may
naturally become concentrated over time \citep{DASSTI88}.

\section{Theory: Existence}
\label{sec:exist}

In this section, we study the existence
of {\em light-tailed stationary equilibria}.  We recall that $(\mu, f)$ is a stationary
equilibrium if and only if $f$ is a fixed point of $\Phi(f) =
\mc{D}(\mc{P}(f), f)$, such that $\mu \in \mc{P}(f)$ and
$f\in\mc{D}(\mu, f)$.  Thus our approach is to
find conditions under which the correspondence $\Phi$ has a fixed point; in
particular, we aim to apply Kakutani's fixed point theorem to $\Phi$ to find an
SE.

Kakutani's fixed point theorem requires three essential pieces:
(1) {\em compactness} of the range of~$\Phi$; (2) {\em convexity} of
both the domain of $\Phi$, as well as $\Phi(f)$ for each $f$; and
(3) appropriate {\em continuity} properties of the operator $\Phi$.
It is clear, therefore, that our analysis requires topologies on both
the set of possible strategies and the set of population states.  For
the set of oblivious strategies $\mfr{M}_O$, we use the topology of
pointwise convergence.

For the set of population states, we recall that a key concept in our analysis is that
of ``light-tailed'' population states. To formalize this notion, for the
set of population states we consider a topology 
induced by the {\em $\onep$ norm}.  Given $p > 0$, the~$\onep$-norm of
a function $f : \mc{X} \to \R$ is given by $\norm{f}_{\onep} = \sum_{x \in \mc{X}}\norm{x}_{p}^{p}|f(x)|,$
where $\norm{x}_p$ is the usual~$p$-norm of a vector. Let $\mfr{F}_p$ be the
set of all possible population states on~$\mc{X}$ with finite $\onep$ norm, i.e.,
$\mfr{F}_p = \big\{f \in \mfr{F} : \|f\|_{\onep}<\infty \big\}.$
The requirement $f \in \mfr{F}_p$ imposes a {\em light-tail
condition} over the population state $f$.  The exponent~$p$ controls the weight in the tail of the population
state: distributions with finite $\onep$-norms for larger $p$ have
lighter tails. The condition essentially requires that larger states
must have a small probability of occurrence under $f$. As we discussed
in the context of our examples, light-tailed SE imply that the market
structure becomes fragmented in the limit of a large number of firms.
 
We start with the following restatement of Kakutani's theorem.
\begin{theorem}[Kakutani-Fan-Glicksberg]
\label{th:existence}
Suppose there exists a set $\mfr{C} \subseteq
  \mfr{F}_p$ such that \text{(1)} $\mfr{C}$ is convex and compact (in the $1\mhyphen
  p$ norm), with $\Phi(\mfr{C}) \subset \mfr{C}$; \text{(2)} $\Phi(f)$ is convex and nonempty for every $f \in \mfr{C}$; and
\text{(3)} $\Phi$ has a closed graph  on $\mfr{C}$.\footnote{$\Phi$ has a closed graph if the set $\{ (f,g) : g \in \Phi(f) \}
\subset \mfr{F}_p \times \mfr{F}_p$ is closed (where $\mfr{F}_p$ is endowed
with the $\onep$ norm).} Then there exists a stationary equilibrium $(\mu, f)$ with $f \in \mfr{C}$.
\end{theorem}

In the remainder of this section, we find exogenous conditions on
model primitives to ensure these requirements are met.  We tackle them
in reverse order.  We first show that under
an appropriate continuity condition, 
$\Phi$ has a closed graph.  Next, we study
conditions under which $\Phi(f)$ can be guaranteed to be convex.  Finally, we provide
conditions on model primitives under which there exists a compact, convex set
$\mfr{C}$ with $\Phi(\mfr{F}) \subset \mfr{C}$.  The conditions
we provide suffice to guarantee that $\Phi(f)$
is nonempty for all ${f \in \mfr{F}}$.  Taken together our conditions ensure
existence of SE, as well as an additional stronger characterization:
{\em all} SE are light-tailed, i.e., they have finite $\onep$ norm.
This fact will allow us to show that 
every SE satisfies the AME property in the next section.

\subsection{Closed Graph}
\label{ssec:continuity}

In this section we develop conditions to ensure the model is
appropriately ``continuous.''
Before stating the desired assumption, we introduce one more piece of notation. 
Without loss  of generality, we can view the state Markov process in terms of the
increments from the current state. In particular, we can write $x_{i, t+1} = x_{i,t} +
\xi_{i,t}$, where $\xi_{i,t}$ is a random increment distributed
according to the probability mass function $\mbf{Q}(\cdot \ | \ x, a,
f)$ defined by $\mbf{Q}(z'\ | \ x, a, f) = \mbf{P}(x + z'\ | \  x, a, f).$
Note that $\mbf{Q}(z'\ | \  x, a, f)$ is positive for only those $z'$ such
that $x + z' \in \mc{X}$.
We make the following assumptions over model primitives.
\begin{assumption}[Continuity]
\label{as:continuity}
\begin{enumerate}
\setlength{\itemsep}{1pt}
 \setlength{\parskip}{0pt}
  \setlength{\parsep}{0pt}
\item {\em Compact action set}. The set of feasible actions for a player,
denoted by $\mc{A}$, is compact.
\item {\em Bounded increments}. There exists $M \geq 0$ such that, for all $z$ with $\|z\|_\infty > M$,
  $\mbf{Q}(z \ | \  x, a, f) = 0$, for all $x\in\mc{X}$, $a\in\mc{A}$, and~$f\in\mfr{F}$.
\item {\em Growth rate bound}. There exist constants $K$ and $n \in \Z_{+}$ such
  that $\sup_{a \in \mc{A}, {f \in \mfr{F}}} | \pi(x,a,f)| \leq K
(1 + \norm{x}_\infty)^{n}$ for every $x \in \mc{X}$, where
$\norm{\cdot}_\infty$ is the sup norm.
\item {\em Payoff and kernel continuity}.  {For each fixed $x,x' \in
  \mc{X}$ and $f \in \mfr{F}$, the payoff $\pi(x,a,f)$ and the kernel
  $\mbf{P}(x'\ |\ x, a, f)$ are continuous in $a \in \mc{A}$.

In addition, for each fixed $x, x' \in \mc{X}$, the payoff
$\pi(x,a,f)$ and
the kernel
$\mbf{P}(x' \ | \  x, a, f)$ are jointly continuous in $a\in \mc{A}$
and $f \in \mfr{F}_p$
(where $\mfr{F}_p$ is
  endowed with the $\onep$ norm).}\footnote{Here we view $\mbf{P}(x'\ | \  x, a, f)$ as a real valued
  function of $a$ and $f$, for fixed $x, x'$; note that since we have
  also assumed bounded increments, this notion of
  continuity is equivalent to assuming that $\mbf{P}(\cdot\ | \  x, a, f)$ is
  jointly continuous in $a$ and $f$, for fixed $x$, with respect to the topology of weak convergence
  on distributions over $\mc{X}$.}
\end{enumerate}
\end{assumption}

The assumptions are fairly mild and are satisfied in a variety of
models of interest. For example, all models in
Section~\ref{sec:examples} satisfy it.
The first
assumption is standard. We also place a finite (but possibly large)
bound on how much an agent's state can change in one period 
(Assumption \ref{as:continuity}.2), an assumption that is reasonably weak.   The
polynomial growth rate bound on the payoff is quite weak, and serves to exclude the possibility
of strategies that yield infinite expected discounted payoff.

Finally, Assumption \ref{as:continuity}.4 ensures that the impact of action on payoff and transitions is continuous. It also imposes that the payoff function and transition kernel are ``smooth'' functions of the population state under an appropriate norm.
We note that when $\mc{X}$
is finite, then $\norm{f}_{\onep}$ induces the same topology as the
standard Euclidean norm.  However, when $\mc{X}$ is infinite,
the $\onep$-norm weights larger states higher
than smaller states.
In many applications, other
players at larger states have a greater impact on the payoff; in such
settings, continuity of the payoff in $f$ in the $\onep$-norm naturally
controls for this effect.   Given a
particular model, the exponent $p$ 
should be chosen to ensure continuity of the payoff and transition
kernel.\footnote{See Section \ref{sec:examples} and Section~\ref{sec:discussion} for concrete examples. For example, in subsection \ref{subsec:RD} the payoff function depends on the distribution $f$ via its $\theta_1$ moment so it is natural to endow the set of distributions with the $\onep$ norm with $p=\theta_1$.}
The following proposition establishes that the continuity assumptions
embodied in Assumption \ref{as:continuity} suffice to ensure that
$\Phi$ has a closed graph.  

\begin{proposition}
\label{prop:uhc}
Suppose that Assumption \ref{as:continuity} holds.  Then $\Phi$ has a
closed graph on $\mfr{F}_p$.
\end{proposition}

\subsection{Convexity}
\label{ssec:convexity}

Next, we develop conditions to ensure that $\Phi(f)$ is
convex. We first provide a result for mixed strategies and then a result for pure strategies.

\subsubsection{Mixed Strategies}
\label{ssec:finiteaction}

We start by considering a simple model, where
the action set $\mc{A}$ is the simplex of {\em randomized actions} on
a base set of finite pure actions. This setting is particularly useful when we assume players are coupled through actions (see Section \ref{sec:extensions}). Formally, we have the following
definition.

\begin{definition}
\label{def:finiteaction}
An anonymous stochastic game has a {\em finite action space} if there exists a
finite set~$S$ such that the following three conditions hold:
\begin{enumerate}
\setlength{\itemsep}{1pt}
 \setlength{\parskip}{0pt}
  \setlength{\parsep}{0pt}
\item $\mc{A}$ consists of all probability distributions over $S$:
  $\mc{A} = \{ a \geq 0 : \sum_s a(s) = 1 \}$.
\item $\pi(x,a,f) = \sum_s a(s) \pi(x,s,f)$, where $\pi(x,s,f)$ is the
  payoff evaluated at state $x$, population state $f$, and pure
  action $s$.
\item $\mbf{P}(x'\ |\ x, a, f) = \sum_s a(s) \mbf{P}(x'\ | \ x, s, f)$, where
  $\mbf{P}(x'\ | \ x, s, f)$ is the kernel evaluated at states $x'$
  and $x$, population state $f$, and pure action $s$.
\end{enumerate}
\end{definition}

Essentially, the preceding definition allows inclusion of {\em
  randomized} strategies in our search for SE.
This model inherits  Nash's original approach to establishing
existence of an equilibrium for static games, where randomization
induces convexity on the strategy space.  We
show next that in any game with finite action spaces, the set
$\Phi(f)$ is always convex.

\begin{proposition}
\label{prop:finiteconvex}
Suppose Assumption \ref{as:continuity} holds.  In any anonymous
stochastic game with a finite action space, $\Phi(f)$
is convex for all $f \in \mfr{F}$.
\end{proposition}

The preceding result ensures that if randomization is allowed over a
set of finite actions, then the map $\Phi$ is convex-valued.
We conclude by noting that another simplification is possible
  when working with a finite action space.  In particular, it is
  straightforward to show that {\em if Assumption~\ref{as:continuity}
    holds for the payoff and transition kernel over all pure actions,
    then it also holds for the payoff and transition kernel over all
    mixed actions}; Proposition \ref{prop:uhc} follows
  similarly.  The proof follows in an easy manner using the 
  linearity of the payoff and transition kernel.   This is a valuable
  insight, since in applications it simplifies the complexity of
  checking the model assumptions necessary to guarantee existence of
  an equilibrium.  We discuss a similar point in Section \ref{ssec:compactness}.

\subsubsection{Pure Strategies}

In contrast to the preceding section, many relevant applications
typically require existence of equilibria in pure strategies.  For
such examples, we employ an approach based on the following
proposition.

\begin{proposition}
\label{prop:single_convex}
Suppose that $\mc{P}(f)$ is a singleton for all $f \in \mfr{F}$.  Then
$\Phi(f)$ is convex for all $f \in \mfr{F}$.
\end{proposition}

The proof is straightforward: $\mc{D}(\mu,f)$ is convex-valued for any
fixed $\mu$ and $f$, since the set of invariant distributions for the
kernel defined by $\mu$ and $f$ are identified by a collection of
linear equations.  Thus if $\mc{P}(f)$ is a singleton, then $\Phi(f) =
\mc{D}(\mc{P}(f), f)$ will be convex.

We now provide two different assumptions over model primitives that
guarantee that $\mc{P}(f)$ is a singleton, for all $f \in \mfr{F}$.
The first assumption is a condition introduced by \cite{DORSAT03} and
is described in detail there. The assumption has found wide
application in dynamic oligopoly models.

\begin{assumption}
\label{as:UIC}
\begin{enumerate}
\setlength{\itemsep}{1pt}
 \setlength{\parskip}{0pt}
  \setlength{\parsep}{0pt}
\item The state space is scalar, i.e., $\mc{X} \subseteq \Z_+$, and
the action space $\mc{A}$ is a compact interval of the real
numbers.
\item The payoff $\pi(x,a,f)$ is
strictly decreasing and concave in $a$ for fixed
$x$ and $f$.
\item For all $f \in \mfr{F}$, the transition kernel $\mbf{P}$
is {\em unique investment choice (UIC)  admissible}: there exist
functions $g_1$, $g_2$, $g_3$ such that $\mbf{P}(x'\ |\ x, a, f) =
g_1(x,a,f)g_2(x',x,f)+g_3(x',x,f),\ \forall x',x,a,f$, where
$g_1(x,a,f)$ is strictly increasing and strictly concave in $a$.
\end{enumerate}
\end{assumption}
The preceding conditions ensure that for all population states $f$ and
initial states $x$, and all
continuation value functions, the maximization problem in the right
hand side of Bellman's equation (cf. \eqref{eq:bellman} in the
Appendix) is strictly concave, or that the unique
maximizer is a corner solution.

The previous assumption requires a single-dimensional state
space and action space. Our next assumption imposes a different set of
conditions over the payoff and the transition kernel, and allows for
multi-dimensional state and action spaces.   Before providing our second condition, we
require some additional terminology. Let $S \subset \R^n$.  We say that a function $g : S \to
\R$ is {\em nondecreasing} if
$g(x') \geq g(x)$ whenever $x' \geq x$ (where we write $x' \geq x$ if
$x'$ is at least as large as $x$ in every component).  We say $g$ is
{\em strictly increasing} if the inequality is strict. Let $\mbf{P}_\theta$ be a family of probability distributions on
$\mc{X}$ indexed by $\theta \in S$. Given a nondecreasing function $u :
\mc{X} \to \R$, define $\E_\theta[u] = \sum_{x} u(x) \mbf{P}_\theta(x)$. We say that $\mbf{P}_\theta$ is {\em stochastically nondecreasing} in the
parameter $\theta$, if $\E_\theta[u]$ is nondecreasing in $\theta$ for every nondecreasing function~$u$. Similarly, we say that $\mbf{P}_\theta$ is {\em stochastically concave} in the parameter $\theta$ if $\E_\theta[u]$ is a concave function of $\theta$ for every nondecreasing function~$u$. We say that $\mbf{P}_\theta$ is {\em strictly stochastically
  concave} if, in addition, $\E_\theta[u]$ is strictly concave for every strictly
increasing function $u$. We have the following assumption.

\begin{assumption}
\label{as:convexity}
\begin{enumerate}
\setlength{\itemsep}{1pt}
 \setlength{\parskip}{0pt}
  \setlength{\parsep}{0pt}
\item The action set $\mc{A}$ is convex.
\item The payoff $\pi(x,a,f)$ is
strictly increasing in $x$ for fixed
$a$ and $f$, and the kernel $\mbf{P}(\cdot\ | \ x, a, f)$ is stochastically
nondecreasing in $x$ for fixed $a$ and $f$.
\item The payoff is concave in $a$ for
fixed $x$ and $f$, and the
kernel is stochastically
concave in $a$ for fixed $x$ and $f$, with at least one of the two {\em
  strictly} concave in $a$.
\end{enumerate}
\end{assumption}

The following result shows the preceding conditions on model
primitives ensure the optimal oblivious strategy is unique.

\begin{proposition}
\label{prop:convexity-new}
Suppose Assumption \ref{as:continuity} holds, and that at least one of
Assumptions \ref{as:UIC} or \ref{as:convexity} holds.
Then $\mc{P}(f)$ is a singleton, and thus $\Phi(f)$ is convex for all $f \in \mfr{F}$.
\end{proposition}


\subsection{Compactness}
\label{ssec:compactness}


In this section, we provide conditions under which we can guarantee
the existence of a compact, convex, nonempty set $\mfr{C}$ such that
$\Phi(\mfr{F}) \subset \mfr{C}$.  The
assumptions we make are closely related to those needed to  ensure
that $\Phi(f)$ is
nonempty.  To see the relationship between these results, observe that in Lemma
\ref{lem:bellman} in the Appendix, we show that under Assumption
\ref{as:continuity} an optimal oblivious strategy always exists for
any {$f \in \mfr{F}$}.  Thus to ensure that $\Phi(f)$ is nonempty, it
suffices to show that there exists at least one strategy that possesses an invariant
distribution.  Our approach to demonstrating existence of an invariant distribution
is based on the {\em Foster-Lyapunov criterion}
\cite{meyn_1993}. Intuitively, this criterion 
checks whether the process that describes the evolution of an agent
eventually has ``negative'' drift and in this way controls for the
growth of the agent's state.
This same argument also allows
us to bound the moments of the invariant distribution---precisely what
is needed to find the desired set $\mfr{C}$ that is compact in the
$\onep$ norm.

One simple condition under which $\Phi(f)$ is nonempty is that the
state space is finite; any
Markov chain on a finite state space possesses at least one positive
recurrent class.  In this case the entire set $\mfr{F}$ is compact
in the $\onep$ norm.  Thus we have the following result.

\begin{proposition}
Suppose Assumption \ref{as:continuity} holds, and that the state space $\mc{X}$
is finite. Then $\Phi(f)$ is nonempty for all $f \in \mfr{F}$, and
$\mfr{F}$ is compact in the $\onep$ norm.
\end{proposition}

We now turn our attention to the setting where the state space may be
unbounded; for notational simplicity, in the remainder of the section we assume
$\mc{X} = \Z_+^d$.  In this case, we must make additional assumptions to
control for the agent's growth; these assumptions
ensure the optimal strategy does not allow the state to become
transient, and also allows us to bound moments of the invariant distribution of any
optimal oblivious strategy.

In the sequel we restrict attention to multiplicatively
  separable transition kernels, as defined below.

\begin{definition}
The transition kernel is {\em multiplicatively separable} if there
exist transition kernels $\mbf{P}_1,\ldots,\mbf{P}_d$ such that for
all $x,x' \in \mc{X}, a \in \mc{A}, f \in \mfr{F}$, there holds $ \mbf{P}(x' | x, a, f) = \prod_{\ell = 1}^d \mbf{P}_\ell(x_\ell' |
x, a, f).$ In this case we let $\mbf{Q}_1, \ldots, \mbf{Q}_\ell$ be the
coordinatewise increment transition kernels; i.e., $\mbf{Q}_\ell(z_\ell | x, a, f) =
\mbf{P}_\ell(x_\ell + z_\ell | x, a, f)$, for $z$ such that $x + z \in
\mc{X}$.
\end{definition}
This is  a natural
class of dynamics in models with multidimensional state spaces. We
note that if $\mc{X}$ is one-dimensional, the definition is vacuous.
We introduce the following assumption.
\begin{assumption}
\label{as:compactness}
\begin{enumerate}
\setlength{\itemsep}{1pt}
 \setlength{\parskip}{0pt}
  \setlength{\parsep}{0pt}
\item For all $\Delta \in \Z_+^d$, there holds $\limsup_{\|x\|_\infty \to \infty} \sup_{{a \in \mc{A}, f \in \mfr{F}}} \pi(x + \Delta, a, f) -
\pi(x, a, f) \leq 0.$
\item The transition kernel $\mbf{P}$ is multiplicatively separable.
\item For $\ell =1,\ldots,d$, $\mbf{P}_\ell(\cdot
  | x, a, f)$ is stochastically nondecreasing in $x \in \mc{X}$ and $a
  \in \mc{A}$ for  fixed {$f \in \mfr{F}$}.
\item For $\ell =1,\ldots,d$, and for each $a \in \mc{A}$ and {$f \in
  \mfr{F}$}, $\mbf{Q}_\ell(\cdot | x, a, f)$ is stochastically
  nonincreasing in $x \in \mc{X}$.  Further, for all $x \in
\mc{X}$, $\sup_f \sum_{z_\ell} z_\ell \mbf{Q}_\ell(z_\ell | x, a, f)$
is continuous in $a$.
\item {There exists a compact set $\mc{A}' \subset \mc{A}$, a constant
  $K'$, and a continuous, strictly increasing function $\kappa : \R_+ \to \R_+$
  with $\kappa(0) = 0$, such that:}
\begin{enumerate}
\setlength{\itemsep}{1pt}
 \setlength{\parskip}{0pt}
  \setlength{\parsep}{0pt}
\item For all $x \in \mc{X}$, {$f \in \mfr{F}$}, $a \not\in \mc{A}'$,
  there exists $a' \in \mc{A}'$ with $a' \leq a$, such that $ \pi(x, a', f) - \pi(x, a, f) \geq \kappa(\|a' - a \|_\infty).$
\item For all $\ell$, and all $x'$ such that $x'_\ell \geq K'$, $\sup_{a' \in \mc{A}'} \sup_{{f\in \mfr{F}}} \sum_{z_\ell} z_\ell \mbf{Q}_\ell(z_\ell | x', a', f) < 0.$
\end{enumerate}
\end{enumerate}
\end{assumption}

Some of the previous conditions are natural, while others impose a
type of ``decreasing returns to higher states.''  First, we discuss
the former. Multiplicative separability (Assumption \ref{as:compactness}.2) is
natural. The first part of Assumption \ref{as:compactness}.3 is also
fairly weak. The transition kernel is stochastically nondecreasing in
state in models for which the state is persistent, in the sense that a
larger state today increases the chances of being at a larger state
tomorrow.  The transition kernel is stochastically nondecreasing in
action in models where larger actions take agents to larger states.


{Assumption \ref{as:compactness}.1, \ref{as:compactness}.4, and
  \ref{as:compactness}.5 impose a form of ``decreasing returns to
  higher states'' in the model.}   In particular, Assumption
\ref{as:compactness}.1 ensures the marginal gain in payoff by
increasing one's state becomes nonpositive as the state grows
large. This assumption is used to show that for large enough states
agents effectively become myopic; increasing the state further does
not provide additional gains. Assumption \ref{as:compactness}.5 then
implies that as the state grows large, optimal actions produce
negative drift inducing a ``light-tail'' on any invariant distribution
of the resulting optimal oblivious strategy.  The set $\mc{A}'$ can be
understood as (essentially) the set of actions that maximize the
single period payoff function.  Assumption \ref{as:compactness}.5 is
often natural because in many models of interest increasing the state
beyond a certain point is costly and requires dynamic incentives;
agents will take larger actions that induce positive drift only if
they consider the future benefits of doing so.

The first part of Assumption
\ref{as:compactness}.4 imposes a form of decreasing returns in the transition kernel.
The second part of Assumption
\ref{as:compactness}.4 will hold if, for example, the transition
kernel is coordinatewise stochastically nonincreasing in $f \in
\mfr{F}$ (with respect to the first order stochastic dominance
ordering) and Assumption \ref{as:continuity} holds.  In this case
$\sup_f \sum_{z_\ell} z_\ell \mbf{Q}_\ell(z_\ell | x, a, f) =
\sum_{z_\ell} z_\ell \mbf{Q}_\ell(z_\ell | x, a, \ul{f})$, where
$\ul{f}$ is the distribution that places all its mass at state $0$.

Much of the difficulty in the proof of the result lies in ensuring
that the tail of any invariant distribution obtained from an optimal
oblivious strategy is uniformly light over the image of $\Phi$. The fact
that Assumptions \ref{as:compactness}.1, \ref{as:compactness}.4, and
\ref{as:compactness}.5 are
uniform over $f$ are crucial for this purpose.


Under the preceding assumptions we have the following result.

\begin{proposition}
\label{prop:compactness}
Suppose $\mc{X} = \Z_+^d$, and Assumptions \ref{as:continuity} and
\ref{as:compactness} hold.  Then $\Phi(f)$ is
nonempty for all {$f \in \mfr{F}$}, and there exists
a compact, convex, nonempty set $\mfr{C}$ such that $\Phi(\mfr{F})
\subset \mfr{C}$.
\end{proposition}
Note that the preceding result ensures $\Phi(f) \subset \mfr{C}$ for
{\em all} $f \in \mfr{F}$.  

We conclude this section with a brief comment regarding finite
action spaces, cf. Definition \ref{def:finiteaction}.  The key
observation we make is that {\em if Assumption \ref{as:compactness}
holds with respect to the pure actions---i.e., with $\mc{A}$ replaced
by $S$---then the same result as Proposition \ref{prop:compactness}
holds for mixed actions.}  A nearly identical argument applies to
establish the result.

\subsection{Summary of Results}

 The previous results can be summarized by the following corollary that
 imposes conditions over model primitives to guarantee the existence of
 a light-tailed SE.

 \begin{corollary} \label{co:existence}
 Suppose that (1) Assumption \ref{as:continuity} holds; (2) 
 either the game has a finite action space, or Assumption
   \ref{as:UIC} holds, or Assumption \ref{as:convexity} holds; and (3)
 either the state space $\mc{X}$ is finite, or $\mc{X} = \Z_+^d$
   and Assumption \ref{as:compactness} holds. Then, there exists a SE, and every SE $(\mu,f)$ has $f\in
 \mfr{F}_p$.
 \end{corollary}

 As we have discussed and as one can show in the examples in
 Section \ref{sec:examples}, many models of interest satisfy 
 Assumption  \ref{as:continuity} and  Assumptions \ref{as:UIC} or
 \ref{as:convexity} (or, more generally, some condition that
 guarantees uniqueness of the optimal oblivious strategy); see Section
 \ref{sec:discussion}. Hence, if these models have a 
 finite state space, existence of SE follows immediately. If the state
 space is unbounded, the only condition that remains to be checked to
 guarantee existence of SE is Assumption \ref{as:compactness}. As
 discussed in the examples in Section \ref{sec:examples}, this
 condition imposes a form of ``decreasing returns to higher states.''

 We conclude by emphasizing that under the assumptions of the
 existence result all SE have $f\in\mfr{F}_p$; in other words, all the
 resulting SE have a light-tail. In the context of our examples, as
 previously discussed, this implies that all SE yield a fragmented
 market structure.  In addition, the light-tail property, together
 with Assumption~\ref{as:continuity}, will be used in the next section
 to ensure that the AME property holds.





\section{Theory: Approximation}
\label{sec:approx}

In this section we show that under the assumptions of the preceding
section, any SE $(\mu,f)$ possesses the AME property.  
We emphasize that the AME property is essentially a continuity
property in the population state~$f$.  Under reasonable assumptions,
we show that the time $t$ population state in the system with $m$
players, $f_{-i,t}^{(m)}$, approaches the deterministic
  population state $f$ in an appropriate sense 
almost surely for all $t$ as $m \to \infty$; in particular,
  this type of uniform law of large numbers will hold as long as $f$
has tails that are sufficiently light.  If $f_{-i,t}^{(m)}$
approaches $f$ almost surely, then informally, if the payoff satisfies
an appropriate continuity property in $f$, we should expect the AME
property to hold.  The remainder of the section is devoted to
formalizing this argument.

\begin{theorem}[AME]
Suppose Assumption \ref{as:continuity} holds. Let $(\mu, f)$ be a stationary equilibrium with~$f \in \mfr{F}_p$.  Then the AME property holds
for $(\mu, f)$.
\label{th:AME}
\end{theorem}
Observe that Assumption \ref{as:continuity} is also required for the existence of SE that
satisfy $f \in \mfr{F}_p$.  In this sense, under our assumptions, {\em
  the AME property is a direct consequence of existence.}  This
relationship between existence and the AME property is a significant
insight of our work.

The proof of the AME property exploits the fact that the $\onep$-norm of
$f$ must be finite (since~$f \in \mfr{F}_p$) to show that
$\norm{f_{-i,t}^{(m)} - f}_{\onep} \to 0$ almost surely as $m \to
\infty$; i.e., the population state of other players approaches $f$
almost surely under an appropriate norm.  Continuity of the payoff $\pi$ in $f$, together with
the growth rate bounds in Assumption \ref{as:continuity}, yields the desired
result.

In practice, the light-tail condition---i.e., the requirement that $f \in
\mfr{F}_p$---ensures that an agent's state rarely becomes too large 
under the invariant distribution $f$ associated with the dynamics
\eqref{eq:oe_dynamics}. \cite{weintraub_2010} provide a
similar result in a dynamic industry model with entry and exit. Our
result, on the other hand, is more general in terms of {the definition
of the AME property}, as well as the payoff
functions and transition kernels considered. In particular, we allow
for dependence of the transition kernel on the population
state. {This necessitates a significantly different proof
  technique, since agents' states do not evolve independently in general.}  
We note that the light-tail condition
is consequential, as it is possible to construct examples for which
stationary equilibria exist, but $f\notin \mfr{F}_p$ and the AME
property does not hold \citep{weintraub_2010}.

We conclude by noting that in many models of interest it is more reasonable to assume that the
payoff function explicitly depends on the number of agents. To study
these environments, we consider a sequence of payoff
functions indexed by the number of agents, $\pi_m(x,a,f)$. Here, the
profit function~$\pi$ is a {\em limit}: $\lim_{m\to\infty}
\pi_m(x,a,f)=\pi(x,a,f)$.  (See Section \ref{sec:examples} for concrete
examples.) In this case, if the number of players is $m$, the actual expected net present value is defined with $\pi_m$; hence, the payoff function in the AME property depends on $m$. In Appendix~\ref{se:sequence} we show that under a strengthening of Assumption \ref{as:continuity}, Theorem  \ref{th:AME} can be generalized to this setting.

\section{Examples Revisited}
\label{sec:discussion}

In this section we revisit each of the examples presented in
 Section \ref{sec:examples} and show that {\em all} the
  propositions for these examples are consequences of
Corollary \ref{co:existence} and Theorem~\ref{th:AME}.  This
establishes the key connection in the paper between existence of SE
and the AME property on one hand, and the impact of model primitives
on market structure on the other hand. In particular, our conditions
over model primitives imply that all SE are light-tailed, and
therefore, in all SE the industry yields a fragmented market
structure, and the AME property is satisfied. 

Formally, recall that the conditions required to establish the main
 results of this paper are Assumption~\ref{as:continuity} (used to
 ensure continuity properties); Assumption~\ref{as:UIC} and/or
 \ref{as:convexity} (used to ensure convexity of the image of $\Phi$);
 and Assumption~\ref{as:compactness} (used to ensure the existence of a
 compact subset $\mfr{C} \subset \mfr{F}$ such that  $\Phi(\mfr{C})
 \subset \mfr{C}$).  Of these properties, continuity and convexity are typically
 straightforward to guarantee in each of the models we consider below.
 Thus we primarily focus on the role of Assumption~\ref{as:compactness}.

\subsection{Dynamic Oligopoly Models}
\label{sec:discussion_DO}
In this section, we provide the proof of Proposition~\ref{prop:DO}.  Note that the payoff function depends on the
 distribution~$f$ via its $\theta_1$ moment, and hence we endow the set of
 distributions with the topology induced by the~$\onep$ norm
 with~$p = \theta_1$.  Since the payoff is continuous and nonincreasing in the $\theta_1$ moment of $f$, and
 the transition kernel is independent of $f$, it is straightforward to
 check that Assumption \ref{as:continuity} holds.  In addition, \cite{DORSAT03}
 show that the transition kernel of this model satisfies Assumption
 \ref{as:UIC} (it can also be shown that Assumption~\ref{as:convexity}
 is satisfied).

 Thus the desired result is reduced determining
 whether Assumption \ref{as:compactness} holds.  It is straightforward
 to check that Assumptions \ref{as:compactness}.2-4 hold; we omit the
 details.  Assumption \ref{as:compactness}.5 holds because positive
 drift is costly, as the kernel defined above exhibits depreciation; in
 particular, it suffices to set $\mc{A}' = \{0\}$.  Thus the central
 condition to check in this model is Assumption \ref{as:compactness}.1.
 This assumption holds {\em if and only if $\theta_1 
   < 1$}: in this case, $\sup_{a, f} \pi(x+\Delta, a, f) - \pi(x, a,
 f) \to 0$ as  $x \to \infty$ for all $\Delta > 0$.  Using Corollary
 \ref{co:existence} and Theorem \ref{th:AME}, the result follows.
 
Thus, existence of SE and the AME property are closely tied to the parameter
$\theta$ which represents how much consumers value the quality of the product.
For $\theta < 1$, the firms have decreasing marginal returns in their payoff 
from increasing their state. This ensures that the Markov process associated 
with a single firm optimization process is stable which in turn ensures that the
range of $\Phi$ is compact. As discussed earlier, this condition
leads to a natural separation between industries where we expect to
see a fragmented market structure and the industries where market concentration
is likely to result in the limit.

\subsection{Dynamic Oligopoly Models with Positive Spillovers}
\label{sec:discussion_DOspill}

In this section we provide the proof of Proposition~\ref{prop:DOspill}. 
Assumption \ref{as:continuity} and \ref{as:UIC} follow as in the
 preceding result; the proof is omitted.  Again we focus on Assumption
 \ref{as:compactness}.  Assumption \ref{as:compactness}.1,
 \ref{as:compactness}.2, \ref{as:compactness}.3, and the first part of
 \ref{as:compactness}.4 hold as before; we omit the details.
 The key assumptions that we need to verify are thus the
 second part of Assumption \ref{as:compactness}.4, and Assumption
 \ref{as:compactness}.5.

 Observe that the maximum possible value of the effective investment
 when a firm takes action~$a$ is $e_{\max}(a) \triangleq a + \gamma \sup_y
 \zeta(y)$.  A straightforward calculation yields:
 \begin{align}
 \sup_f \sum_z z \mbf{Q}(z | x, a,f ) &= (1-
 \delta)\left(\frac{\alpha e_{\max}(a)}{1+\alpha e_{\max}(a)}\right) - \delta
   \left( \frac{1}{1 + \alpha e_{\max}(a)}\right)\\
 & = \frac{\alpha e_{\max}(a)}{1 +
     \alpha e_{\max}(a)} - \delta. \label{eq:spilloversdrift}
 \end{align}
 It
 follows from the definition of the transition kernel that the second
 part of Assumption \ref{as:compactness}.4
 holds.  
 In order for
 Assumption~\ref{as:compactness}.5 to hold with $\mc{A}' = \{0\}$, it
 follows that we need:
 \[
 \gamma < \frac{\delta}{(1-\delta) \alpha \sup_y \zeta(y)}
   \]
 Using Corollary
 \ref{co:existence} and Theorem \ref{th:AME}, we conclude that the
 result of the proposition follows if \eqref{eq:spillovers} holds.

For industries with spillovers, the compactness assumptions requires that the spillover 
effect is not too large relative to depreciation. This along with decreasing marginal returns
in the payoff ensures that the firms do not have unbounded growth in
their state.  As a result, the market structure becomes fragmented in
the limit of a large number of firms.


\subsection{Learning-By-Doing}
\label{sec:discussion_LbD}

In this section, we provide the proof of Proposition~\ref{prop:LbD}.
Since $P$ is decreasing and $C(x,s)$ is decreasing
 in~$x$, Assumption~\ref{as:continuity} follows in a straightforward
 manner in this model, as long as $P$ is continuous.  Since this is a
 model with finite action spaces, the result of Proposition
 \ref{prop:finiteconvex} also applies.  Thus, as before, the proof is
 reduced to determining 
 whether Assumption \ref{as:compactness} holds for the given model.  As
 in the preceding examples, it is straightforward to check that
 Assumptions \ref{as:compactness}.2, \ref{as:compactness}.3, and
 \ref{as:compactness}.4 hold.  Note that $\pi(x+\Delta, s, f) - \pi(x,s,f)
 = C(x,s) - C(x+\Delta,s)$ for all $x,s,f$ and $\Delta \geq 0$, and the
 action space is finite.  Thus Assumption \ref{as:compactness}.1
 follows since $C(x,s)$ is nonnegative, decreasing, and convex in $x$.

 Therefore, our focus turns to Assumption
 \ref{as:compactness}.5.   
  Using standard supermodularity arguments, it is simple to
 check that any $s$ that maximizes $\pi(x,s,f)$ for some $x,f$ is
 contained in the interval $[0,s^*]$.  In particular, then, suppose
 that for all sufficiently large $x$ and all actions $s \in
   [0,s^*]$, we have $\sum_z z \mbf{Q}(z | x, s) < 0$.  Then
 Assumption \ref{as:compactness}.5 holds, so using Corollary
 \ref{co:existence} and Theorem \ref{th:AME} the result follows.

 In learning-by-doing models, compactness of the image of~$\Phi$ is ensured by 
 requiring that the transition kernel exhibits decreasing returns to
 higher states.  In other words, if the productivity gains induced by
 learning-by-doing are reduced 
 at larger states, light-tailed SE will exist and the AME property will hold.
 As discussed earlier this is consistent with the observation that a
 very strong learning-by-doing effect (that persists even at large
 scale) will likely lead to market concentration.

\section{Conclusions}
\label{sec:conclusions_new}

This paper considered stationary equilibrium in dynamic games with
many players.  Our main results provide a parsimonious set of
assumptions on the model primitives which ensure that a stationary
equilibrium exists in a large variety of games.  We also showed
that the same set of assumptions ensure that SE yield fragmented
market structures, and is a good 
approximation to MPE in large finite games. Through a set of examples,
we illustrate that our conditions on model primitives can be naturally
interpreted as enforcing ``decreasing returns to higher states.''  

We conclude by noting several extensions that can be developed for the models described here.

\begin{enumerate}

\item {\em Entry and exit.} A natural extension, particularly relevant
  for dynamic oligopoly models, would be to consider a scenario where
  agents (i.e., firms) make entry and exit decisions endogenously in
  equilibrium. We conjecture that under some mild additional
  assumptions our results would extend to this setting.

\item {\em Connections between SE and oblivious equilibrium in finite
    models.} In some contexts, particularly in empirical settings, it
  may be more appropriate to work over a model with a finite number of
  agents. In these cases, as discussed in Section~\ref{se:sequence},
  it is possible to define an ``oblivious equilibrium'' for finite
  models \citep{weintraub_2008}. We conjecture that under some additional technical
  conditions over the model primitives we can prove that a sequence of
  OE satisfies the AME property.

\item {\em Nonstationary equilibrium.} Our focus was on SE because it
  is of practical interest and has received significant attention in
  the literature. We conjecture, however, that our results can be
  extended to nonstationary versions of an equilibrium concept based
  on averaging effects that could be used to approximate transitional
  short-run dynamics as oppose to long-run behavior.
\end{enumerate} 

We leave these directions for future research.

\bibliographystyle{ormsv080}
\bibliography{OE-TR}

\appendix

\section{Extensions to the Basic Model}
\label{appendix-extension}


\subsection{Heterogeneous Players}

In this section, we study anonymous stochastic games with ex-ante heterogeneous
players. To represent this heterogeneity, at the beginning of the
game, a player is assigned a {\em type} (denoted by $\theta$)
that stays fixed for the entire duration of the game. For simplicity,
we assume that the players'
types are randomly and independently drawn out of a {\em finite set}
$\Theta$ with a common prior distribution~$\Gamma$.   Let $\mbf{P}(
\cdot | x, a, f; \theta)$ and $\pi(x, a, f; \theta)$
denote the transition kernel and payoff of a type $\theta$ player.

To analyze a stochastic game with heterogeneous players, we define a
new state as follows. Let  $\xhat = (x, \theta)$ be
an extended state; if a player's extended state is $\xhat$, we interpret it to
mean that the player is in state $x$ and has a type $\theta$.  We let
$\mc{\hat{X}} = \mc{X} \times \Theta$ denote the expanded state
space.  Let $\hat{f}$ denote a population state over the expanded
state space, i.e., $\hat{f}$ is a distribution over $\mc{\hat{X}}$.
Given $\hat{f}$, we define $F(\hat{f}) \in \mfr{F}$ by:
\[ F(\hat{f})(x) = \sum_{\theta} \hat{f}(x, \theta). \]
We have the following two definitions:
\begin{align*}
 \hat{\pi}(\hat{x}, a, \hat{f}) &= \pi(x, a, F(\hat{f}); \theta);\\
\hat{\mbf{P}}(\hat{x}' | \hat{x}, a, \hat{f}) &= \left\{
\begin{array}{ll}
0,&\ \text{if}\ \theta' \neq \theta;\\
\mbf{P}(x' | x, a, F(\hat{f}); \theta),&\ \text{if}\ \theta' = \theta.
\end{array}
\right.
\end{align*}
These definitions simply map the payoff and transition kernel with
respect to the extended state back to the payoff and transition kernel
in the original game.  Now observe that in the new game defined in
this way, it can be verified that if the original game satisfied
Assumptions \ref{as:continuity}, \ref{as:UIC} or \ref{as:convexity}, and~\ref{as:compactness} for each $\theta$, then the extended game
satisfies the same assumptions as well.  Thus all our preceding
results apply even in games with heterogeneous players. Because
strategies are a function of the extended state,  in this case
players of different types will use different strategies.

\subsection{Coupling Through Actions} \label{se:couple actions}

In the main development of this paper, we considered anonymous
stochastic games where players are coupled to each other via the
population state as defined in equation~\eqref{eqn:actual-dist}; note,
in particular, that the population state gives the fraction of players
at each state.  As discussed in the Introduction, however, in many
models of interest the transition kernel and payoff of a player may
depend on both the current state and {\em current actions} of other
players.  In particular, the example in Section
\ref{subsec:learning} is a model where players are
coupled through their actions.

To formally model such a scenario, we consider an~$m$ player
stochastic game being played in discrete time over the infinite
horizon, where again the payoff and transition kernel of a player are
denoted {by} $\pi(x, a,f)$ and $\mbf{P}(\cdot | x, a, f)$
respectively.\footnote{For the purposes of this subsection we assume
  players are homogeneous.}  However, we now assume that $f$ is a
distribution over both states and actions.  We refer to $f$ as the
{\em population state-action profile} (to distinguish it from just the
population state, which is the marginal distribution of $f$ over
$\mc{X}$).  For simplicity, since the prior development in this paper
assumes state spaces are discrete, for the purposes of this subsection
we restrict attention to a game with a {\em finite} action space $S
\subset \Z^k$, cf. Definition \ref{def:finiteaction}; in particular,
we assume that players maximize payoffs with respect to randomized
strategies over~$S$.  Thus the population state-action profile is a
distribution over $\mc{X} \times S$.

We again let $x_{i,t} \in \mc{X}$ be the state of player~$i$ at
time~$t$, where $\mc{X} \subseteq \Z^{d}$. We let~$s_{i,t} \in S$ be
the (pure) action taken by player~$i$ at time~$t$.  Let $\fmi$ denote
the empirical population state-action profile at time $t$ in an
$m$-player game; in other
words, $f_{i,t}^{(m)}(x,s)$ is the fraction of players other than~$i$
at state $x$ who play $s$ at time $t$.  With these definitions,
$x_{i,t}$ evolves according to the transition kernel $\mbf{P}$ as before, i.e.,
$x_{i,t+1} \sim \mbf{P}(\cdot | x_{i,t}, a_{i,t}, f_{-i,t}^{(m)})$.

A player acts to maximize their expected discounted payoff, as before.
Note that a potential challenge here is that a player's time $t$
payoff and transition kernel depend on the actions of his competitors, which are chosen
{\em simultaneously} with his own action.  Thus to evaluate the time~$t$ expected  payoffs and transition kernel, a player must take an expectation with respect to
the randomized strategies employed by his competitors.

Our first step is to extend the appropriate assumptions to this game
model.  Let $\mfr{F}$ now denote the set of all distributions over $\mc{X}
\times S$, and let $\mfr{F}_p$ denote the set of all distributions in
$\mfr{F}$ with finite $\onep$-norm as before.  Assumptions
\ref{as:continuity} and \ref{as:compactness} thus extend naturally to
games with coupling through actions, with these new interpretations of
$\mfr{F}$ and $\mfr{F}_p$.

The AME property continues to hold for games with coupling through
actions.  Recall that in the proof of Theorem \ref{th:AME}, we
establish that if $(\mu,f)$ is a stationary equilibrium, then
$\|f_{-i,t}^{(m)} - f\|_{\onep} \to 0$ almost surely for all $t$, if
players' initial states are sampled independently from $f$, all
players other than $i$ follow strategy $\mu$, and player $i$ follows
any strategy.  (See Lemma \ref{lem:convergence-of-f} in the Appendix.)
In a game with coupling through actions, $f_{-i,t}^{(m)}$ also tracks
the empirical distribution of players' actions.  However, since all
players other than~$i$ use the same oblivious strategy $\mu$, and
since the base action space $S$ is finite, it is straightforward to
extend the argument of Lemma~\ref{lem:convergence-of-f} to the current
setting. The remainder of the existing proof of Theorem~\ref{th:AME}
carries over essentially unchanged under Assumption~\ref{as:continuity}; for brevity we omit the details.

Next, recall that to prove existence of a stationary equilibrium, we consider two maps: $\mc{P}(f)$ (which identifies
the set of optimal oblivious strategies given $f$), and $\mc{D}(\mu,f
)$ (which identifies the set of invariant distributions of the Markov
process induced by $\mu$ and $f$).  The analysis of $\mc{P}(f)$
proceeds exactly as before (but with randomized strategies, as
  discussed in Section \ref{ssec:finiteaction}).  However, in a game
with coupling through 
actions, we redefine $\mc{D}(\mu, f)$ to be the
set of invariant distributions over $\mc{X} \times S$ induced by $\mu$
and $f$.  In other words, $f' \in \mc{D}(\mu, f)$ if and only if there
exists a distribution $g$ over $\mc{X}$ such that the following two conditions hold:
\begin{align*}
g(x') &= \sum_x g(x) \mbf{P}(x' | x, \mu(x),f),\ \text{for all}\ x';\\
f'(x, s)  &= g(x) \cdot \mu(x)(s),\ \text{for all}\ x,s.
\end{align*}
Note that here $\mu(x)(s)$ is the probability assigned to
pure action~$s$ by the randomized strategy $\mu$ at state~$x$.
The first equation requires that $g$ is an invariant distribution of
the state Markov process induced by~$\mu$ and~$f$ (recall Definition~\ref{def:finiteaction} of the transition kernel with mixed
actions).  The second 
equation requires $f'$ to be derived from $g$ in the natural way, via $\mu$.
As before, we let $\Phi(f) = \mc{D}(\mc{P}(f), f)$.

It is now straightforward to show that if Assumption~\ref{as:continuity} holds, then the result of Proposition~\ref{prop:uhc} holds, i.e.,~$\Phi$ has a closed graph.  Further, if Assumptions~\ref{as:continuity} and~\ref{as:compactness} hold, then
the result of Proposition~\ref{prop:compactness} holds as well.  From this and the result in Proposition~\ref{prop:finiteconvex} we conclude that under those assumptions, a stationary equilibrium
exists, and all SE are light-tailed (i.e., have finite $\onep$ norm).  The arguments involved are analogous to the existing proofs,
and we omit the details.

We conclude by commenting on the restriction that the action space
must be finite.  From a computational standpoint this is not very
restrictive, since in many applications discretization is required or
can be used efficiently.  From a theoretical standpoint, we can
analyze games with general compact Euclidean action spaces using
techniques similar to this paper, at the expense of additional
measure-theoretic complexity, since now the population state-action
profile is a measure over a continuous extended state space.

 \section{Approximation: Sequence of Payoff Functions}
 \label{se:sequence}

 In many models of interest it is more reasonable to assume that the
 payoff function explicitly depends on the number of agents. To study
 these environments, in this section we consider a sequence of payoff
 functions indexed by the number of agents, $\pi_m(x,a,f)$. Here, the
 profit function~$\pi$ is a {\em limit}: $\lim_{m\to\infty}
 \pi_m(x,a,f)=\pi(x,a,f)$. See Section \ref{sec:examples} for concrete
 examples.

 In this case, the actual expected net present value of a player using a cognizant strategy $\mu'$ when every other of the  $m-1$ players uses an oblivious strategy $\mu$ is given by equation \eqref{eqn:mpe-value-func}, but where~$\pi$ is replaced by $\pi_m$. That is, if the number of players is $m$, the payoff obtained each period is given by $\pi_m$. Hence, with some abuse of notation, for this section, we define:
 \begin{multline}
 \label{eqn:mpe-value-func-seq}
 V^{(m)}\big(x, f\ |\ \mu', \muVecmone \big) \triangleq\\
 \E\Big[\sum_{t =0}^{\infty}\beta^{t}\pi_m\big(x_{i,t}, a_{i,t},\fmi\big) \ \big| \
 x_{i,0} = x, f_{-i,0}^{(m)} = f;
 \mu_{i} = \mu', \muVeci = \muVecmone \Big].
 \end{multline}

 We generalize Theorem \ref{th:AME} for this setting. First, we need to strengthen Assumption \ref{as:continuity}.

 \begin{assumption} \label{as:continuity-S}
 For each $m\in \Z_+$, Assumption \ref{as:continuity} holds, with the
 following strengthened properties.
 \begin{enumerate}
 \item {\em Equicontinuity.}  The set of functions $\{\pi_m(x, a, f): m\in\Z_+\}$
   is jointly equicontinuous \footnote{\newrj{Let $\mc{X}$ and
       $\mc{Y}$ be two metric spaces, with metrics $d_\mc{X}$ and
       $d_{\mc{Y}}$ respectively.  A set of functions
       $\mc{F}$ mapping $\mc{X}$ to $\mc{Y}$ is said to be
       equicontinuous at $x_{0} \in \mc{X}$, if for every $\epsilon >
       0$, there exists a $\delta > 0$ such that $d_{\mc{Y}}(f(x), f(x_{0}))<
       \epsilon$ for all $f \in \mc{F}$ and all $x$ such that
       $d_{\mc{X}}(x_{0}, x) < \delta$.}} in $a \in \mc{A}$ and $f \in
   \mfr{F}_p$. 
 \item {\em Uniform growth rate bound.}  There exist constants $K$ and $n \in \Z_{+}$ such
   that $\sup_{m\in\Z_+,a \in \mc{A},f \in \mfr{F}_p} | \pi_m(x,a,f)| \leq K
 (1 + \norm{x}_\infty)^{n}$ for every $x \in \mc{X}$.
 \end{enumerate}
 \end{assumption}

 The following result is more general than Theorem \ref{th:AME}, because the payoff function in the AME property depends on $m$.

 \begin{theorem}[AME]
 Suppose Assumption \ref{as:continuity-S} holds. Let $(\mu, f)$ be a stationary equilibrium with $f \in \mfr{F}_p$.  Then the AME property holds
 for $(\mu, f)$.
 \label{th:AME-S}
 \end{theorem}

 The proof is similar to Theorem \ref{th:AME}, but requires an
 additional step to accommodate the sequence of payoff
 functions. However, note that similar to Theorem \ref{th:AME}, the
 stationary equilibrium $(\mu,f)$ is fixed and is computed with the
 limit payoff function $\pi$. Alternatively, it is possible to
     define an ``oblivious equilibrium'' (OE) for each finite model. An
     OE is similar to SE in the sense that agents optimize assuming
     that the long run population state is constant; the main
     difference is that it is defined in a finite model rather than in the
     limit model. Under a {\em
       uniform} light-tail condition, it can be shown that the sequence
     of OE satisfies the AME property
     \cite{weintraub_2008}. In addition, we conjecture
     that a version of the assumptions that guarantee existence of SE
     in Section \ref{sec:exist}, but that applies uniformly over all
     finite models, would guarantee that such a uniform light-tail
     condition holds. For clarity of presentation, we chose to work
     with the SE of the limit model directly.

 Moreover, we believe that
     the existence result for the limit model that we provide is
     important, because even though OE might exist under mild conditions for each finite
     model, SE in the limit model may fail to exist.  In particular, as
     we discuss in Section \ref{sec:examples}, this might be the case
     in applications that exhibit ``increasing returns to scale''.  See
     in particular Sections \ref{subsec:RD}, \ref{subsec:spillovers}, and
     \ref{subsec:learning}, for examples of how
     limit models are derived in specific applications, and also
     conditions in such models that ensure stationary equilibria provide
     accurate approximations.

\section{Additional Examples}
\label{ap:additional_examples}
In this section we present two additional applications to our results, to 
a model of supply chain competition, and a model of consumer learning.

 \subsection{Supply Chain Competition}
 \label{subsec:SC}

 We now consider an example of supply chain competition among firms \citep{CAC99},
 where the firms use a common resource that is sold by a single supplier.
 The firms only interact with each other in the sourcing stage as the goods produced are 
 assumed to be sold in independent markets.

 {\em States.}  We let the state $x_{i,t}$ be the inventory of goods
 held by firm $i$ at time $t$.

 {\em Actions.}  At each time period, the supplier runs an auction to
 sell the goods.  Each firm $i$ places a bid $a_{i,t}$  at time $t$; for example, $a_{i,t}$ may denote the willingness-to-pay of
 the supplier, or it may be a two-dimensional bid consisting of desired
 payment and quantity. Since the interaction between firms is via
 their action profiles we again assume that the action taken by a firm 
 lies in a  finite subset  $S$ of the integer lattice.

 {\em Transition dynamics.}  Suppose that each firm $i$ sees demand
 $d_{i,t}$ at time $t$; we assume $d_{i,t}$ are i.i.d.~and
 independent across firms, with bounded nonnegative support and
 positive expected value.  Further, suppose that
 when a firm bids $s$ and the population state-action
 profile is $f$, the firm receives an allocation $\xi(s, f)$.  Then the state
 evolution for a firm $i$ is given by $x_{i,t+1} = \max\{x_{i,t} - d_{i,t}, 0\} + \xi(s_{i,t}, \fmi)$.
 Note that $\xi$ depends on $\fmi$
 only through the marginal distribution over actions.  We make the
 natural assumptions that
 $\xi(s,f)$ is increasing in $s$ and decreasing in $f$ (where the set of distributions
 is ordered in the first order stochastic dominance sense).
 Thus the
 transition kernel captures inventory evolution in the usual way:
 demand consumes inventory, and procurement restocks inventory.  The
 amount of resource procured by a firm
 and the price it pays depends on its own bid, as
 well as bids of other firms competing for the resource.

 As one example of how $\xi$ might arise, suppose that the supplier
 uses a {\em proportional allocation mechanism} \citep{Kelly97}.  In such a mechanism,
 the bid $s$ denotes the total amount a firm pays.  Further, suppose
 the total quantity $Q_m$ of the
 resource available scales with the number of firms, i.e., $Q_m = mQ$.
 Let $k(s|f) = \sum_x f(x,s)$ denote the fraction of agents bidding $s$
 in population state-action profile $f$.  

   As $m \to \infty$, and introducing $R$ as a small ``reserve'' bid that ensures the denominator is
 always nonzero, we obtain the following limiting proportional allocation
 function: $\xi(s,f) = s Q /\Big(R + \sum_{s'} s' k(s' | f)\Big)$. Note that this expression is increasing in $s$ and decreasing in $f$.

 {\em Payoffs.}  A firm earns revenue for demand served, and incurs a
 cost both for holding inventory, as well as for procuring additional
 goods via restocking.  We assume every firm faces an exogenous retail
 price $\phi$.  (Heterogeneity in the retail price could be captured
 via the description in Section \ref{sec:extensions}.)
 Let $h$ be
 the unit cost of holding inventory for one period and let $\Omega(s,f)$ be
 the procurement payment made by a firm with bid $s$, when the
 population state-action profile is $f$; of course, $\Omega$ also
 depends on $f$ only through $k(\cdot|f)$.  In general we assume that
 $\Omega$ is increasing in $f$ for each fixed $s$.  In the proportional
 allocation mechanism described above, we simply have $\Omega(s,f) =
 s$.
 Since the demand is i.i.d., the single period payoff for a firm is
 given by the expected payoff it receives, where the expectation is
 over the demand uncertainty; i.e. $ \pi(x,s,f) = \phi\E[\min\{d,x\}] - hx -\Omega(s,f).$

 {\em Discussion.}  We have the following proposition.

 \begin{proposition}
 \label{prop:SCC}
 Suppose that $d$ has positive expected value.  Then there exists an SE
 for the supply chain competition model with the proportional
 allocation mechanism, and all SE possess the AME property.
 \end{proposition}

\begin{proof}
We present the proof in a more general setting, and specialize to the
 proportional allocation mechanism.
 If $\xi$ and $\Omega$ are uniformly bounded and
 appropriately continuous in $f$ for each pure action $s$, then Assumption \ref{as:continuity} follows in a straightforward
 manner.  For example, in the proportional allocation mechanism with a
 positive reserve bid $R$, note that $\xi$ is continuous in $f$ in the $\onep$ norm with $p =
 1$, since $\xi$ depends on $f$ through its first moment.   Since this is a
 model with finite action spaces, the result of Proposition
 \ref{prop:finiteconvex} also applies.  Thus, as before, the proof is reduced to determining
 whether Assumption \ref{as:compactness} holds for the given model.

 As before, Assumption \ref{as:compactness}.2, Assumption \ref{as:compactness}.3,
 and Assumption \ref{as:compactness}.4 are easy to
 check.  Assumption \ref{as:compactness}.1 follows because the payoff
 function is decreasing in $x$ for large $x$.  Finally, suppose $0 \in
 S$ and $\xi(0,f) = 0$ for all $f$; this will be the case, for example,
 in the proportional allocation mechanism with reserve $R$.  Then if
 $\mc{A}' = \{0\}$, it follows that Assumption
 \ref{as:compactness}.5 holds, as long as (1) $d$ has positive expected
 value; and (2) bidding zero is myopically optimal, and this induces
 negative drift in the inventory level.  Note that bidding zero is
 myopically optimal for the proportional allocation mechanism, and
 this induces negative drift in the inventory. Using Corollary
 \ref{co:existence} and Theorem \ref{th:AME} the result follows.
 \end{proof}

 More generally, for other choices of allocation mechanism, it can be
 shown that the same result holds if $d$ has positive expected value
 and the following conditions hold: (1) if $\xi$ and $\Omega$ are uniformly bounded and 
 appropriately continuous in $f$ for each pure action $s$; (2) $0 \in
 S$ and $\xi(0,f) = 0$ for all $f$; and (3) bidding zero maximizes
 a firm's single period payoff, and this induces negative drift in the
 inventory level.  

 In this model, decreasing returns to higher states are naturally enforced because the
 payoff function  becomes {\em decreasing} in the state as the state grows.  Simply because
 holding inventory is costly, firms prefer not to become arbitrarily
 large.  Thus in this model light tails in the population state can be
 guaranteed under fairly weak assumptions on the model primitives.

 \subsection{Consumer Learning}
 \label{subsec:consumer-learning}

 In this section, we analyze a model of social learning.  Imagine a
 scenario where a group of individuals decide to consume a product
 (e.g., visiting a restaurant). These individuals learn from each
 other's experience, perhaps through product reviews or word-of-mouth
 (see, for example, \citealp{ching2010}).

 {\em States.} We let~$x_{i,t}$ be the experience level of an
 individual at time~$t$.

 {\em Actions.} At each time period~$t$, an individual invests an
 ``effort''~$a_{i,t} \in [0, \ol{a}] $ in searching for a new product. 

 {\em Payoffs.} At each time period, an individual selects a product to
 consume.  The quality of the product
 is a normally distributed random variable~$Q$ with a
 distribution given by $Q \sim \mc{N}\left(\gamma a, \omega(x,f)\right)$,
 where $\gamma > 0$ is a constant.  Thus, the average quality of the
 product is proportional to the amount of effort made. Furthermore, the variance of the product is dependent on both
 individual and population experience levels.

 We assume that $\omega(x,f)$ is continuous in the population state $f$
 (in an appropriate norm, cf. Section \ref{sec:exist}).  We make the natural assumption that $\omega(x, f)$ is a
 nonincreasing function of~$f$ and strictly decreasing in $x$ (where the set of distributions
 is ordered in the first order stochastic dominance sense).  
 This is
 natural as we expect that as an individual's experience increases or
 if she can learn from highly expert people, the
 randomness in choosing a product will decrease.
 We also assume
 that there exists constants $\sigma_L,\sigma_H$, such that  $\sigma^{2}_{L} \leq \omega(x, f) \leq \sigma^{2}_{H}$.

 The individual
 receives a utility~$U(Q)$, where~$U(\cdot)$ is a nondecreasing concave
 function of the quality. For concreteness, we let $U(Q) = 1 -
 e^{-Q}$. Since at each time, the individual selects the product or the
 restaurant in an i.i.d.~manner, the single period payoff is given by
 {$ \pi(x,a,f)  = E\left[ U(Q) \ | \ Q \sim \mc{N}\left(\gamma a, \omega(x,f)\right)\right] - d a  = 1 - e^{-\gamma a + \frac{1}{2} \omega(x,f)} - d a,$},
 where $d$ is the marginal cost of effort.

 {\em Transition dynamics.}
 An individual's experience level is improved as she expends effort
 because she learns more about the quality of products.  However, this
 experience level also
 depreciates over time; this depreciation is assumed to be
 player-specific and comes about because an individual's tastes may
 change over time.
 Thus, an
 individual's experience evolves (independently of the experience of
 others or their investments) in a stochastic manner. Several
 specifications for the transition kernel satisfying our assumptions
 can be used; for concreteness we assume that the dynamics are the same as
 those described in Section~\ref{subsec:RD}.

 {\em Discussion.}  Our main result is the following proposition.
 \begin{proposition}
 \label{prop:CL}
 Suppose that:
 \begin{equation}
 \label{eq:CLcondition}
 d \geq \gamma e^{-\gamma c_0 + \frac{1}{2}\sigma^{2}_{H}},
 \end{equation}
 where $c_0 = \delta/(\alpha(1-\delta))$.  Then there exists an SE
 for the consumer learning model, and all SE possess the AME property.
 \end{proposition}
 \begin{proof}
 Note that $\omega(x, f) < \sigma^{2}_{H}$ and thus the growth rate
 bound in Assumption \ref{as:continuity} is trivially satisfied. If
 $\omega(x, f)$ is continuous in~$f$ (in the appropriate~$\onep$ norm),
 then Assumption~\ref{as:continuity} follows in a straightforward
 manner. To verify that~$\Phi(f)$ is convex, we note that
 Assumption~\ref{as:convexity} will hold if~$\pi(x,a,f)$ is strictly
 increasing in~$x$ and concave in~$a$. Since~$\omega(x,f)$ is strictly
 decreasing in~$x$, these conditions are naturally satisfied for our
 model. Thus, to complete the proof, we need to verify Assumption~\ref{as:compactness}.

 It is straightforward to check that Assumption
 \ref{as:compactness}.2-4 hold; we omit the
 details. Assumption~\ref{as:compactness}.1 follows since $\omega(x,
 f)$ is nonincreasing in~$x$ and bounded below, so $\omega(x, f) -
 \omega(x + \Delta,f) \to 0$ as $x \to \infty$.  In order for
 Assumption~\ref{as:compactness}.5 to hold, we require~$\mc{A}'$ to
 contain all {\em myopically} optimal actions. A straightforward
 calculation shows that~$\arg\max_a \pi(x,a,f) = a^*(x,f)$, where
 \[
 a^{*}(x,f) = \frac{1}{2\gamma} \omega(x,f) - \frac{1}{\gamma}\log\left(\frac{d}{\gamma}\right);
 \]
 for simplicity we assume $0 < a^*(x,f) < \ol{a}$ for all $x,f$,
 though an analogous argument holds otherwise.
 Thus we define $\mc{A}' = [0, a_{\max}]$, where: 
 \[ a_{\max} = \frac{1}{2\gamma} \sigma_H^2 -
 \frac{1}{\gamma}\log\left(\frac{d}{\gamma}\right). \]
 To verify Assumption \ref{as:compactness}.5(a), note that if $a
 \not\in \mc{A}'$, then:
 \begin{align*}
 \pi(x, a^*(x,f), f) - \pi(x,a,f) &= e^{-\gamma a^*(x,f) + \frac{1}{2}
   \omega(x,f)} ( e^{-\gamma(a - a^*(x,f))} - 1) + d(a - a^*(x,f)) \\
 &= \frac{d}{\gamma} \kappa(a - a^*(x,f)),
 \end{align*}
 where $\kappa(x) = e^{-\gamma x} - 1 + \gamma x$, which is strictly
 increasing and nonnegative with $\kappa(0) = 0$.  Here the preceding derivation follows
 by observing that the optimality condition for $a^*(x,f)$ ensures that $\gamma
 e^{-\gamma a^*(x,f) + \frac{1}{2} \omega(x,f)} = d$.  Thus Assumption
 \ref{as:compactness}.5(a) holds.

 When do the actions in $\mc{A}'$ produce negative drift in the state?
 For the dynamics given
 in~Section~\ref{subsec:RD}, one can easily verify that the drift is
 negative if the action is sufficiently small; in particular, the drift
 is negative for any action $a$ such that:
 \[
 a < \frac{\delta}{(1-\delta)\alpha} \triangleq c_{0},
 \]
 where $\delta \in (0,1)$ is the probability that the experience
 depreciates and $\alpha > 0$ controls the probability that a player is
 successful in improving the experience. The above
 inequality is satisfied by all $a' \in \mc{A}'$ if :
 \[
 d \geq \gamma e^{-\gamma c_0 + \frac{1}{2}\sigma^{2}_{H}}.
 \]
 Using Corollary
 \ref{co:existence} and Theorem \ref{th:AME} the result follows.
 \end{proof}

 Recall that $\delta \in (0,1)$ is the probability that the experience
 depreciates and $\alpha > 0$ controls the probability that a player is
 successful in improving the experience.  The right hand side is an
 upper bound to the marginal gain in utility due to effort, at effort
 level $c_0$; while the left hand side is the marginal cost of effort.
 Thus the condition  
 \eqref{eq:CLcondition} can be interpreted as a requirement that the
 marginal cost of effort should be sufficiently large relative to the
 marginal gain in utility due to effort.
  Otherwise, an individual's
 effort level when her experience is high will cause her state to
 continue to increase, so a light-tailed SE may not
 exist.  Hence we see the same dichotomy as before: decreasing returns
 to higher states yield existence of SE and the AME property, while
 increasing returns may not.

\section{Existence and AME: Preliminary Lemmas}

We begin with the following lemma, {which follows from the growth rate
bound and bounded increments in Assumption \ref{as:continuity}.}

\begin{lemma}
\label{lem:supbound}
Suppose Assumption \ref{as:continuity} holds.  Let $x_0 = x$.
Let $a_{t} \in \mc{A}$ be any sequence of (possibly history dependent)
actions, and let {$f_t \in \mfr{F}$} be
any sequence of (possibly history dependent) population states.  Let
$x_t$ be the state sequence generated, i.e., $ x_t \sim \mbf{P}(\cdot\ |\ x_{t-1}, a_{t-1}, f_{t-1}).$
Then for all $T \geq 0$, there exists $C(x,T) < \infty$ such that $\E\left[\sum_{t = T}^\infty \beta^t |\pi(x_t, a_t, f_t)|\ \big |\
  x_0 = x \right] \leq C(x,T)$. 
Further, $C(x,T) \to 0$ as $T \to \infty$.
\end{lemma}

\begin{proof}
Observe that by Assumption \ref{as:continuity}, the increments are bounded. Thus
starting from state~$x$, we have $\norm{x_{t}}_\infty \leq \norm{x}_{\infty} + tM$. Again by
Assumption \ref{as:continuity}, $|\pi(x_t, a_t, f_t)| \leq K
(1 + \norm{x_t}_\infty)^{n}$.  Therefore:
\[ \E\left[ \sum_{t=T}^\infty \beta^t |\pi(x_t, a_t, f_t)|\ |\ x_0 = x
\right] \leq K \sum_{t = T}^\infty
 \beta^t ( 1 + \norm{x}_\infty + tM)^{n}. \]
We define
$C(x,0)$ as the right hand side above when $T = 0$:
\[ C(x,0) = K \sum_{t = 0}^\infty
 \beta^t ( 1 + \norm{x}_\infty + tM)^{n}. \]
Observe that $C(x,0) < \infty$.

We now reason as follows for $T \geq 1$:
\begin{align*}
K\sum_{t = T}^{\infty}\beta^{t}(1 + \norm{x}_\infty + tM)^{n} &= K\beta^{T} \sum_{t=0}^{\infty}\beta^{t}(1 + \norm{x}_\infty + tM + TM)^{n} \\
&= K\beta^{T} \sum_{t=0}^{\infty}\beta^{t}\sum_{j=0}^{n} {{n} \choose {j}}(1 + \norm{x}_\infty + tM )^{j}(TM)^{n-j} \\
&\leq K\beta^{T} \sum_{t=0}^{\infty}\beta^{t}\sum_{j=0}^{n} {{n} \choose {j}}(1 + \norm{x}_\infty + tM )^{n}(TM)^{n}\\
&= K\beta^{T}2^{n}(TM)^{n}\sum_{t=0}^{\infty}\beta^{t}(1 + \norm{x}_\infty + tM )^{n} \\
& = C(x,0) \beta^{T}(2MT)^{n}.
\end{align*}
Here the inequality holds because $1 + \|x\|_\infty + tM \geq 1$, $M
\geq 0$, and $T \geq 1$.  So for $T \geq 1$, define:
\begin{equation}
\label{eq:C(x,T)}
C(x,T) = C(x,0) \beta^T (2MT)^n.
\end{equation}
Then $C(x,T) \to 0$ as $T \to \infty$, as required. 
\end{proof}


We now show that the Bellman equation holds for the dynamic program
solved by a single agent given a population state $f$.  {Given our unbounded state space,} our proof
involves the use of a {\em weighted sup norm}, defined as follows.
For each $x \in \mc{X}$, let $W(x) = (1 + \norm{x}_\infty)^{n}$.  For a function
$F: \mc{X} \to \R$, define:
\[ \norm{F}_{W\mhyphen \infty} = \sup_{x \in \mc{X}} \left|
  \frac{F(x)}{W(x)}\right|. \]
This is the weighted sup norm with weight function $W$.  We let
$B(\mc{X})$ denote the set of all functions $F : \mc{X} \to \R$ such
that $\norm{F}_{W\mhyphen \infty} < \infty$.

Let $T_f$ denote the dynamic programming operator with population
state $f$: given a function $F : \mc{X} \to \R$, we have $(T_fF)(x) = \sup_{a \in \mc{A}} \Big\{ \pi(x,a,f) +
 \beta\sum_{x' \in \mc{X}}F(x')\mbf{P}(x'\ |\ x, a, f)\Big\}$.
We define $T^{k}_{f}$ to be the composition of the mapping $T_{f}$
with itself~$k$ times.
The following lemma applies standard dynamic programming arguments.

\begin{lemma}
\label{lem:bellman}
Suppose Assumption \ref{as:continuity} holds.  For all {$f \in \mfr{F}$}, if $F \in B(\mc{X})$ then $T_fF \in
B(\mc{X})$.  Further, there exist $k,\rho$ independent of $f$ with
$0 < \rho < 1$ such that $T_f$ is a $k$-stage $\rho$-contraction
on $B(\mc{X})$; i.e., if $F, F' \in B(\mc{X})$, then for all $f$: $\norm{T_f^k F - T_f^k F'}_{W\mhyphen \infty} \leq \rho \norm{F -
 F'}_{W\mhyphen \infty}.$

In particular, value iteration converges to $\vTilde^*(\cdot | f) \in B(\mc{X})$
from any initial value function in $B(\mc{X})$, and for all {$f
\in \mfr{F}$} and $x \in \mc{X}$, the Bellman equation holds:
\begin{equation}
\label{eq:bellman}
\vTilde^{*}(x \ | \ f) = \sup_{a \in \mc{A}}\Big\{ \pi(x,a,f) +
 \beta\sum_{x' \in \mc{X}}\vTilde^{*}(x' \ | \ f)\mbf{P}(x'\ | \  x,
 a, f)\Big\}.
\end{equation}
Further, $\vTilde^{*}(x | f)$ is continuous in {$f \in \mfr{F}_p$}.

Finally, there exists at least one optimal oblivious strategy among
all (possibly history-dependent, possibly randomized)
strategies; i.e., $\mc{P}(f)$ is nonempty {for all $f\in\mfr{F}$}.  An
oblivious strategy $\mu \in \mfr{M}_O$ is optimal given $f$ if and only
if $\mu(x)$ achieves the maximum on the right hand side of \eqref{eq:bellman} for every $x \in \mc{X}$.
\end{lemma}

\begin{proof}
 We have the following three properties:
 \begin{enumerate}
 \item By growth rate bound in Assumption \ref{as:continuity} we have $\sup_{a} |
   \pi(x,a,f)|/W(x) \leq K$ for all $x$.
 \item We have:
 \[  \b{W}(x) = \sup_{a \in \mc{A}} \sum_{x'} \mbf{P}(x'\ | \ x, a, f) W(x') \leq (1 +
  \norm{x}_\infty + M)^{n}, \]
 since the increments are bounded (Assumption~\ref{as:continuity}).  Thus
 $\b{W}(x)/W(x) \leq  (1 + M)^{n}$ for all $x$.
 \item Finally, fix $\rho$ such that $0 < \rho < 1$ and let:
 \[ \b{W}_k(x) = \sup_{\mu \in \mfr{M}_O} \mbf{E}[ W(x_k) | x_0 =
 x, \mu], \]
 where the state evolves according to $x_{t+1} \sim \mbf{P}(\cdot \ | \ x_t,
 \mu(x_t), f)$.
 By bounded increments in Assumption \ref{as:continuity}, we
 have:
 \[  \beta^k \b{W}_k(x) \leq \beta^k (1 + \norm{x}_\infty +
 kM)^{n} \leq \beta^k (1 + kM)^n W(x). \]
 By choosing $k$ sufficiently large so that $\beta^k (1 + kM)^n <
 \rho$, we have:
 \[  \beta^k \b{W}_k(x) \leq \rho W(x).\]
 %

 \end{enumerate}

 Given (1)-(3), by standard arguments (see, e.g., \citealp{bertsekas_2007}), it
 follows that $T_f$ is a $k$-stage $\rho$-contraction with respect to
 the weighted sup norm, value iteration converges to $\vTilde^*(\cdot\ |\
 f)$, the Bellman equation holds, and any (stationary, nonrandomized)
 oblivious strategy that maximizes
 the right hand side in \eqref{eq:bellman} for each $x \in \mc{X}$ is optimal.
 Observe that
 since $\vTilde^*(\cdot\ |\ f) \in B(\mc{X})$ for any~$f$, it follows
 that $\vTilde^*(x\ |\ f) < \infty$ for all $x$.  In fact, by Lemma
 \ref{lem:supbound}, $|\vTilde^*(x\ |\ f)| \leq C(x,0)$ for all $x$.


 Next we show that $\vTilde^*(x \ |\ f)$ is continuous in $f$.  Define $Z(x) = 0$
 for all $x$, and let $V_f^{(\ell)} = T_f^{\ell}Z$.  We first show that
 $V_f^{(\ell)}(x)$ is continuous in $f$.
 To see this, we proceed by induction.  The result is trivially true at
 $\ell = 0$.  Next, observe that $\pi(x,a,f)$ is jointly continuous in~$a$
 and~$f$ for each fixed~$x$ by Assumption~\ref{as:continuity}.  Suppose
 $V_f^{(\ell)}(x)$ is continuous in $f$ for each $x$;
 then
 $V_f^{(\ell)}(x') \mbf{P}(x'\ |\  x, a, f)$ is jointly continuous in $a$ and
 $f$ for each fixed $x,x'$.  Since the kernel has bounded increments from
 Assumption \ref{as:continuity}, we conclude that $\sum_{x'}
 V_f^{(\ell)}(x') \mbf{P}(x'\ | \ x, a, f)$ is jointly continuous
 in~$a$ and~$f$ for each fixed $x$.  It follows by Berge's maximum
 theorem
 \citep{aliprantis_2006} that $V_f^{(\ell+1)}(x)$ is continuous in $f$.

 Fix $\epsilon > 0$.   Since $T_f$
 is a $k$-stage $\rho$-contraction in the weighted sup
 norm for every $f$, it follows that for all sufficiently large $\ell$,
 for every $f$ there holds:
 \[ | V_f^{(\ell)}(x) - \vTilde^*(x\ |\ f)| \leq W(x) \epsilon. \]
 So now suppose that $f_n \to f$ in the $\onep$ norm.  Since
 $V_f^{\ell}(x)$ is continuous in $f$, for all sufficiently
 large $n$ there holds:
 \[ | V_{f_n}^{(\ell)}(x) - V_f^{(\ell)}(x) | \leq \epsilon. \]
 Thus using the triangle inequality, for all sufficiently large $n$ we
 have:
 \[ | \vTilde^*(x\ |\ f) - \vTilde^*(x\ |\ f_n)| \leq (2W(x)
 +1)\epsilon. \]
 Since $\epsilon$ was arbitrary it follows that the left hand side
 approaches zero as $n \to\infty$, as required.
 Finally, observe that by a similar argument as above,
 \[ \sum_{x'} \vTilde^*(x'\ |\ f) \mbf{P}(x'\ | \ x, a, f) \]
 is a continuous function of $a$ for each fixed $x$ and $f$; since
 $\pi(x,a,f)$ is also continuous in $a$ for each fixed $f$, the right hand side of
 \eqref{eq:bellman} is continuous in $a$ for each fixed $f$.  Since $\mc{A}$ is
 compact, it follows that there exists an optimal action at each state $x$, and thus
 there exists an optimal strategy given~$f$. 
 \end{proof}

\section{Existence: Proof}

%

\subsection{Closed Graph: Proof}

Throughout this subsection we suppose Assumption \ref{as:continuity} holds.

\begin{lemma}
\label{lemma:cont-P}
For each $f$, $\mc{P}(f)$ is compact; further, the
correspondence $\mc{P}$ is upper hemicontinuous
on $\mfr{F}_p$.
\end{lemma}
\begin{proof}
By Assumption \ref{as:continuity}, $\pi(x, a, f)$ is jointly continuous in
$a$ and $f$.
Lemma~\ref{lem:bellman} establishes that the optimal oblivious value
function $\vTilde^*(x\ |\ f)$ is continuous
in $f$, and so as in the proof of that lemma, it follows that for a
fixed state  $x$, $\pi(x, a, f) + \beta \sum_{x'}\vTilde^*(x'\ |\ f)
\mbf{P}(x'\ |\  x,a, f)$ is finite and jointly continuous in $a$ and $f$.
Define the set $\mc{P}_x(f) \subset \mc{A}$ as the set of actions
that achieve the maximum on the right hand side of \eqref{eq:bellman}; this is nonempty
as $\mc{A}$ is compact (Assumption \ref{as:continuity}) and the right hand side is continuous in $a$.
By Berge's maximum theorem, for each $x$ the correspondence~$\mc{P}_x$ is upper hemicontinuous with compact values~\citep{aliprantis_2006}.

By Lemma \ref{lem:bellman}, $\mu \in \mc{P}(f)$ if and only if $\mu(x) \in
\mc{P}_x(f)$ for each $x$.  Note that we have endowed the set of
strategies with the topology of pointwise convergence. The range space of $\mc{P}$ is
an infinite product of the compact action space~$\mc{A}$ (Assumption
\ref{as:continuity}) over the countable state space.
Hence by Tychonoff's theorem~\citep{aliprantis_2006}, the range space
of $\mc{P}$ is compact.  Further, since~$\mc{P}_x$ is compact-valued,
it follows that $\mc{P}$ is compact-valued.
Since $\mc{P}_{x}(f)$ is compact-valued and upper hemicontinuous, the Closed Graph
Theorem ensures that $\mc{P}_{x}$  has a closed
graph~\citep{aliprantis_2006}. This in turn ensures that
$\mc{P}$ has closed graph; again by the Closed Graph Theorem, we
conclude that $\mc{P}$ is upper hemicontinuous. 
\end{proof}

\begin{sproof}{Proof of Proposition \ref{prop:uhc}.}
Suppose $f_k \to f$ in the $\onep$ norm, and that $g_k \to g$ in the
$\onep$ norm, where $g_k \in \Phi(f_k)$ for all $k$.  We must show
that $g \in \Phi(f)$.  For each $k$, let $\mu_k \in \mc{P}(f_k)$ be an
optimal oblivious strategy such that $g_k \in \mc{D}(\mu_k, f_k)$.  As in the proof
of Lemma~\ref{lemma:cont-P}, the range space of $\mc{P}$ is compact in
the topology of pointwise convergence; therefore, taking subsequences
if necessary, we can assume without loss of generality that $\mu_k$
converges to some strategy $\mu \in \mfr{M}_O$ pointwise.  By upper
hemicontinuity of $\mc{P}$ (Lemma \ref{lemma:cont-P}), we have $\mu
\in \mc{P}(f)$.

By definition of $\mc{D}$, it follows that for all $x$:
\begin{equation}
\label{eq:invariant_k}
 g_k(x) = \sum_{x'} g_k(x') \mbf{P}(x | x', \mu_k(x'), f_k).
\end{equation}
Since $\mbf{P}(x | x', a, f)$ is jointly continuous in action and
population state (Assumption \ref{as:continuity}), it follows that for all $x$ and $x'$:
\[ \mbf{P}(x | x', \mu_k(x'), f_k) \to \mbf{P}(x | x', \mu(x'), f) \]
as $k \to \infty$.  Further, if $g_k \to g$ in the $\onep$ norm, then
in particular, $g_k(x) \to g(x)$ for all $x$.  Finally, observe that
for all $a$ and $f$, we have $\mbf{P}(x | x', a, f) = 0$ for all
states $x'$ such that $\|x' - x\|_\infty > M$, since increments are
bounded (Assumption \ref{as:continuity}).  Thus:
\[ \sum_{x'} g_k(x') \mbf{P}(x | x', \mu_k(x'), f_k) \to \sum_{x'}
g(x') \mbf{P}(x | x', \mu(x'), f) \]
as $k \to \infty$.  Taking the limit as $k \to \infty$ on both sides
of \eqref{eq:invariant_k} yields:
\begin{equation}
 g(x) = \sum_{x'} g(x') \mbf{P}(x | x', \mu(x'), f),
\end{equation}
which establishes that $g \in \mc{D}(\mu, f)$.  Since we had $\mu \in
\mc{P}(f)$, we conclude $g \in \Phi(f)$, as required. 
\end{sproof}

\subsection{Convexity: Proof}

\begin{sproof}{Proof of Proposition \ref{prop:finiteconvex}.}
Fix $f \in \mfr{F}_p$, and let $g_1, g_2$ be elements of $\Phi(f)$.  Let
$\mu_1, \mu_2 \in \mc{P}(f)$ be strategies such that $g_i \in
\mc{D}(\mu_i, f)$, $i = 1,2$.  Then for $i = 1,2$ and all $x' \in \mc{X}$, we have:
\[ g_i(x') = \sum_x g_i(x') \mbf{P}(x'\ | \ x, \mu_i(x), f). \]

Fix $\delta$, $0 \leq \delta \leq 1$, and for each $x$, define $g(x)$ by:
\[ g(x) = \delta g_1(x) + (1 - \delta) g_2(x). \]
We must show $g \in \Phi(f)$.  Define a new strategy $\mu$ as follows:
for each $x$ such that $g(x) > 0$,
\[ \mu(x) = \frac{ \delta g_1(x) \mu_1(x) + (1 - \delta) g_2(x) \mu_2(x)}{g(x)}. \]
For each $x$ such that $g(x) = 0$, let $\mu(x) = \mu_1(x)$.

We claim that $\mu \in \mc{P}(f)$, i.e., $\mu$ is an optimal oblivious
strategy given $f$; and that $g \in \mc{D}(\mu, f)$, i.e., that $g$ is an
invariant distribution given strategy $\mu$ and population state $f$.  This suffices to
establish that $g \in \Phi(f)$.

To establish the claim, first observe that under Definition
\ref{def:finiteaction}, the right hand side of \eqref{eq:bellman} is
{\em linear} in $a$.  Thus any convex combination of two optimal
actions is also an optimal action.  This establishes that for every
$x$, $\mu(x)$ achieves the maximum on the right hand side of
\eqref{eq:bellman}; so we conclude $\mu \in \mc{P}(f)$.

Let $T = \{ x : g(x) > 0 \}$.  Then:
\begin{align*}
g(x') &= \delta g_1(x') + (1-\delta)g_2(x') \\
&= \sum_{x} \delta g_1(x) \mbf{P}(x'\ | \ x, \mu_1(x), f) + (1 - \delta)
g_2(x) \mbf{P}(x'\ | \ x, \mu_2(x), f) \\
&= \sum_x \sum_s \left( \delta g_1(x) \mu_1(x)(s) + (1 - \delta)
g_2(x) \mu_2(x)(s) \right) \mbf{P}(x'\ |\ x,s, f) \\
&= \sum_{x \in T} \sum_s g(x) \mu(x)(s) \mbf{P}(x'\ | \ x, s, f).
\end{align*}
The first equality is the definition of $g(x')$, and the second equality
follows by expanding the invariant distribution equations for $g_1$
and $g_2$.  The third equality follows by expanding the sum over pure
actions $s$.  Finally, in the last equality, we substitute the
definition of $g(x)$, and we also observe that for $x \not\in T$,
$g(x) = 0$---and therefore, $g_1(x) = g_2(x) = 0$.  Since $g(x) = 0$
for $x \not\in T$, it follows that:
\[ \sum_{x \not\in T} \sum_s g(x) \mu(x)(s) \mbf{P}(x'\ |\ x, s, f) = 0. \]
It follows that:
\[ g(x') = \sum_x g(x) \mbf{P}(x'\ |\ x, \mu(x), f), \]
as required. 
\end{sproof}

\begin{lemma}
\label{lem:incrconc}
Suppose Assumptions \ref{as:continuity} and \ref{as:convexity} hold.
Then $\vTilde^*(\cdot\ |\ f)$ is strictly increasing for every $f \in
\mfr{F}_p$, and the right hand side of \eqref{eq:bellman} is
strictly concave in $a$.
\end{lemma}

\begin{proof}
Define $Z(x) = 0$ for all $x$, and let $V_f^{(\ell)} =
T_f^\ell Z$.  Observe that if $V_f^{(\ell)}$ is nondecreasing, then under
the conditions of the lemma, it follows that $V_f^{(\ell+1)}$ will be
nondecreasing.  Taking the limit as $n \to
\infty$, we conclude (from convergence of value iteration) that
$\vTilde^*(\cdot\ |\ f)$ is nondecreasing, and thus the right hand side of
\eqref{eq:bellman} is {\em strictly increasing} in $x$.  From this it
follows that in fact, $\vTilde^*(\cdot\ |\ f)$ is strictly increasing.

Since $\vTilde^*(\cdot\ |\ f)$ is strictly increasing, $\pi(x,a,f)$ is concave in $a$, and the kernel is
stochastically concave in $a$, with at least one of the last two
strictly concave, it follows that the right hand
side of \eqref{eq:bellman} is strictly concave in $a$. 
\end{proof}

\begin{sproof}{Proof of Proposition \ref{prop:convexity-new}.}
Under Assumptions \ref{as:continuity} and  \ref{as:UIC}, the optimal
action in \eqref{eq:bellman} can be shown to be unique (see
\citealp{DORSAT03}).  It follows that $\mc{P}(f)$ is a singleton.

From the preceding lemma, Assumptions \ref{as:continuity} and
\ref{as:convexity} together also guarantee a unique optimal
solution in the right hand side of \eqref{eq:bellman}, for every $x \in
\mc{X}$.  Thus under either of these conditions the optimal strategy
given $f$ is unique, i.e., $\mc{P}(f)$ is a singleton. The result follows by Proposition \ref{prop:single_convex}. 
\end{sproof}

\subsection{Compactness: Proof}

Throughout this subsection we suppose $\mc{X} = \Z_+^d$ and that
Assumptions \ref{as:continuity} and \ref{as:compactness} are in
effect.

\begin{lemma}
\label{lem:coupling}
Given $x' \geq x$, $x,x'\in\mc{X}$, $a \in \mc{A}$, and {$f \in \mfr{F}$}, there exists a probability
space with random variables $\xi' \sim \mbf{Q}( \cdot \ | \ x', a, f)$, $\xi
\sim \mbf{Q}( \cdot \ |\ x, a, f)$, such that $\xi' \leq \xi$ almost surely,
and $x' + \xi' \geq x + \xi$ almost surely.
\end{lemma}

\begin{proof}
The proof uses a coupling argument.  Let $U$ be a
uniform random variable on $[0,1]$.  Let $F_\ell$ (resp., $F_\ell'$) be the
cumulative distribution function of $\mbf{Q}_\ell(\cdot\ |\ x,a, f)$ (resp.,
$\mbf{Q}_\ell(\cdot\ |\ x', a, f)$), and let $G_\ell$ (resp., $G_\ell'$)
be the cumulative distribution function of $\mbf{P}_\ell(\cdot\ |\ x, a, f)$
(resp., $\mbf{P}_\ell(\cdot\ |\ x', a, f)$).  By Assumption
\ref{as:compactness}, $\mbf{P}_\ell(\cdot\ |\ x, a, f)$ is
stochastically nondecreasing in $x$, and $\mbf{Q}_\ell(\cdot\ |\ x, a,
f)$ is stochastically
nonincreasing in $x$.  Thus for all $z$, $F_\ell(z) \leq F_\ell'(z)$,
but for all $y$, $G_\ell(y) \geq
G_\ell'(y)$; further, $G_\ell(y) = F_\ell(y-x_\ell)$ (and $G_\ell'(y) =
F_\ell'(y - x_\ell')$).  Let $\xi_\ell = \inf \{ z_\ell: F_\ell(z_\ell) \geq U\}$, and let
$\xi_\ell' = \inf \{ z_\ell : F'(z_\ell) \geq U \}$.  Then $\xi_\ell
\geq \xi_\ell'$ for all $\ell$, i.e., $\xi \geq \xi'$.  Rewriting
the definitions, we also have $x_\ell + \xi_\ell = \inf \{ y_\ell : F_\ell(y_\ell - x_\ell) \geq
U\}$, and $x_\ell' + \xi_\ell' = \inf \{ y_\ell : F_\ell'(y_\ell -
x_\ell') \geq U\}$, i.e., $x_\ell +
\xi_\ell = \inf \{ y_\ell : G_\ell(y_\ell) \geq U\}$, and $x_\ell' +
\xi_\ell' = \inf \{ y_\ell : G_\ell'(y_\ell) \geq U\}$.  Thus $x_\ell
+ \xi_\ell \leq x_\ell' + \xi_\ell'$ for all $\ell$, i.e., $x' + \xi'
\geq x + \xi$, as required. 
\end{proof}

Given a set $S$ define $\rho_\infty(x,S) = \inf_{y \in S} \|x -
y\|_\infty.$ Thus $\rho_\infty$ gives the $\infty$-norm distance to a set.  We have
the following lemma.

\begin{lemma}
\label{lem:mulimit}
As $\|x\|_\infty \to \infty$, $\sup_{{f\in \mfr{F}}}\sup_{\mu \in \mc{P}(f)} \rho_\infty(\mu(x), \mc{A}') \to 0.$
\end{lemma}

\begin{proof}
Suppose the statement of the lemma fails; then there exists $r > 0$ and
a sequence ${f_n \in \mfr{F}}$, $\mu_n \in \mc{P}(f_n)$, and $x_n$ (where
$\|x_n\|_\infty\to\infty$) such that
$\rho_\infty(\mu_n(x_n), \mc{A}') \geq r$ for all $n$.
We use
this fact to construct a profitable deviation from the policy $\mu_n$,
for sufficiently large $n$.

Observe that by Assumption \ref{as:compactness}, there must exist
$a_n'\in \mc{A}'$ with $a_n' \leq \mu_n(x_n)$, such that:
\[ \pi(x_n, a_n', f_n) - \pi(x_n, \mu_n(x_n), f_n) \geq \kappa( \|a_n'
- \mu_n(x_n)\|_\infty) \geq \kappa( r) > 0, \]
where the last inequality follows since $\kappa$ is strictly
increasing with $\kappa(0) = 0$.  Importantly, note the bound on the
right hand side is a constant, 
independent of $n$.

Let $x_{0,n} = x_n$, and let $x_{t,n}$ and $a_{t,n}$ denote the state
and action sequence realized under $\mu_n$, starting from $x_{0,n}$,
under the kernel $\mbf{P}(\cdot | x, a, f_n)$.  We consider a
deviation from $\mu_n$, where at time~$0$, instead of playing $a_{0,n}
= \mu_n(x_n)$, the agent plays $a_{0,n}' = a_n'$; and then at all times in the
future, the agent {\em follows the same actions as the original
  sequence}, i.e., $a_{t,n}' = a_{t,n}$.  Let $x_{t,n}'$ denote the
resulting state sequence.

Since the kernel is stochastically nondecreasing in $a$, and $a_n'
\leq a_n$, it follows that there exists a common probability space
together with increments $\xi_{0,n}, \xi_{0,n}'$, such that $\xi_{0,n} \sim
\mbf{Q}(\cdot | x_n, a_n, f_n)$, $\xi_{0,n}' \sim \mbf{Q}(\cdot | x_n,
a_n', f_n)$, and $\xi_{0,n}' \leq \xi_{0,n}$ almost surely.  Thus we can
couple together $x_{1,n}$ and $x_{1,n}'$, by letting $x_{1,n} = x_n +
\xi_{0,n}$, and $x_{1,n}' = x_n + \xi_{0,n}'$.  In particular, observe
that with these definitions we have $x_{1,n} \geq x_{1,n}'$.  Let
$\Delta_n = \xi_{0,n} - \xi_{0,n}' \geq 0$.  Note that
$\|\Delta_n\|_\infty \leq 2M$, by Assumption \ref{as:continuity}
(bounded increments).

Next, it follows from Lemma \ref{lem:coupling} that there exists a
probability space with random variables $\xi_{1,n}, \xi_{1,n}'$ such
that $\xi_{1,n} \sim \mbf{Q}(\cdot | x_{1,n}, a_{1,n}, f_n)$ and
$\xi_{1,n}' \sim \mbf{Q}(\cdot | x_{1,n}', a_{1,n}, f_n)$, $\xi_{1,n}
\leq \xi_{1,n}'$ almost surely, and yet $x_{1,n} + \xi_{1,n} \geq
x_{1,n}' + \xi_{1,n}'$ almost surely.  Thus we can couple together
$x_{2,n}$ and $x_{2,n}'$, by letting $x_{2,n} = x_{1,n} + \xi_{1,n}$,
and let $x_{2,n}' = x_{1,n}' + \xi_{1,n}'$.  Proceeding inductively,
it can be shown that there exists a joint probability measure under
which $0 \leq x_{t,n} - x_{t,n}' \leq \Delta_n,$ 
almost surely, for all $t \geq 1$ (where the inequalities are
interpreted coordinatewise); this follows by a standard application of
the Kolmogorov extension theorem.

We now compare the payoffs obtained under these two sequences.  We
have:
\begin{align*}
\E\Big[ \sum_t \beta^t ( \pi(x_{t,n}, a_{t,n}, f_n)& - \pi(x_{t,n}',
  a_{t,n}', f_n)) \Big] = \pi(x_n, \mu_n(x_n), f_n) - \pi(x_n, a_n', f_n)\\&
\quad \quad  + \E\left[ \sum_{t\geq 1}
  \beta^t ( \pi(x_{t,n}, a_{t,n}, f_n) - \pi(x_{t,n}',
  a_{t,n}, f_n)) \right]\\
& \leq -\kappa(r) + \E\left[ \sum_{t\geq 1}
  \beta^t
\sup_{\delta \geq 0 :
    \|\delta\|_\infty \leq 2M} \sup_{a,f} (\pi(x_{t,n}, a, f)
  -\pi(x_{t,n}-\delta, a, f))\right].
\end{align*}

Since increments are bounded (Assumption
\ref{as:compactness}), in time $t$, the maximum distance the
state could have moved in each coordinate from the initial state $x$
is bounded by $tM$.  Thus if $x_{0,n} =
x_n$, then:
\[
\sup_{\delta \geq 0: \|\delta\|_\infty \leq
2M} \sup_{a,f} (\pi(x_{t,n}, a, f) -\pi(x_{t,n}-\delta, a, f))
\leq \sup_{ \substack{\delta \geq 0, \epsilon: \|\delta\|_\infty \leq
2M, \\ \|\epsilon\|_\infty \leq tM }} \sup_{a,f} (\pi(x_n + \epsilon, a, f) - \pi(x_n +
\epsilon-\delta, a, f)). \]
Let $A_{t,n}$ denote the right hand side of the preceding equation; note
that this is a deterministic quantity, and that the supremum is over a
finite set.  Thus from Assumption \ref{as:compactness}, we have
$\limsup_{n \to \infty} A_{t,n} \leq 0$.

Finally, observe that since $\limsup_{\|x\|_\infty \to \infty}
\sup_{a,f} (\pi(x + \delta, a, f) - \pi(x,
a, f)) \leq 0$, it follows that:
\[ \sup_{y \in \Z_+^d, \delta \geq 0: \|\delta\|_\infty \leq
  2M} \sup_{a,f} (\pi(y, a, f) - \pi(y-\delta,
a, f)) < \infty. \]
We denote the left hand side of the preceding inequality by $D$.  Note
that this is a constant independent of $n$.

Combining our arguments, we have that for all sufficiently large $n$,
there holds:
\[ \E\left[ \sum_t \beta^t ( \pi(x_{t,n}, a_{t,n}, f_n) - \pi(x_{t,n}',
  a_{t,n}', f_n)) \right] \leq -\kappa(r) + \sum_{t = 1}^T
\beta^t A_{t,n} + \frac{\beta^T D}{1 - \beta}. \]
By taking $T$ sufficiently large, we can ensure that the last term on
the right hand side is strictly less than $\kappa(r)/2$; and by
then taking $n$ sufficiently large, we can ensure that the second term
on the right hand side is also strictly less than $\kappa(r)/2$.  Thus for sufficiently large $n$, we conclude that the
left hand side is negative---contradicting optimality of $\mu_n$.  The
lemma follows. 
\end{proof}

\begin{lemma}
\label{lem:negincr}
There exists $\ol{\epsilon} > 0$ and $\ol{K}$ such that for all $\ell$
and all $x$ with $x_\ell \geq \ol{K}$, $ \\
\sup_f \sup_{\mu \in \mc{P}(f)} \sum_{z_\ell} z_\ell \mbf{Q}_\ell(z_\ell\ |\ x, \mu(x),
f) < -\ol{\epsilon}.$
\end{lemma}

\begin{proof}
Fix $\epsilon > 0$
so that for all $\ell$ and all $x'$ with $x_\ell' \geq K', $ $ \sup_{a' \in \mc{A}'} \sup_f \sum_{z_\ell} z_\ell \mbf{Q}_\ell(z_\ell | x', a', f) <
-\epsilon;$ such a constant exists by the last part of Assumption
\ref{as:compactness}.  Observe that since $\mc{A}$ is compact and
$\sup_f \sum_{z_\ell} {z_\ell} \mbf{Q}_\ell({z_\ell}\ |\ x, a, f)$ is
continuous in $a$ (Assumption \ref{as:compactness}), it follows that
$\sup_f \sum_{z_\ell} {z_\ell} \mbf{Q}_\ell({z_\ell}\ |\ x, a, f)$ is
in fact {\em uniformly} continuous
in $a \in \mc{A}$.  Let
$e^{(\ell)}$ denote the $\ell$'th standard basis vector (i.e.,
$e_{\ell'}^{(\ell)} = 0$ for $\ell' \neq \ell$, and $e_{\ell}^{(\ell)}
= 1$).  By uniform continuity, we can conclude there must exist a $\delta_\ell > 0$ such
that if $\|a -  a'\|_\infty < \delta_\ell$, then:
\[ \left| \sup_f \sum_{z_\ell} {z_\ell} \mbf{Q}_\ell({z_\ell} |
  K'e^{(\ell)}, a, f) - \sup_f \sum_{z_\ell}
  {z_\ell} \mbf{Q}_\ell({z_\ell} | K'e^{\ell}, a', f) \right| < \epsilon/2. \]
Note in particular, if $\rho_\infty(a, \mc{A}') < \delta_\ell$, then
there exists $a' \in \mc{A}'$ with $\|a - a'\|_\infty < \delta_\ell$.
By our choice of~$\epsilon$ we have $ \sup_f \sum_{{z_\ell}} {z_\ell} \mbf{Q}_\ell({z_\ell}\ | K'e^{(\ell)}, a, f) < -\frac{\epsilon}{2}.$
Now let $\delta = \min\{ \delta_1, \ldots, \delta_d\}$.  Since the
increment kernel is stochastically
nonincreasing in $x$, it follows that if $\rho_\infty(a, \mc{A}') <
\delta$ and $x_\ell \geq K'$, then $ \sup_f \sum_{z_\ell} {z_\ell} \mbf{Q}_\ell({z_\ell} | x, a, f) < -\frac{\epsilon}{2}. $
Since $\sup_f \sup_{\mu \in \mc{P}(f)}
\rho_\infty(\mu(x), \mc{A}') \to 0$ as $\|x\|_\infty \to \infty$, the
result follows if we let $\ol{\epsilon} = \epsilon/2$. 
\end{proof}

\begin{lemma}
\label{lem:phinonempty}
For every {$f\in\mfr{F}$}, $\Phi(f)$ is nonempty.
\end{lemma}

\begin{proof}
As described in the discussion of Section \ref{ssec:compactness}, it
suffices to show that the state Markov chain induced by an optimal
oblivious strategy possesses at least one invariant
distribution---i.e., that $\mc{D}(\mu, f)$ is nonempty, where $\mu$ is
an optimal oblivious strategy given $f$.

We first show that for every $f$ and every $\mu \in \mc{P}(f)$, the Markov chain on
$\mc{X}$ induced by $\mu$ and $f$ has at least one closed class.
Let $S = \{ x : \|x\|_\infty \leq \ol{K} + M\}$.  By Lemma \ref{lem:negincr},
if $x \not\in S$, then there
exists some state $x'$ with $\mbf{P}(x' | x, \mu(x), f) > 0$ such that
$x'_\ell \leq x_\ell - \ol{\epsilon}$ for all $\ell$ where $x_\ell
> \ol{K}$.  On the other hand, since increments are bounded, for any
$\ell$ where $x_\ell \leq \ol{K}$, we have $x_\ell' \leq \ol{K} + M$.
Applying this fact inductively,
we find that
for any $x \not\in S$, there must exist a positive probability sequence
of states from $x$ to $S$; i.e., a sequence $y_0, y_1, y_2,
\ldots, y_\tau$ such that $y_0 = x$, $y_\tau \in S$, and $\mbf{P}(y_t |
y_{t-1}, \mu(y_{t-1}), f) >0$ for all $t$.  We say that $S$ is {\em
  reachable} from $x$.

So now suppose the chain induced by $\mu$ and $f$ has no closed class.  Fix $x_0 \in S$.  Since the class containing
$x_0$ is not closed, there must exist a state $x'$ reachable from
$x_0$ with positive probability, such that the chain never returns to
$x_0$ starting from $x'$.  If $x' \in S$, let $x_1 = x'$.  If $x'
\not\in S$, then using the argument in the preceding paragraph, there
must exist a state $x_1 \in S$ reachable from~$x'$.  Arguing
inductively, we can construct a sequence of states $x_0, x_1, x_2,
\ldots$ where $x_t \in S$ for all $t$, and yet $x_0, \ldots, x_{t-1}$
are not reachable from $x_t$.  But $S$ is finite, so at least one
state must repeat in this sequence---contradicting the construction.
We conclude that the chain must have at least one closed class.

To complete the proof, we use a Foster-Lyapunov argument.  Let $U(x) =
\sum_\ell x_\ell^2$.  Then $\{ x \in
\mc{X} : U(x) \leq R\}$ is finite for all $R$.  So now let $\omega =
(2 d \ol{K} M + d M^2 + 1)/(2\ol{\epsilon})$, and suppose
$\|x\|_\infty > \max\{\omega, \ol{K}\}$.  We reason as follows:
\begin{align*}
\sum_{x'} U(x') \mbf{P}(x' | x, \mu(x), f) &=  U(x) + 2 \sum_\ell x_\ell \sum_{z_\ell} z_\ell \mbf{Q}_\ell(z_\ell | x,
\mu(x), f) + \sum_\ell \sum_{z_\ell} z_\ell^2 \mbf{Q}_\ell(z_\ell | x,
\mu(x), f) \\
&\leq U(x) + 2 \sum_{\ell : x_\ell \leq \ol{K}} M x_\ell - 2
\sum_{\ell : x_\ell > \ol{K}} \ol{\epsilon} x_\ell + d M^2  \leq U(x) - 1.
\end{align*}
The first equality follows by definition of $\mbf{Q}$ and $U$,
and multiplicative separability of $\mbf{Q}$. The next step follows
since increments are bounded (Assumption~\ref{as:compactness}), and by applying Lemma~\ref{lem:negincr} for
$x_\ell > \ol{K}$. The last inequality follows from the fact that the state space is
$d$-dimensional, $\|x\|_\infty >
\max\{\ol{K}, \omega\}$, and by definition of $\omega$.
Since increments are bounded, it is
trivial that for every $R$:
\[ \sup_{x : \|x\|_\infty \leq R} \left( \sum_{x'} U(x') \mbf{P}(x'\  |\ x,
  \mu(x), f) - U(x) \right) < \infty. \]
It follows by the Foster-Lyapunov criterion that every closed class of
the Markov chain induced by $\mu$
is positive recurrent, as required \citep{hajek1982hitting, meyn_1993,  glynn2006bounding}. 
\end{proof}


\begin{lemma}
\label{lem:momentbound}
For every $\eta \in \Z_+$, $\sup_f \sup_{\phi \in \Phi(f)} \sum_x \|x\|_\eta^\eta \phi(x) < \infty.$
\end{lemma}

\begin{proof}
We again use a Foster-Lyapunov argument.  We proceed by induction; the
claim is clearly true if $\eta = 0$.  So assume the claim is true up to
$\eta-1$; in particular, define:
\[ \alpha_{k} = \sup_f \sup_{\phi \in \Phi(f)} \sum_x \|x\|_k^k \phi(x) \]
for $k = 0, \ldots, \eta-1$.  Fix $f$, and let $\mu \in \mc{P}(f)$ be an optimal
oblivious strategy given $f$.  The preceding lemma establishes that
the Markov chain induced by $\mu$ possesses at least one invariant
distribution.  Let $U(x) =
\sum_\ell x_\ell^{\eta+1}$.  Then we have:
\begin{align*}
\sum_{x'} U(x') \mbf{P}(x'|\ x, \mu(x), f) &= \sum_\ell \sum_{z_\ell} (x_\ell + z_\ell)^{\eta+1}
\mbf{Q}_\ell(z_\ell\ |\ x, \mu(x), f) \\
&= \sum_\ell \sum_{z_\ell} \sum_{k = 0}^{\eta+1} {{\eta+1} \choose {k}} x_\ell^k z_\ell^{\eta+1-k}
\mbf{Q}_\ell(z_\ell\ |\ x, \mu(x), f) \\
&= U(x) + (\eta+1)\sum_\ell x_\ell^\eta \sum_{z_\ell} z_\ell
\mbf{Q}_\ell(z_\ell\ |\ x, \mu(x), f) \\
& \quad + \sum_\ell \sum_{z_\ell}
\sum_{k = 0}^{\eta-1} {{\eta+1} \choose {k}} x_\ell^k z_\ell^{\eta+1-k}
\mbf{Q}(z\ |\ x, \mu(x), f).
\end{align*}
Define $g(x)$ as:
\[ g(x) = \sum_{k = 0}^{\eta-1} {{\eta+1}\choose k} M^{\eta+1-k}
\sum_\ell x_\ell^k.\]
By the inductive hypothesis,
\[ \gamma \triangleq \sup_f \sup_{\phi \in \Phi(f)} \sum_x g(x) \phi(x) < \infty. \]
Further, by Lemma \ref{lem:negincr}, for all $\ell$ and all $x$ such that
$x_\ell  \geq \ol{K}$, we have:
\[ \sum_{z_\ell} z_\ell \mbf{Q}_\ell(z_\ell\ |\ x, \mu(x), f) < -\ol{\epsilon}. \]
Define $h(x)$ as:
\[ h(x) = -(\eta + 1) M \sum_{\ell : x_\ell \leq \ol{K}} x_\ell^\eta +
\ol{\epsilon} (\eta + 1) \sum_{ \ell : x_\ell > \ol{K}}
x_\ell^\eta. \]
It follows that:
\[ \sum_{x'} U(x') \mbf{P}(x'\ |\ x, \mu(x), f) - U(x) \leq -h(x) +
g(x). \]
Now fix any distribution $\phi \in \mc{D}(\mu, f)$.  Since the Markov
chain induced by $\mu$ and $f$ must be irreducible on the support of
$\phi$, it follows by the
Foster-Lyapunov criterion \citep{meyn_1993} that:
\[ \sum_{x} h(x) \phi(x) \leq \sum_x g(x)\phi(x) \leq \gamma. \]
Rearranging terms, we conclude that:
\[ \sum_{x} \left(\sum_{\ell : x_\ell > \ol{K}} x_\ell^\eta\right)
\phi(x) \leq \frac{\gamma}{\ol{\epsilon}(\eta+1)} + \frac{dM \ol{K}^\eta}{\ol{\epsilon}}. \]
Thus:
\[ \sum_{x} \|x\|_\eta^\eta \phi(x) \leq
\frac{\gamma}{\ol{\epsilon}(\eta+1)} + \left( \frac{dM}{\ol{\epsilon}}+d\right) \ol{K}^\eta. \]
(Recall that the sum is only over $x \in \Z_+^d$.)  Since the right
hand side is finite and independent of $f$ and $\phi$, the result follows. 
\end{proof}

\begin{sproof}{Proof of Proposition \ref{prop:compactness}.}
We have already established that 
$\Phi(f)$ is nonempty {for all $f \in \mfr{F}$} in Lemma 
\ref{lem:phinonempty}.  Define $B = \sup_f \sup_{\phi \in \Phi(f)} \sum_x \|x\|_{p+1}^{p+1} \phi(x) <
\infty$, where the inequality is the result of Lemma \ref{lem:momentbound}.

We define the set $\mfr{C} = \left\{ f \in \mfr{F} : \sum_x \|x\|_{p+1}^{p+1} f(x) \leq B\right\}$.
By the preceding observation, {$\Phi(\mfr{F}) \subset \mfr{C}$}.  It is clear that
$\mfr{C}$ is nonempty and convex.  It remains to be shown that
$\mfr{C}$ is compact in the $\onep$-norm.  It is straightforward to
check that $\mfr{C}$ is complete; we show that $\mfr{C}$ is
totally bounded, thus establishing compactness.

Fix $\epsilon > 0$.  Choose $K_\epsilon$ so that $B/K_\epsilon <
\epsilon$.  Then for all $f \in \mfr{C}$:
\begin{equation}
\label{eq:lighttail}
 \sum_{x : \|x\|_\infty \geq K_\epsilon} \|x\|_p^p f(x) \leq \frac{B}{K_\epsilon} <
\epsilon.
\end{equation}

Let $S_\epsilon = \{ x : \|x\|_\infty < K_\epsilon\}$ and let
$\mfr{S}_C$ be the projection of $\mfr{C}$ onto $S_\epsilon$; i.e.,
\[ \mfr{S}_C = \{ g \in \R^{S_\epsilon} : \exists\ f \in
\mfr{C} \text{ with } g(x) = f(x) \forall\ x \in S_\epsilon \}. \]
It is straightforward to check that $\mfr{S}_C$ is a compact subset
of the finite-dimensional space $\R^{S_\epsilon}$; so let $f_1, \ldots, f_k \in \mfr{S}_C$ be
a $\epsilon$-cover of $\mfr{S}_C$ (i.e., $\mfr{S}_C$ is covered by the
balls around $f_1, \ldots, f_k$ of radius~$\epsilon$ in the
$\onep$-norm).  Then it follows that $f_1, \ldots, f_k$ is a
$2\epsilon$-cover of $\mfr{C}$, since \eqref{eq:lighttail} bounds
the tail of any $f \in \mfr{C}$ by $\epsilon$.  This establishes
that $\mfr{C}$ is totally bounded in the $\onep$-norm, as required.
\end{sproof}

\subsection{Finite Actions}
\label{subsec:finite-actions}

We conclude by briefly discussing how the proof of Proposition
\ref{prop:compactness} may be adapted in the case of finite action
spaces (cf. Definition \ref{def:finiteaction}).  Suppose that $S
\subset \R^q$ is a finite set.  We now show that as long as
Assumption \ref{as:compactness} holds with respect to pure
actions---i.e., with $\mc{A}$ replaced by $S$---Proposition~\ref{prop:compactness} continues to hold.

Lemma \ref{lem:coupling} follows as before, except with $\mc{A}$
replaced by $S$.  Lemma \ref{lem:mulimit} follows the same argument if
we restrict attention to {\em pure} strategies $\mu$, i.e., strategies
that take a pure action in every state.  Let $\hat{\mc{P}}(f)$ denote
the set of optimal pure strategies given $f$.  Then Lemma
\ref{lem:mulimit} then yields that as $\|x\|_\infty \to \infty$:
\[ \sup_f \sup_{\mu \in \hat{\mc{P}}(f)} \rho_\infty(\mu(x), \mc{A}')
\to 0. \]
Since $\mc{A}' \subset S$, it is finite as well.  It follows that there
exists $\zeta$ such that for
$x$ such that $\|x\|_\infty \geq \zeta$, for all $f$, and all $\mu \in
\hat{\mc{P}}(f)$, we have $\mu(x) \in \mc{A}'$.  From this and Assumption~\ref{as:compactness} the
result of Lemma~\ref{lem:negincr} holds for $\mu \in \hat{\mc{P}}(f)$,
i.e., there exists $\epsilon > 0$ such that for all $\ell$
and all $x$ with $x_\ell \geq K'$,
\[ \sup_f \sup_{\mu \in \hat{\mc{P}}(f)} \sum_{z_\ell} z_\ell \mbf{Q}_\ell(z_\ell\ |\ x, \mu(x),
f) < -\epsilon. \]

To complete our proof, we need only note that the set of all optimal
oblivious strategies $\mc{P}(f)$ can be obtained by pointwise convex
combinations of optimal pure oblivious strategies; this follows from
Bellman's equation and the fact that the payoff is linear in the mixed
action.  Thus we also have:
\[ \sup_f \sup_{\mu \in \mc{P}(f)} \sum_{z_\ell} z_\ell \mbf{Q}_\ell(z_\ell\ |\ x, \mu(x),
f) < -\epsilon. \]
The remainder of the proof follows as before.


\section{AME: Proof}

Throughout this section we suppose Assumption \ref{as:continuity} holds. We begin by defining the following sets.

\begin{definition}
\label{def:SS-new}
For every $x \in \mc{X}$, define
\begin{align}
\label{eqn:SS-per-state}
\mc{X}_{x} = \left\{z \in \mc{X} \ \Big | \  \mbf{P}(x \ | \ z, a, f) > 0\ \text{for some}\ a \in \mc{A}\ \text{and for some}\ f \in \mfr{F}_p\right\}.
\end{align}
Also define $\mc{X}_{x, t}$ as
\begin{align}
\label{eqn:SS-per-time}
\mc{X}_{x,t} = \left\{ z \in \mc{X} \ \Big | \  \norm{z}_{\infty} \leq \norm{x}_{\infty} + tM\right\}.
\end{align}
\end{definition}

Thus, $\mc{X}_{x}$ is the set of all initial states that can result in
the final state as~$x$. Since the increments are bounded
(Assumption~\ref{as:continuity}), for every~$x \in \mc{X}$, the
set~$\mc{X}_{x}$ is finite. The set~$\mc{X}_{x, t}$ is a superset of
all possible states that can be reached at time~$t$ starting from
state~$x$ (since the increments are uniformly bounded over action~$a$
and distribution~$f$); note that $\mc{X}_{x,t}$ is finite as well.

{The following key lemma establishes that as the number of
players grows large, the population empirical distribution in a game
with finitely many players approaches the limiting SE population.
The
result is similar in spirit to related results on mean field limits
of interacting particle systems, cf.~\cite{sznitman1991topics}; there the main insight
is that, under appropriate conditions, the stochastic evolution of a
finite-dimensional interacting particle system approaches the
deterministic mean field limit over finite time horizons.  Our model
introduces two sources of complexity.  First, agents' state
transitions are coupled, so the population state Markov process is not simply the
aggregation of independent agent state dynamics.  Second, our state
space is unbounded, so additional care is required to ensure the tail
of the population state distribution is controlled in games with a
large but finite number of players.  This is where the light tail
condition plays a key role.  Our proof proceeds by induction over time
periods.

\begin{lemma}
\label{lem:convergence-of-f}
Let $(\mu, f)$ be a stationary equilibrium {with $f \in \mfr{F}_p$}. Consider an $m$-player
game. Let $x_{i,0}^{(m)} = x_{0}$ and suppose the initial state of
every player (other than player~$i$) is independently sampled from the
distribution~$f$. That is, suppose $x_{j, 0}^{(m)} \sim f$ for all $j
\neq i$; let $f^{(m)}\in \mfr{F}^{(m)}$ denote the initial population
state. Let $a_{i,t}^{(m)}$ be any sequence of (possibly random, possibly history dependent) actions. Suppose players' states evolve as $x_{i, t+1}^{(m)} \sim \mbf{P}\big(\cdot\ | \ x_{i,t}^{(m)}, a_{i,t}^{(m)}, f_{-i,t}^{(m)}\big)$ and for all $j \neq i$, as $x_{j, t+1}^{(m)} \sim \mbf{P}\big(\cdot\ | \ x_{j,t}^{(m)}, \mu(x_{j,t}^{(m)}), f_{-j,t}^{(m)}\big)$.
\ignore{
\begin{align*}
x_{j, t+1}^{(m)} &\sim \mbf{P}\big(\cdot\ | \ x_{j,t}^{(m)}, \mu(x_{j,t}^{(m)}), f_{-j,t}^{(m)}\big) \ \forall \ j = 1, 2, \cdots, m,\ \ j \neq i,\\
x_{i, t+1}^{(m)} &\sim \mbf{P}\big(\cdot\ | \ x_{i,t}^{(m)}, a_{i,t}^{(m)}, f_{-i,t}^{(m)}\big).
\end{align*}
}
Then, for every initial state~$x_{0}$, for all times~$t$, $\norm{\fmi - f}_{\onep} \rightarrow 0$ almost surely \footnote{Note that the convergence is almost surely in the randomness associated with the initial population state.} as $m \rightarrow \infty$.
\end{lemma}

\begin{proof}
Note that $f \in \mfr{F}_p$ and hence $\norm{f}_{\onep} < \infty$.  Thus, given any $\epsilon > 0$, there exists a finite set $\mc{C}_{\epsilon, f}$ such that:
\begin{align}
\label{eqn:cf-tail}
\sum_{x \notin \mc{C}_{\epsilon, f}} \norm{x}_{p}^{p}f(x) < \epsilon.
\end{align}
At~$t = 0$, we have
\[
f_{-i, 0}^{(m)}(x) = \frac{1}{m-1}\sum_{j = 1}^{m-1}\indic{\{X_{j,0} = x\}},
\]
where $X_{j,0}$ are i.i.d random variables distributed according to the distribution~$f$. Define:
\[ Y_j = \norm{X_{j, 0}}_p^p \indic{\{ X_{j, 0} \not\in \mc{C}_{\epsilon, f}\}}. \]
Note that the $Y_j$ are i.i.d. random variables, with:
\[ \mbf{E}[Y_j] = \sum_{x \not\in \mc{C}_{\epsilon, f}} \norm{x}_p^p
f(x). \]
Further, observe that:
\[ \sum_{x \not\in \mc{C}_{\epsilon, f}} \norm{x}_p^p f_{-i, 0}^{(m)}(x) =
\frac{1}{m-1} \sum_{j = 1}^{m-1} Y_j. \]
Thus by the strong law of large numbers, almost surely as $m \to \infty$,
\[ \sum_{x \not\in \mc{C}_{\epsilon, f}} \norm{x}_p^p f_{-i, 0}^{(m)}(x) \to
\sum_{x \not\in \mc{C}_{\epsilon, f}} \norm{x}_p^p f(x) < \epsilon. \]

Now observe that:
\[ \norm{f_{-i, 0}^{(m)}(x) - f}_{1\mhyphen p} \leq \sum_{x \in
  \mc{C}_{\epsilon, f}}\norm{x}_{p}^{p}|f_{-i, 0}^{(m)}(x) - f(x)| +
\sum_{x \notin \mc{C}_{\epsilon, f}}\norm{x}_{p}^{p}f_{-i, 0}^{(m)}(x)  + \sum_{x
  \notin \mc{C}_{\epsilon, f}}\norm{x}_{p}^{p} f(x). \]
Each of the second and third terms on the right hand side is almost
surely less than $\epsilon$ for sufficiently large $m$.  For the first
term, observe that   $|f_{-i, 0}^{(m)}(x) - f(x)| \to 0$ almost surely, again by the strong law
  of large numbers (since $f^{(m)}(x)$ is the sample average of
  $m-1$ Bernoulli random variables with parameter $f(x)$).  Thus
  the first term approaches zero almost surely as $m \to \infty$ by the bounded
  convergence theorem. Since $\epsilon$ was arbitrary, this proves that $\norm{f_{-i,0}^{(m)} - f}_{\onep} \rightarrow 0$ almost surely as $m\rightarrow \infty$.

We now use an induction argument; let us assume that, $\norm{f_{-i,\tau}^{(m)} - f}_{\onep} \rightarrow 0$ almost surely as $m\rightarrow \infty$ for all times $\tau \leq t$. From the definition of $f_{-i, t+1}^{(m)}$ we have:
\[
f_{-i,t+1}^{(m)}(y) = \frac{1}{m-1}\sum_{j\neq i}\indic{\{x_{j,t+1}^{(m)} = y\}},
\]
where $x_{j, t+1}^{(m)} \sim \mbf{P}\big(\cdot \ | \ x_{j,t}^{(m)}, \mu(x_{j,t}^{(m)}), f_{-j, t}^{(m)}\big)$ for all $j \neq i$. Note that if two players have same initial state, then the population state from their viewpoint is identical. That is, if $x_{j,t}^{(m)} = x_{k, t}^{(m)}$, then $f_{-j, t}^{(m)}(y) = f_{-k, t}^{(m)}(y)$ for all $y \in \mc{X}$. We can thus redefine the population state from the viewpoint of a player at a particular state. Let $\hat{f}_{ t}^{(x, m)}$ be the the population state at time~$t$ from the viewpoint of a player at state~$x$. Then, if $x_{j,t}^{(m)} = x_{k, t}^{(m)} = x$, then for all $y \in \mc{X}$, $f_{-j, t}^{(m)}(y) = f_{-k, t}^{(m)}(y) = \hat{f}_{ t}^{(x, m)}(y)$. Without loss of generality, we assume $m > 1$. Let $\eta_{-i, t}^{(m)}(x)$ be the total number of players (excluding player $i$) that have their state at time $t$ as $x$, i.e., $\eta_{-i, t}^{(m)}(x) = (m-1) \fmi(x)$. Note that $\eta_{-i, t}^{(m)}(x) = 0$ if and only $ \fmi(x) = 0$. We can now write $f_{-i,t+1}^{(m)}(y)$ as:
\begin{align}
\label{eqn:nextStepDist}
f_{-i,t+1}^{(m)}(y) &= \frac{1}{m-1}\sum_{x\in \mc{X}}\sum_{j= 1}^{\eta_{-i, t}^{(m)}(x)}\indic{\{Y_{j,x,t}^{(m)} = y\}} \nonumber\\
&=\sum_{x\in \mc{X}}\fmi(x)\left[\frac{1}{\eta_{-i, t}^{(m)}(x)}\sum_{j= 1}^{\eta_{-i, t}^{(m)}(x)}\indic{\{Y_{j,x,t}^{(m)} = y\}}\right] \nonumber \\
& =\sum_{x\in \mc{X}_{y}}\fmi(x)\left[\frac{1}{\eta_{-i, t}^{(m)}(x)}\sum_{j= 1}^{\eta_{-i, t}^{(m)}(x)}\indic{\{Y_{j,x,t}^{(m)} = y\}}\right]
\end{align}
where the last equality follows from the Definition~\ref{def:SS-new}. Here, $Y_{j,x,t}^{(m)}$ are random variables that are independently drawn according to the transition kernel $\mbf{P}(\cdot\ |\  x, \mu(x),\hat{f}_{t}^{(x, m)})$. Note that if $\eta_{-i, t}^{(m)}(x) = 0$, we interpret the term inside the parentheses as zero.
%
%

Let us now look at $\hat{f}_{ t}^{(x, m)}$. We have
\begin{align*}
\hat{f}_{ t}^{(x, m)}(z) =  \fmi(z) + \frac{1}{m-1}\indic{\{x_{i,t}^{(m)} = z\}} - \frac{1}{m-1}\indic{\{z = x\}}.
\end{align*}
Consider $\norm{\hat{f}_{ t}^{(x, m)} - f}_{\onep}$. We have:
\begin{align*}
\norm{\hat{f}_{ t}^{(x, m)} - f}_{\onep} &= \sum_{z \in \mc{X}}\norm{z}_{p}^{p}\left|\hat{f}_{ t}^{(x, m)}(z) - f(z)\right|\\
& = \sum_{z \in \mc{X}}\norm{z}_{p}^{p}\left|\fmi(z) + \frac{1}{m-1}\indic{\{x_{i,t}^{(m)} = z\}} - \frac{1}{m-1}\indic{\{z = x\}} - f(z)\right|\\
&\leq \sum_{z \in \mc{X}}\norm{z}_{p}^{p}\left|\fmi(z) -f(z)\right| + \frac{1}{m-1}\sum_{z \in \mc{X}}\norm{z}_{p}^{p}\indic{\{x_{i,t}^{(m)} = z\}} + \frac{1}{m-1}\sum_{z \in \mc{X}}\norm{z}_{p}^{p}\indic{\{z = x\}}\\
& = \norm{\fmi - f}_{\onep} + \frac{1}{m-1}\sum_{z \in \mc{X}}\norm{z}_{p}^{p}\indic{\{x_{i,t}^{(m)} = z\}} + \frac{1}{m-1}\sum_{z \in \mc{X}}\norm{z}_{p}^{p}\indic{\{z = x\}}
\end{align*}
From the induction hypothesis, we have $\norm{\fmi - f}_{\onep}
\rightarrow 0$ almost surely as $m\rightarrow \infty$. Note that at
time~$t$, $x_{i,t}^{(m)} \in \mc{X}_{x_{0}, t}$ from
equation~\eqref{eqn:SS-per-time}, and $\mc{X}_{x_0,t}$ is finite. Thus,
\[
\sup_{m}\sum_{z \in \mc{X}}\norm{z}_{p}^{p}\indic{\{x_{i,t}^{(m)} = z\}} < \infty \mbox{, almost surely.}
\]
This implies that for all states $x\in \mc{X}$, $\norm{\hat{f}_{ t}^{(x, m)} - f}_{\onep} \rightarrow 0$ almost surely as $m\rightarrow \infty$. From Assumption~\ref{as:continuity}, we know that the transition kernel is continuous in the population state~$f$ (where~$\mfr{F}_p$ is endowed with the $\onep$ norm). Thus for every $x \in \mc{X}$, we have almost surely:
\begin{align}
\label{eqn:conv-kernel}
\mbf{P}(\cdot\ | \ x, \mu(x),\hat{f}_{t}^{(x, m)}) \rightarrow \mbf{P}(\cdot \ |\  x, \mu(x),f),
\end{align}
as $m\rightarrow \infty$.

Next, we show that $f_{-i,t+1}^{(m)}(y) \to f(y)$ almost surely as $m
\to \infty$, for all $y$.  We leverage
equation~\eqref{eqn:nextStepDist}.  Observe that the set of points $x
\in \mc{X}$ where $\|x\|_p \leq 1$ is finite, since $\mc{X}$ is a
subset of an integer lattice. From the induction hypothesis, as $\sum_{x
  \in \mc{X}} \|x\|_p^p | \fmi(x) - f(x)| \to 0$ almost surely as $m
\to \infty$, it follows that $\fmi(x) \rightarrow f(x)$ almost surely
for all $x\in \mc{X}$ as $x\rightarrow \infty$.

Suppose that $x \in \mc{X}_y$ and $f(x) > 0$.  Since $\fmi(x) \to
f(x)$, it follows that $\eta_{-i,t}^{(m)} \to \infty$ as $m \to
\infty$, almost surely.
Note that $Y_{j,x,t}^{(m)}$ are random variables that are independently drawn according to the transition kernel $\mbf{P}(\cdot| x, \mu(x),\hat{f}_{t}^{(x, m)})$. From equation~\eqref{eqn:conv-kernel}, and Lemma~\ref{lem:slln}, we get that for
every $x, y \in \mc{X}$, there holds
\[
\frac{1}{\eta_{-i, t}^{(m)}(x)}\sum_{j= 1}^{\eta_{-i,
    t}^{(m)}(x)}\indic{\{Y_{j,x,t}^{(m)} = y\}} \rightarrow \mbf{P}(y| x, \mu(x),f),
\]
almost surely as $m\rightarrow \infty$.

On the other hand, suppose $x \in \mc{X}_y$ and $f(x) = 0$.  Again,
since $\fmi(x) \to f(x)$ as $x \to \infty$, it follows that as $m \to
\infty$, almost surely:
\[ \fmi(x)\left[\frac{1}{\eta_{-i, t}^{(m)}(x)}\sum_{j= 1}^{\eta_{-i,
      t}^{(m)}(x)}\indic{\{Y_{j,x,t}^{(m)} = y\}}\right] \to 0, \]
since the term in brackets is nonnegative and bounded.  (Recall we
interpret the term in brackets as zero if $\fmi(x) = 0$.)

We conclude that, almost surely, as $m \to \infty$:
\[ f_{-i,t+1}^{(m)}(y) = \sum_{x\in \mc{X}_{y}}\fmi(x)\left[\frac{1}{\eta_{-i,
      t}^{(m)}(x)}\sum_{j= 1}^{\eta_{-i,
      t}^{(m)}(x)}\indic{\{Y_{j,x,t}^{(m)} = y\}}\right] \to \sum_{x
  \in \mc{X}_y} f(x) \mbf{P}(y | x, \mu(x), f) = f(y). \]

To complete the proof, we need to show that $\norm{f_{-i,t+1}^{(m)} -
  f}_{\onep}\rightarrow 0$ almost surely as $m \rightarrow \infty$.
Since $\fmi(x) \to f(x)$ almost surely, for all $\epsilon > 0$ we have:
\[  \sum_{x \in \mc{C}_{\epsilon, f}} \|x\|_p^p \fmi(x) \to \sum_{x \in
    \mc{C}_{\epsilon, f}} \|x\|_p^p f(x).\]
This together with the fact that $\|\fmi - f\|_{\onep} \to 0$ implies
that, almost surely:
\begin{equation}
\label{eqn:fmi_bound}
 \limsup_{m \to \infty} \sum_{x \in \mc{C}_{\epsilon, f}} \|x\|_p^p
  \fmi(x) < \epsilon.
\end{equation}

Now at time~$t+1$, we have
\begin{align}
\sum_{x\notin \mc{C}_{\epsilon, f}}\norm{x}_{p}^{p}f_{-i, t+1}^{(m)} &= \sum_{x\notin \mc{C}_{\epsilon, f}}\sum_{\ell=1}^{d}|x_{\ell}|^{p}f_{-i,t+1}^{(m)}(x)\nonumber \\
&\leq \sum_{x\notin \mc{C}_{\epsilon, f}}\sum_{\ell=1}^{d}\big(|x_{\ell}| + M\big)^{p}f_{-i,t}^{(m)}(x),
\label{eqn:conv-f-temp4}
\end{align}
where the equality follows because $\mc{X}$ is a subset of the $d$-dimensional integer lattice. The last inequality follows from the fact that the increments are bounded (Assumption~\ref{as:continuity}). Without loss of generality, assume that $|x_{\ell}| \geq 1$ and that $M \geq 1$. Then we have:
\begin{align*}
\big(|x_{\ell}| + M\big)^{p} &= \sum_{j=1}^{p}{{p} \choose {j}}|x_{\ell}|^{j} M^{p-j}\\
&\leq \sum_{j=1}^{p}{{p} \choose {j}}|x_{\ell}|^{p} M^{p}\\
& = 2^{p}M^{p} |x_{\ell}|^{p} = K_{1} |x_{\ell}|^{p},
\end{align*}
where we let $K_{1} = (2M)^{p}$. Substituting in
equation~\eqref{eqn:conv-f-temp4}, we have, almost surely,
\begin{align*}
\limsup_{m \to \infty} \sum_{x\notin \mc{C}_{\epsilon, f}}\norm{x}_{p}^{p}f_{-i, t+1}^{(m)} &\leq \sum_{x\notin \mc{C}_{\epsilon, f}}\sum_{\ell=1}^{d}K_{1}|x_{\ell}|^{p}f_{-i,t}^{(m)}(x)\\
& = K_{1}\sum_{x\notin \mc{C}_{\epsilon, f}}\norm{x}_{p}^{p}f_{-i, t}^{(m)}(x)\\
& < K_{1} \epsilon,
\end{align*}
where the last inequality follows from equation~\eqref{eqn:fmi_bound}. Now observe that:
\begin{align*}
\norm{f_{-i,t+1}^{(m)} - f}_{\onep} &\leq \sum_{x \in \mc{C}_{\epsilon, f}}\norm{x}_{p}^{p}|f_{-i, t+1}^{(m)}(x) - f(x)| +
\sum_{x \notin  \mc{C}_{\epsilon, f}}\norm{x}_{p}^{p}f_{-i, t+1}^{(m)}(x)  + \sum_{x
  \notin \mc{C}_{\epsilon, f}}\norm{x}_{p}^{p} f(x).
\end{align*}
In taking a limsup on the left hand side, the second term on the right
hand side is almost surely less than~$K_{1} \epsilon$. From the
definition of $\mc{C}_{\epsilon, f}$ and equation~\eqref{eqn:cf-tail},
we get that the third term on the right hand side is also less
than~$\epsilon$.
Finally, since for every $x$ $|f_{-i,t+1}^{(m)}(x) - f(x)|
\rightarrow 0$ almost surely as $m\rightarrow \infty$, and
$C_{\epsilon, f}$ is finite, the first
term in the above equation approaches zero almost surely as
$m\rightarrow \infty$ by the Bounded Convergence
Theorem.
Since~$\epsilon$ was arbitrary, this proves the induction
step and hence the lemma. 
\end{proof}


The preceding proof uses the following refinement of the strong law of
large numbers.

\begin{lemma}
\label{lem:slln}
Suppose $0 \leq p_k \leq 1$ for all $k$, and that $p_k \to p$ as $k
\to \infty$.  For each $k$, let $Y_1^{(k)}, \ldots, Y_k^{(k)}$ be
i.i.d. Bernoulli random variables with parameter $p_k$.  Then almost
surely:
\[ \lim_{k \to \infty} \frac{1}{k} \sum_{i = 1}^k Y_i^{(k)} = p. \]
\end{lemma}

\begin{proof}
Let $\epsilon > 0$.  By Hoeffding's inequality, we have:
\[ \prob\left( \left|\frac{1}{k} \sum_{i = 1}^k Y_i^{(k)} - p_k\right|
  > \epsilon\right) \leq 2 e^{-k \epsilon_k^2}, \]
since $0 \leq Y_i^{(k)} \leq 1$ for all $i,k$.  Let $\epsilon_k =
1/k$; then by the Borel-Cantelli lemma, the event on the left hand
side in the preceding expression occurs for only finitely many $k$,
almost surely.  In other words, almost surely:
\[ \lim_{k \to \infty} \left[ p_k - \frac{1}{k} \sum_{i = 1}^k Y_i^{(k)}\right] =
0. \]
The result follows. 
\end{proof}


Before we prove the AME property, we need some additional notation. Let $(\mu, f)$ be a stationary equilibrium. Consider again an $m$ player game and focus on player~$i$. Let $x_{i,0}^{(m)} = x_{0}$ and assume that player~$i$ uses a cognizant strategy~$\mu_{m}$. The initial state of every other player $j \neq i$ is independently drawn from the distribution~$f$, that is, $x_{j,0}^{(m)} \sim f$. Denote the initial distribution of all $m-1$ players (excluding player~$i$) by $f^{(m)} \in \mfr{F}^{(m)}$. The state evolution of player~$i$ is given by
\begin{align}
\label{eqn:dynamics-i-1}
x_{i, t+1}^{(m)} \sim \mbf{P}\left(\cdot\ | \ x_{i,t}^{(m)}, a_{i,t}^{(m)}, \fmi\right),
\end{align}
where $a^{(m)}_{i,t} = \mu_{m}\big(x_{i,t}^{(m)}, \fmi\big)$ and $\fmi$ is the actual population distribution. Here the superscript~$m$ on the state variable represents the fact that we are considering an~$m$ player stochastic game. Let every other player~$j$ use the oblivious strategy~$\mu$ and thus their state evolution is given by
\begin{align}
\label{eqn:dynamics-j-1}
x_{j, t+1}^{(m)} \sim \mbf{P}\left(\cdot\ | \ x_{j,t}^{(m)}, \mu\big( x_{j,t}^{(m)}\big), f_{-j,t}^{(m)}\right).
\end{align}
Define $V^{(m)}\big(x, f^{(m)}\ | \ \mu_{m}, \muVec^{(m-1)}\big)$ to be the actual value function of player~$i$, with its initial state~$x$, the initial distribution of the rest of the population as $f^{(m)} \in \mfr{F}^{(m)}$, when the player uses a cognizant strategy~$\mu_{m}$ and every other player uses an oblivious strategy~$\mu$. We have
\begin{multline}
\label{eqn:actual-V-1}
V^{(m)}\big(x, f^{(m)}\ |\ \mu_{m}, \muVec^{(m-1)} \big) = \E\Big[\sum_{t =0}^{\infty}\beta^{t}\pi\big(x_{i,t}, a_{i,t},\fmi\big) \ \big| \
x_{i,0} = x, f_{-i,0}^{(m)} = f^{(m)};\\
 \mu_{i} = \mu_{m}, \muVeci = \muVecmone \Big].
\end{multline}

We define a new player that is coupled to player~$i$ in the $m$ player stochastic games defined above. We call this player the {\em coupled} player. Let $\hat{x}_{i,t}^{(m)}$ be the state of this coupled player at time~$t$. The subscript~$i$ and the superscript~$m$ reflect the fact that this player is coupled to player~$i$ in an~$m$ player stochastic game. We assume that the state evolution of this player is given by:
\begin{align}
\label{eqn:tm-coupled-player}
\xhat_{i,t+1}^{(m)} \sim \mbf{P}\big(\cdot\ | \ \xhat_{i,t}^{(m)}, \hat{a}_{i,t}^{(m)}, f\big),
\end{align}
where $\hat{a}_{i,t}^{(m)} = a^{(m)}_{i,t} =
\mu_{m}\big(x_{i,t}^{(m)}, \fmi\big)$. In other words, this coupled
player takes the same action as player~$i$ at every time~$t$ and this
action depends on the {\em actual} population state of $m-1$
players. However, note that the state evolution is dependent only on
the mean field population state~$f$. Let us define
\begin{align}
\label{eqn:V-coupled-player}
\hat{V}^{(m)}\left(x \ \Big | \ f; \mu_{m}, \muVec^{(m-1)} \right) = \E\left[\sum_{t=0}^{\infty}\beta^{t}\pi\left(\xhat_{i,t}^{(m)}, \hat{a}_{i,t}^{(m)}, f\right)\ | \ \xhat_{i, 0}^{(m)} = x_{0}, \hat{a}_{i,t}^{(m)} = \mu_{m}(x_{i,t}, \fmi); \muVec^{(m-1)}\right].
\end{align}
Thus, $\hat{V}^{(m)}(x \ | \ f; \mu_{m}, \mu)$ is the expected net present value of this coupled player, when the player's initial state is $x$, the long run average population state is~$f$, and the initial population state is $f^{(m)}_{-i,0}=f^{(m)}$. Observe that
\begin{align}
\label{eqn:cd-vhat-ub}
\hat{V}^{(m)}\left(x \ | \ f; \mu_{m}, \muVec^{(m-1)} \right) &\leq \sup_{\mu' \in \mfr{M}}\hat{V}^{(m)}(x \ | \ f; \mu', \muVec^{(m-1)})  = \sup_{\mu' \in \mfr{M}_O}\hat{V}^{(m)}(x \ | \ f; \mu', \muVec^{(m-1)}) \nonumber \\
& = \vTilde^*(x \ | \ f) = \vTilde(x \ | \ \mu, f).
\end{align}
Here, the first equality follows from Lemma~\ref{lem:bellman}, which implies that the supremum over all cognizant strategies is the same as the supremum over oblivious strategies (since the state evolution of other players does not affect the payoff of this coupled player), and the last equality follows since $\mu \in \mc{P}(f)$.


\begin{lemma}
\label{lem:conv-state-dist-1}
Let $(\mu, f)$ be a stationary equilibrium and consider an $m$ player game. Let the initial state of player~$i$ be $x^{(m)}_{i,0} = x$, and let $f^{(m)} \in \mfr{F}^{(m)}$ be the initial population state of $m-1$ players whose initial state is sampled independently from the distribution~$f$. Assume that player~$i$ uses a cognizant strategy $\mu_{m}$ and every other player uses the oblivious strategy~$\mu$. Their state evolutions are given by equation~\eqref{eqn:dynamics-i-1} and~\eqref{eqn:dynamics-j-1}. Also define a coupled player with initial state $\hat{x}_{i,0}^{(m)} = x$ and let its state evolution be given by equation~\eqref{eqn:tm-coupled-player}. Then, for all times~$t$, and for every $y \in \mc{X}$, we have $\left|\prob\big(\hat{x}_{i,t}^{(m)} = y\big) - \prob\big(x_{i,t}^{(m)} = y\big)\right| \rightarrow 0$, almost surely\footnote{{The almost sure convergence of the probabilities is in the randomness associated with the initial population state.}} as $m \rightarrow \infty$.
\end{lemma}
\begin{proof}
The lemma is trivially true for~$t = 0$. Let us assume that it holds for all times~$\tau = 0, 1, \cdots, t-1$. Then, we have
\begin{align*}
\prob\left(x_{i,t}^{(m)} = y\right) &= \sum_{z \in \mc{X}_{y}}\prob\left(x_{i, t-1}^{(m)} = z\right)\mbf{P}\left(y \ \Big | \  z, \mu_{m}(z, f_{-i,t-1}^{(m)}), f_{-i,t-1}^{(m)}\right) \\
\prob\big(\hat{x}_{i,t}^{(m)} = y\big) &= \sum_{z \in \mc{X}_{y}}\prob\left(\hat{x}_{i, t-1}^{(m)} = z\right)\mbf{P}\left(y \ \Big | \  z, \mu_{m}(z, f_{-i,t-1}^{(m)}), f\right).
\end{align*}
Here we use the fact that the coupled player uses the same action as player~$i$ and the state evolution of the coupled player is given by equation~\eqref{eqn:tm-coupled-player}. Note that the summation is over all states in the finite set~$\mc{X}_{y}$, where $\mc{X}_{y}$ is defined as in equation~\eqref{eqn:SS-per-state}.

From Lemma~\ref{lem:convergence-of-f}, we know that for all times~$t$,
$\norm{f_{-i, t}^{(m)} - f}_{\onep} \rightarrow 0$ almost surely as
$m\rightarrow \infty$. From Assumption~\ref{as:continuity}, we know
that the transition kernel is jointly continuous in the action~$a$ and
distribution~$f$ (where the set of distributions~$\mfr{F}_p$ is endowed
with $\onep$ norm).   Since the action set is compact, this implies
that for all $y,z \in \mc{X}$, $
\lim_{m\rightarrow \infty} \sup_{a \in \mc{A}} \left| \mbf{P}\left(y \
    \Big | \  z, a, f_{-i,t-1}^{(m)}\right) - \mbf{P}\left(y \ \Big |
    \  z, a, f\right) \right| = 0.$ almost  surely. It follows that for every $y, z \in \mc{X}$, $\lim_{m\rightarrow \infty} \left |\mbf{P}\left(y \ \Big | \  z, \mu_{m}(z, f_{-i,t-1}^{(m)}), f_{-i,t-1}^{(m)}\right) -
  \mbf{P}\left(y \ \Big | \  z, \mu_{m}(z, f_{-i,t-1}^{(m)}), f\right)
\right| = 0$ almost surely. From  the induction hypothesis, we know that for
every~$z\in \mc{X}$, $\left|\prob\big(\hat{x}_{i,t-1}^{(m)} = z\big) -
  \prob\big(x_{i,t-1}^{(m)} = z\big)\right|\rightarrow 0$ almost surely as $m\rightarrow \infty$. This along with the finiteness of the set $\mc{X}_{y}$, gives that for every $y \in \mc{X}$ $\left|\prob\big(\hat{x}_{i,t}^{(m)} = y\big) - \prob\big(x_{i,t}^{(m)} = y\big)\right| \rightarrow 0$ almost surely as $m\rightarrow \infty$. This proves the lemma. 
\end{proof}

%
%

\begin{lemma}
\label{lemma:payoff-conv-1}
Let $(\mu, f)$ be a stationary equilibrium and consider an $m$ player game. Let the initial state of player~$i$ be $x^{(m)}_{i,0} = x$, and let $f^{(m)} \in \mfr{F}^{(m)}$ be the initial population state of $m-1$ players whose initial state is sampled independently from the distribution~$f$. Assume that player~$i$ uses a cognizant strategy $\mu_{m}$ and every other player uses the oblivious strategy~$\mu$. Their state evolutions are given by equation~\eqref{eqn:dynamics-i-1} and~\eqref{eqn:dynamics-j-1}. Also define a coupled player with initial state $\hat{x}_{i,0}^{(m)} = x$ and let its state evolution be given by equation~\eqref{eqn:tm-coupled-player}. Then, for all times~$t$, we have $\limsup_{m \rightarrow \infty}\E\left[\pi\left(x_{i,t}^{(m)}, \mu_m\big(x_{i,t}^{(m)}, \fmi\big), \fmi\right) - \pi\left(\hat{x}_{i,t}^{(m)}, \mu_m\big(x_{i,t}^{(m)}, \fmi\big), f\right)\right] \leq 0,$ almost surely\footnote{{The almost sure convergence of the expected value of the payoff is in the randomness associated with the initial population state.}}
\end{lemma}
\begin{proof}
Let us write $a_{i,t}^{(m)} =  \mu_m\big(x_{i,t}^{(m)}, \fmi\big)$. We have
\begin{align*}
\Delta^{(m)}_{i,t} &= \E\left[\pi\left(x_{i,t}^{(m)}, a_{i,t}^{(m)},
    \fmi\right) - \pi\left(\hat{x}_{i,t}^{(m)},   a_{i,t}^{(m)}, f\right)\right] \\
&= \E\left[\pi\left(x_{i,t}^{(m)}, a_{i,t}^{(m)}, \fmi\right) -
  \pi\left(x_{i,t}^{(m)},  a_{i,t}^{(m)}, f\right)\right] +
  \E\left[\pi\left(x_{i,t}^{(m)}, a_{i,t}^{(m)}, f\right) - \pi\left(\hat{x}_{i,t}^{(m)}, a_{i,t}^{(m)}, f\right)\right] \\
&\triangleq T^{(m)}_{1,t} + T^{(m)}_{2,t}.
\end{align*}

Consider the first term. We have

\begin{align*}
T^{(m)}_{1,t} &\leq \sum_{y \in \mc{X}}\prob\big(x_{i,t}^{(m)} = y\big) \sup_{a\in \mc{A}}\left|\pi\left(y, a, \fmi \right) - \pi\left(y, a, f \right)\right|\\
& = \sum_{y \in \mc{X}_{x,t}}\prob\big(x_{i,t}^{(m)} = y\big) \sup_{a\in \mc{A}}\left|\pi\left(y, a, \fmi \right) - \pi\left(y, a, f \right)\right|,
\end{align*}

\ignore{
\begin{align*}
T^{(m)}_{1,t} &= \sum_{y \in \mc{X}}\prob\big(x_{i,t}^{(m)} = y\big) \left(\pi\left(y, a_{i,t}^{(m)}, \fmi \right) - \pi\left(y, a_{i,t}^{(m)}, f \right)\right)\\
&\leq \sum_{y \in \mc{X}}\prob\big(x_{i,t}^{(m)} = y\big) \sup_{a\in \mc{A}}\left|\pi\left(y, a, \fmi \right) - \pi\left(y, a, f \right)\right|\\
& = \sum_{y \in \mc{X}_{x,t}}\prob\big(x_{i,t}^{(m)} = y\big) \sup_{a\in \mc{A}}\left|\pi\left(y, a, \fmi \right) - \pi\left(y, a, f \right)\right|,
\end{align*}
}
where the last equality follows from the fact that $x_{i,0}^{(m)} = x$ and from equation~\eqref{eqn:SS-per-time}. From Assumption~\ref{as:continuity}, we know that the payoff is jointly continuous in action~$a$ and distribution~$f$ (with the set of distributions~$\mfr{F}_p$ endowed with $\onep$ norm) and the set~$\mc{A}$ is compact. Thus, for every $y \in \mc{X}$, we have $\sup_{a\in \mc{A}}\left|\pi\left(y, a, \fmi \right) - \pi\left(y, a, f \right)\right| \rightarrow 0,$ almost surely as $m\rightarrow \infty$. This along with the fact that $\mc{X}_{x,t}$ is finite shows that $\limsup_{m\rightarrow \infty}T^{(m)}_{1,t} \leq 0$ almost surely.

%

Now consider the second term. We have

\begin{align*}
T^{(m)}_{2,t} & = \E\left[\pi\left(x_{i,t}^{(m)}, a_{i,t}^{(m)}, f\right) - \left(\hat{x}_{i,t}^{(m)}, a_{i,t}^{(m)}, f\right)\right]\\
&\leq \sum_{y \in \mc{X}}\left|\prob\big(x_{i,t}^{(m)} = y\big) - \prob\big(\hat{x}_{i,t}^{(m)} = y\big)\right|\sup_{a \in \mc{A}}\left|\pi\big(y, a, f)\right|\\
& = \sum_{y \in \mc{X}_{x,t}}\left|\prob\big(x_{i,t}^{(m)} = y\big) - \prob\big(\hat{x}_{i,t}^{(m)} = y\big)\right|\sup_{a \in \mc{A}}\left|\pi\big(y, a, f)\right|,
\end{align*}
where the last equality follows from the fact that $x_{i,0}^{(m)} = \hat{x}_{i,0}^{(m)} = x$ and from Definition~\ref{def:SS-new}. From Lemma~\ref{lem:conv-state-dist-1}, we know that for every~$y \in \mc{X}$, $\left|\prob\big(x_{i,t}^{(m)} = y\big) - \prob\big(\hat{x}_{i,t}^{(m)} = y\big)\right| \rightarrow 0$ almost surely  $m\rightarrow \infty$. Since $\mc{X}_{x,t}$ is finite for every fixed $x\in \mc{X}$ and every time~$t$, this implies that $\limsup_{m \rightarrow \infty}T^{(m)}_{2,t} \leq 0$ almost surely. This proves the lemma. 
\end{proof}

Before we proceed further, we need one additional piece of notation. Once again let $(\mu, f)$ be a stationary equilibrium and consider an oblivious player. Let $\tilde{x}_{t}$ be the state of this oblivious player at time~$t$. We assume that $\tilde{x}_{0} = x$ and since the player used the oblivious strategy $\mu$, the state evolution of this player is given by
\begin{align}
\label{eqn:dynamics-oblivious-player}
\tilde{x}_{t+1} \sim \mbf{P}\big(\cdot \ | \ \tilde{x}_{t}, \tilde{a}_{t}, f\big)
\end{align}
where $\tilde{a}_{t} = \mu(\tilde{x}_{t})$. We let $\vTilde\big(x\ |\ \mu, f\big)$ (as defined in equation~\eqref{eqn:oe-value-func}) to be the oblivious value function for this player starting from state~$x$.

Also, consider an $m$ player game and focus on player~$i$. We represent the state of player~$i$ at time~$t$ by $\check{x}_{i,t}^{(m)}$. As before, the superscript $m$ on the state variable represents the fact that we are considering an~$m$ player stochastic game. Let $\check{x}_{i,0}^{(m)} = x$ and let player~$i$ also use the oblivious strategy~$\mu$.  The initial state of every other player $j \neq i$ is drawn independently from the distribution~$f$, that is, $\check{x}_{j,0}^{(m)} \sim f$. Denote the initial distribution of all $m-1$ players (excluding player~$i$) by $f^{(m)} \in \mfr{F}^{(m)}$. The state evolution of player~$i$ is then given by
\begin{align}
\label{eqn:dynamics-i-2}
\check{x}_{i, t+1}^{(m)} \sim \mbf{P}\left(\cdot\ | \ \check{x}_{i,t}^{(m)}, \check{a}_{i,t}^{(m)}, \fmi\right),
\end{align}
where $\check{a}^{(m)}_{i,t} = \mu\big(\check{x}_{i,t}^{(m)}\big)$. Note that even though the player uses an oblivious strategy, its state evolution is affected by the {\em actual} population state. Let every other player~$j$ also use the oblivious strategy~$\mu$ and let their state evolution be given by
\begin{align}
\label{eqn:dynamics-j-2}
\check{x}_{j, t+1}^{(m)} \sim \mbf{P}\left(\cdot\ \Big | \ \check{x}_{j,t}^{(m)}, \mu\big( \check{x}_{j,t}^{(m)}\big), f_{-j,t}^{(m)}\right).
\end{align}
Define $V^{(m)}\big(x, f^{(m)}\ | \ \muVec^{(m)}\big)$ to be the actual value function of the player, when the initial state of the player is~$x$, the initial population distribution is~$f^{(m)}$ and every player uses the oblivious strategy~$\mu$. That is,

\begin{align}
\label{eqn:actual-V-2}
V^{(m)}\big(x, f^{(m)}\ |\  \muVec^{(m)} \big) = \E\left[\sum_{t =0}^{\infty}\beta^{t}\pi\big(\check{x}_{i,t}, \check{a}_{i,t},\fmi\big) \ \big| \
\check{x}_{i,0} = x, f_{-i,0}^{(m)} = f^{(m)}; \mu_{i} = \mu, \muVeci = \muVec^{(m)} \right].
\end{align}



\begin{lemma}
\label{lem:conv-state-dist-2}
Let $(\mu, f)$ be a stationary equilibrium and consider an $m$ player stochastic game. Let $\check{x}_{i,0}^{(m)} = x$, and let $f^{(m)} \in \mfr{F}^{(m)}$ be the initial population state of $m-1$ players whose initial state is sampled independently from~$f$. Assume that every player uses the oblivious strategy~$\mu$ and their state evolutions are given by equations~\eqref{eqn:dynamics-i-2} and~\eqref{eqn:dynamics-j-2}. Also, consider an oblivious player with $\tilde{x}_{0} = x$ and let its state evolution be given by equation~\eqref{eqn:dynamics-oblivious-player}. Then, for every time~$t$ and for all $y\in \mc{X}$, we have $\left|\prob(\tilde{x}_{t} = y) - \prob(\check{x}_{i,t}^{(m)} = y)\right| \rightarrow 0,$ almost surely as $m\rightarrow \infty$.
\end{lemma}
\begin{proof}
The lemma is trivially true for~$t = 0$. Let us assume that it holds for all times~$\tau = 0, 1, \cdots, t-1$. Then, we have
\begin{align*}
\prob\left(\tilde{x}_{t} = y\right) &= \sum_{z \in \mc{X}_{y}}\prob\left(\tilde{x}_{t-1} = z\right)\mbf{P}\left(y \ \Big | \  z, \mu(z), f\right) \\
\prob\big(\check{x}_{i,t}^{(m)} = y\big) &= \sum_{z \in \mc{X}_{y}}\prob\left(\check{x}_{i, t-1}^{(m)} = z\right)\mbf{P}\left(y \ \Big | \  z, \mu(z), \fmi\right).
\end{align*}
Note that the summation above is over all states in a finite set~$\mc{X}_{y}$ (as defined in Definition~\ref{def:SS-new}).

From Lemma~\ref{lem:convergence-of-f}, we know that for all times~$t$, $\norm{f_{-i, t}^{(m)} - f}_{\onep} \rightarrow 0$ almost surely as $m\rightarrow \infty$. From Assumption~\ref{as:continuity}, we know that the transition kernel is continuous in the distribution (where the set of distributions~$\mfr{F}_p$ is endowed with $\onep$ norm). From the induction hypothesis, we know that $\left|\prob\big(\tilde{x}_{t-1} = z\big) - \prob\big(\check{x}_{-i,t-1}^{(m)} = z\big)\right| \rightarrow 0$. This along with the finiteness of the set $\mc{X}_{y}$, gives that for every $x \in \mc{X}$
\[
\left|\prob\left(\tilde{x}_{t} = x\right) - \prob\big(\check{x}_{i,t}^{(m)} = x\big)\right| \rightarrow 0
\]
almost surely as $m\rightarrow \infty$. This proves the lemma. 
\end{proof}

\begin{lemma}
\label{lemma:payoff-conv-2}
Let $(\mu, f)$ be a stationary equilibrium and consider an $m$ player stochastic game. Let $\check{x}_{i,0}^{(m)} = x$, and let $f^{(m)} \in \mfr{F}^{(m)}$ be the initial population state of $m-1$ players whose initial state is sampled independently from~$f$. Assume that every player uses the oblivious strategy~$\mu$ and their state evolutions are given by equations~\eqref{eqn:dynamics-i-2} and~\eqref{eqn:dynamics-j-2}. Also, consider an oblivious player with $\tilde{x}_{0} = x$ and let its state evolution be given by equation~\eqref{eqn:dynamics-oblivious-player}. Then for all times~$t$, we have $\E\left[\pi\big(\tilde{x}_{t}, \mu(\tilde{x}_{t}), f\big) - \pi\big(\check{x}_{i,t}^{(m)}, \mu(\check{x}_{i,t}^{(m)}), \fmi\big)\right] \rightarrow 0,$ almost surely as $m\rightarrow \infty$.
\end{lemma}
\begin{proof}
Define $\Delta_{i,t}^{(m)}$ as
\begin{align*}
\Delta_{i,t}^{(m)} &= \E\left[\pi\big(\tilde{x}_{t}, \mu(\tilde{x}_{t}), f\big) - \pi\big(\check{x}_{i,t}^{(m)}, \mu(\check{x}_{i,t}^{(m)}), \fmi\big)\right] \\
& = \E\left[\pi\big(\tilde{x}_{t}, \mu(\tilde{x}_{t}), f\big) - \pi\big(\tilde{x}_{t}, \mu(\tilde{x}_{t}), \fmi)\right] + \E\left[\pi\big(\tilde{x}_{t}, \mu(\tilde{x}_{t}), \fmi) -  \pi\big(\check{x}_{i,t}^{(m)}, \mu(\check{x}_{i,t}^{(m)}), \fmi\big)\right] \\
&\triangleq T^{(m)}_{1,t} + T^{(m)}_{2,t}.
\end{align*}

Note that from Lemma~\ref{lem:convergence-of-f}, we have that $\norm{\fmi - f}_{\onep} \rightarrow 0$ almost surely as $m \rightarrow \infty$. From Assumption~\ref{as:continuity}, we know that the payoff is continuous in the distribution, where the set of distributions~$\mfr{F}_p$ is endowed with $\onep$ norm. Thus, for every~$y$ and $a$, we have
\begin{equation}
\label{eq:pidifflimit}
\left|\pi(y, a, f) - \pi(y, a, \fmi)\right| \rightarrow 0,
\end{equation}
as $m \rightarrow \infty$. Consider the first term. We have:
\begin{align*}
T^{(m)}_{1,t} & = \sum_{y\in\mc{X}}\prob\left(\tilde{x}_{t} = y\right)\left|\pi(y, \mu(y), f) - \pi(y, \mu(y), \fmi)\right|\\
& = \sum_{y\in\mc{X}_{x,t}}\prob\left(\tilde{x}_{t} = y\right)\left|\pi(y, \mu(y), f) - \pi(y, \mu(y), \fmi)\right|,
\end{align*}
where the last equality follows from the fact that $\tilde{x}_{0} = x$ and from Definition~\ref{def:SS-new}. Since $\mc{X}_{x, t}$ is a finite set for every initial state~$x \in \mc{X}$ and every time ~$t$, we get that $T^{(m)}_{1,t} \rightarrow 0$ almost surely as $m\rightarrow \infty$.

Consider now the second term. We have:
\begin{align*}
T^{(m)}_{2,t} &= \E\left[\pi\big(\tilde{x}_{t}, \mu(\tilde{x}_{t}), \fmi) -  \pi\big(\check{x}_{i,t}^{(m)}, \mu(\check{x}_{i,t}^{(m)}), \fmi\big)\right]\\
& = \sum_{y\in\mc{X}}\prob\big(\tilde{x}_{t} = y\big)\pi\big(y, \mu(y), \fmi) - \sum_{y\in\mc{X}}\prob\big(\check{x}_{i,t}^{(m)} = y\big)\pi\big(y, \mu(y), \fmi\big) \\
& = \sum_{y\in\mc{X}_{t}}\left(\prob\big(\tilde{x}_{t} = y\big) - \prob\big(\check{x}_{i,t}^{(m)} = y\big)\right)\pi\big(y, \mu(y), \fmi).
\end{align*}
From Lemma~\ref{lem:conv-state-dist-2},
equation~\eqref{eq:pidifflimit}, and the finiteness of $\mc{X}_{x,
  t}$, we get that $\limsup_{m\rightarrow \infty} T^{(m)}_{2,t} \leq
0$ almost surely. This proves the lemma. 
\end{proof}

\begin{sproof} {{Proof of Theorem \ref{th:AME}.}}
Let us define
\[
\Delta V^{(m)}(x, f^{(m)}) \triangleq V^{(m)}\big(x,
f^{(m)}\ |\ \mu_{m}, \muVec^{(m-1)}\big)
- V^{(m)}\big(x, f^{(m)}\ |\  \muVec^{(m)}\big).
\]
Then we need to show that for all~$x$,~$\limsup_{m\rightarrow \infty}
 \Delta V^{(m)}(x, f^{(m)}) \leq 0$ almost surely. We can write

\begin{align*}
&\Delta V^{(m)}(x, f^{(m)})  = V^{(m)}\big(x,f^{(m)}\ |\ \mu_{m}, \muVec^{(m-1)}\big) - \vTilde(x \ |\ \mu, f) + \vTilde(x \ |\ \mu, f)  - V^{(m)}\big(x, f^{(m)}\ |\ \muVec^{(m)}\big) \\
& \quad \quad \leq V^{(m)}\big(x,f^{(m)}\ |\ \mu_{m}, \muVec^{(m-1)}\big) - \hat{V}^{(m)}\big(x \ | \ f; \mu_{m}, \muVec^{(m-1)}\big) + \vTilde(x \ |\ \mu, f)  - V^{(m)}\big(x, f^{(m)}\ |\ \muVec^{(m)}\big) \\
&\quad \quad \triangleq T^{(m)}_{1} + T^{(m)}_{2}.
\end{align*}
Here the inequality follows from equation~\eqref{eqn:cd-vhat-ub}. Consider the term $T^{(m)}_{1}$. We have
\begin{align*}
T^{(m)}_{1} &= V^{(m)}\big(x,f^{(m)}\ |\ \mu_{m}, \muVec^{(m-1)}\big) - \hat{V}^{(m)}\big(x \ | \ f; \mu_{m}, \muVec^{(m-1)}\big) \\
& = \E\left[\sum_{t=0}^{\infty}\beta^{t}\left(\pi\big(x_{i,t}^{(m)}, a_{i,t}^{(m)}, \fmi\big) - \pi\big(\hat{x}_{i,t}^{(m)}, \hat{a}_{i,t}^{(m)}, f\big)\right)\right],
\end{align*}
where the last equality follows from equations~\eqref{eqn:actual-V-1} and~\eqref{eqn:V-coupled-player}. Note that $x_{i,0} = \xhat_{i,0} = x$ and $a_{i,t} = \hat{a}_{i,t} = \mu_{m}\big(x_{i,t}, \fmi\big)$ and the state transitions of players are given by equations~\eqref{eqn:dynamics-i-1}, \eqref{eqn:dynamics-j-1}, and \eqref{eqn:tm-coupled-player}. From Lemma~\ref{lemma:payoff-conv-1}, we have $\limsup_{m\rightarrow \infty} \E\left[\sum_{t = 0}^{T-1}\beta^{t}\left(\pi\big(x_{i,t}^{(m)}, a_{i,t}^{(m)}, \fmi\big) - \pi\big(\hat{x}_{i,t}^{(m)}, \hat{a}_{i,t}^{(m)}, f\big)\right)\right] \leq 0,$ almost surely for any finite time~$T$. From Lemma~\ref{lem:supbound}, we have, almost surely 
\[
\E\left[\sum_{t = T}^{\infty}\beta^{t}\left(\pi\big(x_{i,t}^{(m)}, \hat{a}_{i,t}^{(m)}, \fmi\big) - \pi\big(\hat{x}_{i,t}^{(m)},  a_{i,t}^{(m)}, f\big)\right)\right] \leq 2 C(x, T),
\]
 which goes to zero as $T \rightarrow \infty$. This proves that
$\limsup_{m\rightarrow \infty} T^{(m)}_{1} \leq 0$ almost
surely. Similar analysis (with an application of Lemma
\ref{lemma:payoff-conv-2}) shows that $\limsup_{m\rightarrow \infty}
T^{(m)}_{2} \leq 0$ almost surely, yielding the result. 
\end{sproof}

\begin{sproof}{Proof of Theorem \ref{th:AME-S}.}
Similar to the proof of Theorem \ref{th:AME}, let us define
\[
\Delta V^{(m)}(x, f^{(m)}) \triangleq V^{(m)}\big(x,
f^{(m)}\ |\ \mu_{m}, \muVec^{(m-1)}\big)
- V^{(m)}\big(x, f^{(m)}\ |\  \muVec^{(m)}\big).
\]
Then we need to show that for all~$x$,~$\limsup_{m\rightarrow \infty}
 \Delta V^{(m)}(x, f^{(m)}) \leq 0$ almost surely. We can write

\begin{align*}
\Delta V^{(m)}(x, f^{(m)}) & = V^{(m)}\big(x,f^{(m)}\ |\ \mu_{m}, \muVec^{(m-1)}\big) - \vTilde(x \ |\ \mu, f) + \vTilde(x \ |\ \mu, f)  - V^{(m)}\big(x, f^{(m)}\ |\ \muVec^{(m)}\big) \\
& \leq V^{(m)}\big(x,f^{(m)}\ |\ \mu_{m}, \muVec^{(m-1)}\big) - \hat{V}^{(m)}\big(x \ | \ f; \mu_{m}, \muVec^{(m-1)}\big) + \vTilde(x \ |\ \mu, f)  - V^{(m)}\big(x, f^{(m)}\ |\ \muVec^{(m)}\big) \\
&\triangleq T^{(m)}_{1} + T^{(m)}_{2},
\end{align*}
where $\hat{V}^{(m)}$ is defined as in \eqref{eqn:V-coupled-player}
(and in particular, using the limit profit function $\pi$).
Here the inequality follows from equation~\eqref{eqn:cd-vhat-ub}. Consider the term $T^{(m)}_{1}$. We have
\begin{align*}
T^{(m)}_{1} &= V^{(m)}\big(x,f^{(m)}\ |\ \mu_{m}, \muVec^{(m-1)}\big) - \hat{V}^{(m)}\big(x \ | \ f; \mu_{m}, \muVec^{(m-1)}\big) \\
& = \E\left[\sum_{t=0}^{\infty}\beta^{t}\left(\pi_m\big(x_{i,t}^{(m)}, a_{i,t}^{(m)}, \fmi\big) - \pi\big(\hat{x}_{i,t}^{(m)}, \hat{a}_{i,t}^{(m)}, f\big)\right)\right],
\end{align*}
where the last equality follows from equation \eqref{eqn:actual-V-1} (with $\pi$ replaced by $\pi_m$) and equation \eqref{eqn:V-coupled-player}. Note that $x_{i,0} = \xhat_{i,0} = x$ and $a_{i,t} = \hat{a}_{i,t} = \mu_{m}\big(x_{i,t}, \fmi\big)$ and the state transitions of players are given by equations~\eqref{eqn:dynamics-i-1}, \eqref{eqn:dynamics-j-1}, and \eqref{eqn:tm-coupled-player}. Now,
\begin{align*}
T^{(m)}_{1}
& = \E\left[\sum_{t=0}^{\infty}\beta^{t}\left(\pi_m\big(x_{i,t}^{(m)}, a_{i,t}^{(m)}, \fmi\big) - \pi_m\big(\hat x_{i,t}^{(m)}, \hat a_{i,t}^{(m)}, f\big) \right)+ \left(\pi_m\big(\hat x_{i,t}^{(m)}, \hat a_{i,t}^{(m)}, f\big)- \pi\big(\hat{x}_{i,t}^{(m)}, \hat{a}_{i,t}^{(m)}, f\big)\right)\right],
\end{align*}
Using equicontinuity and the uniform growth rate bound, a similar
argument to the proof of Theorem \ref{th:AME} (via Lemmas \ref{lemma:payoff-conv-1} and
\ref{lemma:payoff-conv-2}) shows that:
\begin{align*}
\limsup_{m\to\infty}~ \E\left[\sum_{t=0}^{\infty}\beta^{t}\left(\pi_m\big(x_{i,t}^{(m)}, a_{i,t}^{(m)}, \fmi\big) - \pi_m\big(\hat x_{i,t}^{(m)}, \hat a_{i,t}^{(m)}, f\big) \right)\right]\leq 0,
\end{align*}
almost surely. Recall that, for all $x,a,f$, $\lim_{m\to\infty}
\pi_m(x,a,f)=\pi(x,a,f)$.  Since $\mc{A}$ is compact and increments
are bounded, it follows that $\pi_m(\hat x_{i,t}^{(m)}, \hat
a_{i,t}^{(m)}, f)- \pi(\hat{x}_{i,t}^{(m)},
\hat{a}_{i,t}^{(m)}, f) \to 0$ almost surely as $m \to \infty$, for
all times $t$.  Using the fact that increments are bounded, the
uniform growth rate bound, and the dominated convergence theorem, the
expectation of the preceding difference also
approaches zero almost surely.  Finally, by truncating the sum at time
$T$, an argument similar
to the proof of Theorem \ref{th:AME} gives:
\begin{align*} \limsup_{m\to\infty}~\E\left[\sum_{t=0}^{\infty}\beta^{t}  \left(\pi_m\big(\hat x_{i,t}^{(m)}, \hat a_{i,t}^{(m)}, f\big)- \pi\big(\hat{x}_{i,t}^{(m)}, \hat{a}_{i,t}^{(m)}, f\big)\right)\right] \leq 0. \end{align*}
This proves that
$\limsup_{m\rightarrow \infty} T^{(m)}_{1} \leq 0$ almost
surely. Similar analysis  shows that $\limsup_{m\rightarrow \infty}
T^{(m)}_{2} \leq 0$ almost surely, yielding the result. 
\end{sproof}

%


\end{document}